\RequirePackage[save]{silence}
\PassOptionsToPackage{hyphens}{url}
\RequirePackage{mathtools}

\documentclass[runningheads]{llncs}

\newif\iffull
\fulltrue
\usepackage{style}

\begin{document}

\mathsetup

\title{RustHorn: CHC-based Verification for Rust Programs\shortorfull{\thanks{The full version of this paper is available as \cite{FullPaper}.}}{ (full version)\thanks{This paper is the full version of \cite{ShortPaper}.}}}

\author{
  Yusuke Matsushita\inst{1}\orcid{0000-0002-5208-3106} \and
  Takeshi Tsukada\inst{1}\orcid{0000-0002-2824-8708} \and
  Naoki Kobayashi\inst{1}\orcid{0000-0002-0537-0604}
}
\institute{
  The University of Tokyo, Tokyo, Japan
  \email{\{yskm24t,tsukada,koba\}@is.s.u-tokyo.ac.jp}
}
\authorrunning{Y. Matsushita et al.}

\maketitle

\begin{abstract}
Reduction to the satisfiability problem for constrained Horn clauses (CHCs) is a widely studied approach to automated program verification.
The current CHC-based methods for pointer-manipulating programs, however, are not very scalable.
This paper proposes a novel translation of pointer-manipulating Rust programs into CHCs, which clears away pointers and memories by leveraging ownership.
We formalize the translation for a simplified core of Rust and prove its correctness.
We have implemented a prototype verifier for a subset of Rust and confirmed the effectiveness of our method.
\end{abstract}

\section{Introduction}
\label{index:intro}

Reduction to \emph{constrained Horn clauses (CHCs)} is a widely studied approach to automated program verification \cite{ProofRules,HornVerify}.
A CHC is a Horn clause \cite{HornClause} equipped with constraints,
namely a formula of the form \(\varphi \!\impliedby\! \psi_0 \!\land\! \cdots \!\land\! \psi_{k-1}\),
where \(\varphi\) and \(\psi_0,\dots,\psi_{k-1}\) are either an atomic formula of the form \(f(t_0,\dots,t_{n-1})\) (\(f\) is a \emph{predicate variable} and \(t_0,\dots,t_{n-1}\) are terms),
or a constraint (e.g. \(a < b + 1\)).\footnote{%
  Free variables are universally quantified.
  Terms and variables are governed under sorts (e.g. \(\Int, \bool\)), which are made explicit in the formalization of \cref{index:chc}.
}
We call a finite set of CHCs a \emph{CHC system} or sometimes just CHC.
\emph{CHC solving} is an act of deciding whether a given CHC system \(S\) has a \emph{model}, i.e.~a valuation for predicate variables that makes all the CHCs in \(S\) valid.
A variety of program verification problems can be naturally reduced to CHC solving.

For example, let us consider the following C code that defines McCarthy's 91 function.
\begin{cpp}
int mc91(int n) {
  if (n > 100) return n - 10; else return mc91(mc91(n + 11));
}
\end{cpp}
Suppose that we wish to prove \cppi{mc91}(\(n\)) returns \(91\) whenever \(n \le 101\) (if it terminates).
The wished property is equivalent to the satisfiability of the following CHCs, where \(\Mcn(n,r)\) means that \(\texttt{mc91}(n)\) returns \(r\) if it terminates.
\begin{align*}
  & \Mcn(n, r) \impliedby n > 100 \land r = n - 10 \br[-.2em]
  & \Mcn(n, r) \impliedby n \le 100 \land \Mcn(n + 11, r') \land \Mcn(r', r) \br[-.2em]
  & r = 91 \impliedby n \le 101 \land \Mcn(n, r)
\end{align*}
The property can be verified because this CHC system has a model:
\begin{gather*}
  \Mcn(n, r) \allowbreak \defiff r = 91 \lor (n > 100 \land r = n - 10).
\end{gather*}

A CHC solver provides a common infrastructure for a variety of programming languages and properties to be verified.
There have been effective CHC solvers \cite{Spacer,FreqHorn,Eldarica,HoIce} that can solve instances obtained from actual programs\footnote{%
  For example, the above CHC system on \(\Mcn\) can be solved instantly by many CHC solvers including Spacer \cite{Spacer} and HoIce \cite{HoIce}.
}
and many program verification tools \cite{SeaHorn,JayHorn,SMTParameterizedSystems,ThreadModularity,ZEUS,HOMC-RTI} use a CHC solver as a backend.

However, the current CHC-based methods do not scale very well for programs using \emph{pointers}, as we see in \cref{index:intro-challenges}.
We propose a novel method to tackle this problem for pointer-manipulating programs under \emph{Rust-style ownership}, as we explain in \cref{index:intro-ours}.

\subsection{Challenges in Verifying Pointer-Manipulating Programs}
\label{index:intro-challenges}

The standard CHC-based approach \cite{SeaHorn} for pointer-manipulating programs represents the memory state as an \emph{array}, which is passed around as an argument of each predicate (cf. the \emph{store-passing style}), and a pointer as an index.

For example, a pointer-manipulating variation of the previous program
\begin{cpp}
void mc91p(int n, int* r) {
  if (n > 100) *r = n - 10;
  else { int s; mc91p(n + 11, &s); mc91p(s, r); }
}
\end{cpp}
is translated into the following CHCs by the array-based approach:\footnote{%
  \(\updarray{h}{r}{v}\) is the array made from \(h\) by replacing the value at index \(r\) with \(v\).
  \(\select{h}{r}\) is the value of array \(h\) at index \(r\).
}
\begin{align*}
  & \Mcnp(n, r, h, h') \impliedby n > 100 \land h' = \updarray{h}{r}{n - 10} \br[-.2em]
  & \begin{aligned}
      \Mcnp(n, r, h, h') \impliedby & n \le 100 \land \Mcnp(n + 11, \var{ms}, h, h'') \\[-.3em]
      & \land \Mcnp(\select{h''}{\var{ms}}, r, h'', h')
    \end{aligned} \br[-.2em]
  & \select{h'}{r} = 91 \impliedby n \le 101 \land \Mcnp(n, r, h, h').
\end{align*}
\(\Mcnp\) additionally takes two arrays \(\var{h}, \var{h'}\) representing the (heap) memory states before/after the call of \cppi{mc91p}.
The second argument \(r\) of \(\Mcnp\), which corresponds to the pointer argument \cppi{r} in the original program, is an index for the arrays.
Hence, the assignment \cppi{*r = n - 10} is modeled in the first CHC as an update of the \(r\)-th element of the array.
\(\var{ms}\) represents the address of \cppi{s}.
This CHC system has a model
\begin{gather*}
  \Mcnp(n, r, h, h') \defiff \select{h'}{r} = 91 \lor (n > 100 \land \select{h'}{r} = n - 10),
\end{gather*}
which can be found by some array-supporting CHC solvers including Spacer \cite{Spacer},
thanks to evolving SMT-solving techniques for arrays \cite{ArrayDecision,ArrayTheory}.

However, the array-based approach has some shortcomings.
Let us consider, for example, the following innocent-looking code.\footnote{%
  \cppi{rand()} is a non-deterministic function that can return any integer value.
}
\begin{cpp}
bool just_rec(int* ma) {
  if (rand() >= 0) return true;
  int old_a = *ma; int b = rand(); just_rec(&b);
  return (old_a == *ma);
}
\end{cpp}
It can immediately return \cppi{true}; or it recursively calls itself and checks if the target of \cppi{ma} remains unchanged through the recursive call.
In effect this function \emph{does nothing} on the allocated memory blocks, although it can possibly modify some of the unused parts of the memory.

Suppose we wish to verify that \cppi{just_rec} never returns \cppi{false}.
The standard CHC-based verifier for C, SeaHorn \cite{SeaHorn}, generates a CHC system like below:\footnote{%
  \(\beq,\bne,\bge,\band\) denote binary operations that return boolean values.
}\footnote{%
  We omitted the allocation for \cppi{old_a} for simplicity.
}
\begin{align*}
  & \JustRec(\nvar{ma}, h, h',r) \impliedby h' = h \land r = \true \br[-.2em]
  & \begin{aligned}
      & \JustRec(\nvar{ma}, h, h', r) \impliedby
      \var{mb} \neq \var{ma}
      \land h'' = \updarray{h}{\nvar{mb}}{b} \\[-.3em]
      & \hspace{5em} \land \JustRec(\nvar{mb}, h'', h', \_)
      \land r = (\select{h}{\nvar{ma}} \beq \select{h'}{\nvar{ma}})
    \end{aligned} \br[-.2em]
  & r = \true \impliedby \JustRec(\nvar{ma}, h, h', r)
\end{align*}
Unfortunately the CHC system above is \emph{not} satisfiable and thus SeaHorn issues a false alarm.
This is because, in this formulation, \(\var{mb}\) may not necessarily be completely fresh; it is assumed to be different from the argument \(\var{ma}\) of the current call, but may coincide with \(\var{ma}\) of some deep ancestor calls.\footnote{%
  Precisely speaking, SeaHorn tends to even omit shallow address-freshness checks like \(\var{mb} \neq \var{ma}\).
}

The simplest remedy would be to explicitly specify the way of memory allocation.
For example, one can represent the memory state as a pair of an array \(h\) and an index \(\var{sp}\) indicating the maximum index that has been allocated so far.
\begin{align*}
  & \JustRecR(\nvar{ma}, h, \var{sp}, h', \var{sp'}, r) \impliedby h' = h \land \var{sp'} = \var{sp} \land r = \true \br[-.2em]
  & \begin{aligned}
      & \JustRecR(\nvar{ma}, h, \var{sp}, h', \var{sp'}, r) \impliedby
      \var{mb} = \var{sp''} = \var{sp} + 1 \land h'' = \updarray{h}{\nvar{mb}}{b} \\[-.3em]
      & \hspace{5em} \JustRecR(\nvar{mb}, h'', \var{sp''}, h', \var{sp'}, \_)
      \land r = (\select{h}{\nvar{ma}} \beq \select{h'}{\nvar{ma}})
    \end{aligned} \br[-.2em]
  & r = \true \impliedby \JustRecR(\nvar{ma}, h, \var{sp}, h', \var{sp'}, r) \land \var{ma} \le \var{sp}
\end{align*}
The resulting CHC system now has a model, but it involves quantifiers:
\begin{align*}
  \JustRecR(\nvar{ma}, h, \var{sp}, h', \var{sp'}, r) \defiff
  r = \true &\land \var{ma} \le \var{sp} \land \var{sp} \le \var{sp'} \\[-.3em]
  &\qquad \land \forall\,i \le \var{sp}.\, \select{h}{i} = \select{h'}{i}
\end{align*}

Finding quantified invariants is known to be difficult in general despite active studies on it \cite{SkolemQuantifiedInvariants,LazyAbstractionQuantifiedInvariants,JayHornQuantified,IC3QuantifiedInvariants,FreqHornQuantifiedInvariants}
and most current array-supporting CHC solvers give up finding quantified invariants.
In general, much more complex operations on pointers can naturally take place, which makes the universally quantified invariants highly involved and hard to automatically find.
To avoid complexity of models, CHC-based verification tools \cite{SeaHorn,SeaHornContextSensitive,JayHorn} tackle pointers by pointer analysis \cite{PointerAnalysis,DataStructureAnalysis}.
Although it does have some effects, the current applicable scope of pointer analysis is quite limited.

\subsection{Our Approach: Leverage Rust's Ownership System}
\label{index:intro-ours}

This paper proposes a novel approach to CHC-based verification of pointer-manipulating programs, which makes use of \emph{ownership} information to avoid an explicit representation of the memory.

\Paragraph{Rust-style Ownership}

Various styles of \emph{ownership/permission/capability} have been introduced to control and reason about usage of pointers on programming language design, program analysis and verification \cite{OwnershipTypes,Cyclone,OwnershipRaceDeadlock,FractionalPermissions,PermissionSeparationLogic,RaceLinear,FractionalOwnership}.
In what follows, we focus on the ownership in the style of the Rust programming language \cite{RustPaper,Rust}.

Roughly speaking, the ownership system guarantees that, for each memory cell and at each point of program execution, either (i) only one alias has the \emph{update} (write\,\&\,read) permission to the cell, with any other alias having \emph{no} permission to it, or (ii) some (or no) aliases have the \emph{read} permission to the cell, with no alias having the update permission to it.
In summary, \emph{when an alias can read some data} (with an update/read permission), \emph{any other alias cannot modify the data}.

As a running example, let us consider the program below, which follows Rust's ownership discipline (it is written in the C style; the Rust version is presented at \cref{example:cor-program}):
\begin{cpp}
int* take_max(int* ma, int* mb) {
  if (*ma >= *mb) return ma; else return mb;
}
bool inc_max(int a, int b) {
  {
    int* mc = take_max(&a, &b);  // borrow a and b
    *mc += 1;
  }                              // end of borrow
  return (a != b);
}
\end{cpp}
\Cref{fig:intro-1} illustrates which alias has the update permission to the contents of \cppi{a} and \cppi{b} during the execution of \cppi{take_max(5,3)}.
\begin{figure}[t]
  \centering
  \includegraphics[width=.6\linewidth]{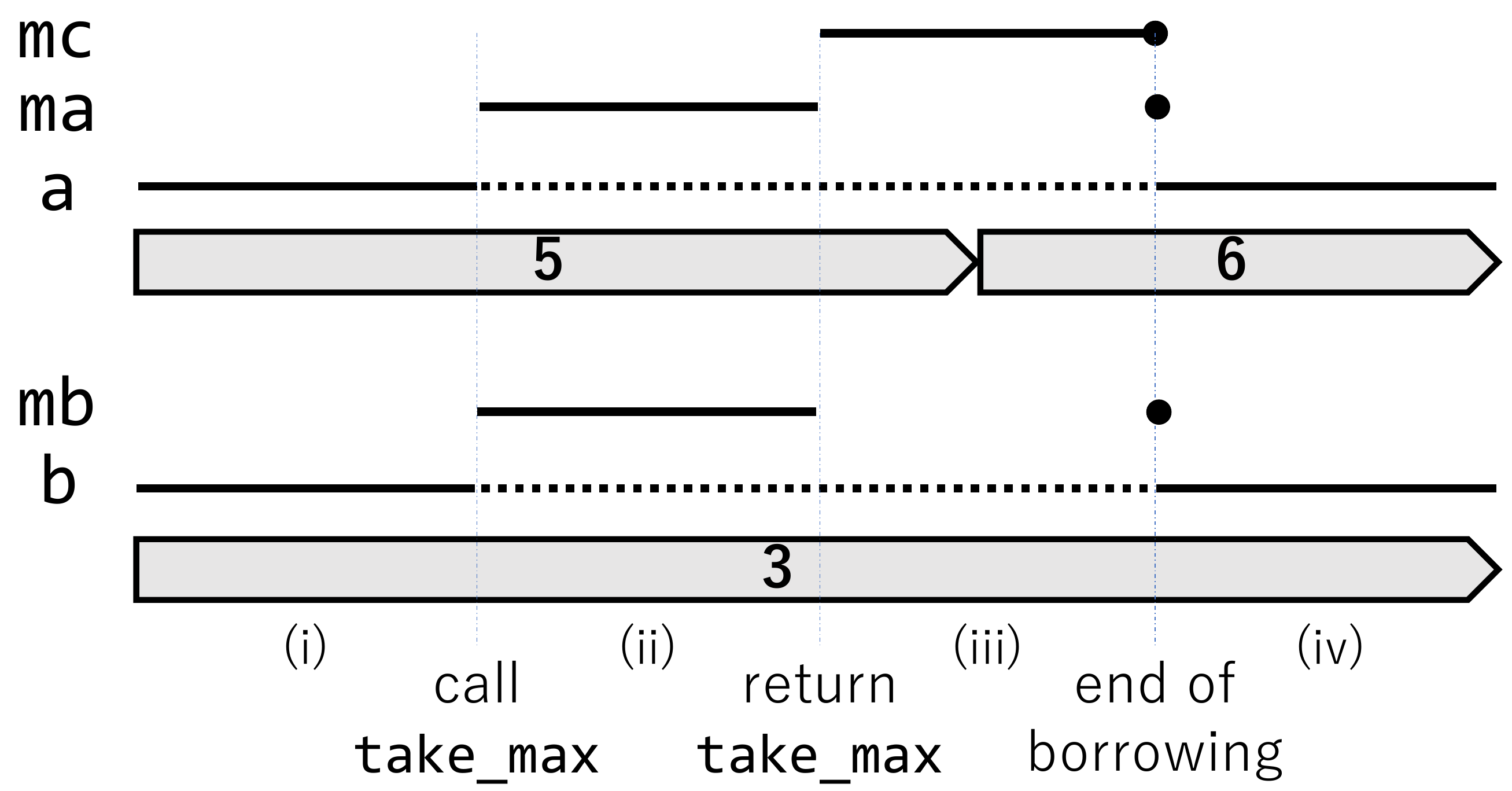}
  \caption{
    Values and aliases of \(a\) and \(b\) in evaluating \cppi{inc_max(5,3)}.
    Each line shows each variable's permission timeline:
    a solid line expresses the update permission and
    a bullet shows a point when the borrowed permission is given back.
    For example, \cppi{b} has the update permission to its content during (i) and (iv),
    but not during (ii) and (iii) because the pointer \cppi{mb}, created at the call of \cppi{take_max}, \emph{borrows} \cppi{b} until the end of (iii).
  }
  \label{fig:intro-1}
\end{figure}

A notable feature is \emph{borrow}.
In the running example, when the pointers \cppi{&a} and \cppi{&b} are taken for \cppi{take_max}, the \emph{update permissions} of \cppi{a} and \cppi{b} are \emph{temporarily transferred} to the pointers.
The original variables, \cppi{a} and \cppi{b}, \emph{lose the ability to access their contents} until the end of borrow.
The function \cppi{take_max} returns a pointer having the update permission until the end of borrow, which justifies the \emph{update operation} \cppi{*mc += 1}.
In this example, the end of borrow is at the end of the inner block of \cppi{inc_max}.
At this point, \emph{the permissions are given back} to the original variables \cppi{a} and \cppi{b}, allowing to compute \cppi{a != b}.
Note that \cppi{mc} can point to \cppi{a} and also to \cppi{b} and that this choice is determined \emph{dynamically}.
The values of \cppi{a} and \cppi{b} after the borrow \emph{depend on the behavior of the pointer} \cppi{mc}.

The end of each borrow is statically managed by a \emph{lifetime}.
See \cref{index:cor} for a more precise explanation of ownership, borrow and lifetimes.

\Paragraph{Key Idea}

The key idea of our method is to \emph{represent a pointer \cppi{ma} as a pair \(\angled{a, a_\0}\) of the current target value \(a\) and the target value \(a_\0\) at the end of borrow}.\footnote{%
  Precisely, this is the representation of a pointer with a borrowed update permission (i.e. \emph{mutable reference}).
  Other cases are discussed in \cref{index:chc}.
}\footnote{%
  For example, in the case of \cref{fig:intro-1}, when \cppi{take_max} is called, the pointer \cppi{ma} is \(\angled{5,6}\) and \cppi{mb} is \(\angled{3,3}\).
}
This representation employs \emph{access to the future information} (it is related to \emph{prophecy variables}; see \cref{index:related}).
This simple idea turns out to be very powerful.

In our approach, the verification problem ``Does \cppi{inc_max} always return \cppi{true}?'' is reduced to the satisfiability of the following CHCs:
\begin{align*}
  & \TakeMax(\angled{a,a_\0}, \angled{b,b_\0}, r) \impliedby
    a \ge b \land b_\0 = b \land r = \angled{a,a_\0} \br[-.2em]
  & \TakeMax(\angled{a,a_\0}, \angled{b,b_\0}, r) \impliedby
    a < b \land a_\0 = a \land r = \angled{b,b_\0} \br[-.2em]
  & \begin{aligned}
      \IncMax(a,b,r) \impliedby &
        \TakeMax(\angled{a,a_\0}, \angled{b,b_\0}, \angled{c,c_\0}) \land c' = c + 1 \\[-.4em]
        & \land c_\0 = c' \land r = (a_\0 \bne b_\0)
    \end{aligned} \br[-.2em]
  & r = \true \impliedby \IncMax(a,b,r).
\end{align*}
The mutable reference \cppi{ma} is now represented as \(\angled{a, a_\0}\), and similarly for \cppi{mb} and \cppi{mc}.
The first CHC models the then-clause of \cppi{take_max}: the return value is \cppi{ma}, which is expressed as \(r = \angled{a,a_\0}\); in contrast, \cppi{mb} is released, which \emph{constrains} \(b_\0\), the value of \cppi{b} at the end of borrow, to the current value \(b\).
In the clause on \(\IncMax\), \(\nvar{mc}\) is represented as a pair \(\angled{c, c_\0}\).
The constraint \(c' = c + 1 \land c_\0 = c'\) models the increment of \cppi{mc} (in the phase (iii) in \cref{fig:intro-1}).
Importantly, the final check \cppi{a != b} is simply expressed as \(a_\0 \bne b_\0\); the updated values of \cppi{a}/\cppi{b} are available as \(a_\0\)/\(b_\0\).
Clearly, the CHC system above has a simple model.

Also, the \cppi{just_rec} example in \cref{index:intro-challenges} can be encoded as a CHC system
\begin{align*}
  & \JustRec(\angled{a,a_\0}, r) \impliedby a_\0 = a \land r = \true \br[-.2em]
  & \begin{aligned}
      \JustRec(\angled{a,a_\0}, r) \impliedby &
      \var{mb} = \angled{b,b_\0} \land \JustRec(\nvar{mb}, \_) \\[-.4em]
      & \land a_\0 = a \land r = (a \beq a_\0)
    \end{aligned} \br[-.2em]
  & r = \true \impliedby \JustRec(\angled{a,a_\0}, r).
\end{align*}
Now it has a very simple model: \(\JustRec(\var{ma}, r) \defiff r = \true\).
Remarkably, arrays and quantified formulas are not required to express the model, which allows the CHC system to be easily solved by many CHC solvers.
More advanced examples are presented in \cref{index:chc-examples}, including one with destructive update on a singly-linked list.

\Paragraph{Contributions}

Based on the above idea, we formalize the translation from programs to CHC systems for a core language of Rust, prove correctness (both soundness and completeness) of the translation, and confirm the effectiveness of our approach through preliminary experiments.
The core language supports, among others, recursive types.
Remarkably, our approach enables us to automatically verify some properties of a program with destructive updates on recursive data types such as lists and trees.

The rest of the paper is structured as follows.
In \cref{index:cor}, we provide a formalized core language of Rust supporting recursions, lifetime-based ownership and recursive types.
In \cref{index:chc}, we formalize our translation from programs to CHCs and prove its correctness.
In \cref{index:expt}, we report on the implementation and the experimental results.
In \cref{index:related} we discuss related work and in \cref{index:concl} we conclude the paper.

\section{Core Language: Calculus of Ownership and Reference}
\label{index:cor}

We formalize a core of Rust as \emph{Calculus of Ownership and Reference (COR)}, whose design has been affected by the safe layer of \lambdaRust in the RustBelt paper \cite{RustBelt}.
It is a typed procedural language with a Rust-like ownership system.

\subsection{Syntax}
\label{index:cor-syntax}

The following is the syntax of COR.
\begingroup\small
\begin{gather*}
  \ltag{program} \varPi \sdef F_0\, \cdots\,F_{n-1} \br
  \ltag{function definition} F \sdef \fn f\, \varSigma\, \{L_0\colon S_0\, \cdots\,L_{n-1}\colon S_{n-1}\} \br
  \begin{aligned}
    \ltag{function signature} \varSigma \sdef
      & \angled{\alpha_0,\dots,\alpha_{m-1} \mid \alpha_{a_0} \!\le\! \alpha_{b_0}, \dots, \alpha_{a_{l-1}} \!\le\! \alpha_{b_{l-1}}} \\[-.3em]
      & \,(x_0\colon T_0, \dots, x_{n-1}\colon T_{n-1})
      \to U
  \end{aligned} \br
  \begin{aligned}
    \ltag{statement} S
      \sdef{}& I;\, \goto L
      \sor \return x \\[-.2em]
      \sor{}& \match{*x}{\inj_0\!*y_0\!\to\!\goto L_0,\, \inj_1\!*y_1\!\to\!\goto L_1}
  \end{aligned} \br
  \begin{aligned}
    \ltag{instruction} I
      \sdef{}& \Let y = \mutbor_\alpha x
      \sor \drop x
      \sor \immut x
      \sor \swap(*x,*y) \\[-.2em]
      \sor{}& \Let *y = x
      \sor \Let y = *x
      \sor \Let *y = \Copy *x
      \sor x \as T \\[-.2em]
      \sor{}& \Let y = f\angled{\alpha_0,\dots,\alpha_{m-1}}(x_0,\dots,x_{n-1}) \\[-.2em]
      \sor{}& \intro \alpha
      \sor \now \alpha
      \sor \alpha \le \beta \\[-.2em]
      \sor{}& \Let *y = \const
      \sor \Let *y = *x \op *x'
      \sor \Let *y = \rand() \\[-.4em]
      \sor{}& \Let *y = \inj^{T_0 \!+\! T_1}_i *x
      \sor \Let *y = (*x_0,*x_1)
      \sor \Let\, (*y_0,*y_1) = *x
  \end{aligned} \br[.2em]
  \ltag{type} T,U
    \sdef X \sor \mu X.T \sor P\, T \sor T_0 \!+\! T_1 \sor T_0 \!\times\! T_1 \sor \Int \sor \unit \br[-.2em]
  \ltag{pointer kind} P \sdef \own \sor R_\alpha \quad
  \ltag{reference kind} R \sdef \mut \sor \immut \br[-.2em]
  \alpha,\beta,\gamma \sdef \dtag{lifetime variable} \quad
  X,Y \sdef \dtag{type variable} \br[-.3em]
  x,y \sdef \dtag{variable} \quad
  f,g \sdef \dtag{function name} \quad
  L \sdef \dtag{label} \br
  \const \sdef n \sor {()} \quad
  \bool \defeq \unit + \unit \quad
  \op \sdef \op_\Int \sor \op_\bool \br[-.3em]
  \op_\Int \sdef {+} \!\sor\! {-} \!\sor\! \cdots \quad
  \op_\bool \sdef {\bge}\!\sor\! {\beq} \!\sor\! {\bne} \!\sor\! \cdots
\end{gather*}
\endgroup

\Paragraph{Program, Function and Label}

A program (denoted by \(\varPi\)) is a set of function definitions.
A function definition (\(F\)) consists of a function name, a function signature and a set of labeled statements (\(L\colon S\)).
In COR, for simplicity, the input/output types of a function are restricted to \emph{pointer types}.
A function is parametrized over lifetime parameters under constraints; polymorphism on types is not supported for simplicity, just as \lambdaRust.
For the lifetime parameter receiver, often \(\angled{\alpha_0,\dots\mid}\) is abbreviated to \(\angled{\alpha_0,\dots}\) and \(\angled{\mid}\) is omitted.

A label (\(L\)) is an abstract program point to be jumped to by \(\goto\).\footnote{%
  It is related to a \emph{continuation} introduced by \(\mathsf{letcont}\) in \lambdaRust.
}
Each label is assigned a \emph{whole context} by the type system, as we see later.
This style, with unstructured control flows, helps the formal description of CHCs in \cref{index:chc-body}.
A function should have the label \(\entry\) (entry point),
and every label in a function should be syntactically reachable from \(\entry\) by \(\goto\) jumps.\footnote{%
  Here `syntactically' means that detailed information such that a branch condition on \(\keyword{match}\) or non-termination is ignored.
}

\Paragraph{Statement and Instruction}

A statement (\(S\)) performs
an instruction with a jump (\(I;\, \goto L\)),
returns from a function (\(\return x\)),
or branches (\(\match{*x}{\cdots}\)).

An instruction (\(I\)) performs an elementary operation:
mutable (re)borrow (\(\Let y = \mutbor_\alpha x\)),
releasing a variable (\(\drop x\)),
weakening ownership (\(\immut\allowbreak x\)),\footnote{%
  This instruction turns a mutable reference to an immutable reference.
  Using this, an immutable borrow from \(x\) to \(y\) can be expressed by \(\Let y = \mutbor_\alpha x;\ \immut y\).
}
swap (\(\swap(*x,*y)\)),
creating/dereferencing a pointer (\(\Let *y = x\), \(\Let y = *x\)),
copy (\(\Let *y = \Copy *x\)),\footnote{%
  Copying a pointer (an immutable reference) \(x\) to \(y\) can be expressed by
  \(\Let *\var{ox} = x;\ \Let *oy = \Copy *\var{ox};\ \Let y = *oy\).
}
type weakening (\(x \as T\)),
function call (\(\Let y = f\angled{\cdots}({\cdots})\)),
lifetime-related ghost operations (\(\intro \alpha,\, \now \alpha,\allowbreak\, \alpha \le \beta\); explained later),
getting a constant\,/\,operation result\,/\,random integer (\(\Let *y = \const\) / \(*x \op *x'\) / \(\rand()\)),
creating a variant (\(\Let *y = \inj^{T_0 \!+\! T_1}_i *x\)), and
creating/destruct\-ing a pair (\(\Let *y = (*x_0,*x_1),\, \Let\, (*y_0,*y_1) = *x\)).
An instruction of form \(\Let *y = \cdots\) implicitly allocates new memory cells as \(y\);
also, some instructions deallocate memory cells implicitly.
For simplicity, every variable is designed to be a \emph{pointer} and every \emph{release of a variable} should be explicitly annotated by `\(\drop x\)'.
In addition, we provide swap instead of assignment;
the usual assignment (of copyable data from \(*x\) to \(*y\)) can be expressed by \(\Let *x' = \Copy *x;\ \swap(*y, *x');\ \drop x'\).

\Paragraph{Type}

As a type (\(T\)), we support recursive types (\(\mu X.T\)), pointer types (\(P\,T\)), variant types (\(T_0 + T_1\)), pair types (\(T_0 \times T_1\)) and basic types (\(\Int,\unit\)).

A pointer type \(P\, T\) can be an \emph{owning pointer} \(\own T\) (\rusti{Box<T>} in Rust), \emph{mutable reference} \(\mut_\alpha T\) (\rusti{&'a mut T}) or \emph{immutable reference} \(\immut_\alpha T\) (\rusti{&'a T}).
An \emph{owning pointer} has data in the heap memory,
can freely update the data (unless it is borrowed),
and has the obligation to clean up the data from the heap memory.
In contrast, a \emph{mutable/immutable reference} (or \emph{unique/shared reference}) borrows an update/read permission from an owning pointer or another reference with the deadline of a \emph{lifetime} \(\alpha\) (introduced later).
A mutable reference cannot be copied, while an immutable reference can be freely copied.
A reference loses the permission at the time when it is released.\footnote{%
  In Rust, even after a reference loses the permission and the lifetime ends, its address data can linger in the memory, although dereferencing on the reference is no longer allowed.
  We simplify the behavior of lifetimes in COR.
}

A type \(T\) that appears in a program (not just as a substructure of some type) should satisfy the following condition (if it holds we say the type is \emph{complete}):
every type variable \(X\) in \(T\) is bound by some \(\mu\) and guarded by a pointer constructor (i.e. given a binding of form \(\mu X.U\), every occurrence of \(X\) in \(U\) is a part of a pointer type, of form \(P\,U'\)).

\Paragraph{Lifetime}

A \emph{lifetime} is an \emph{abstract time point in the process of computation},\footnote{%
  In the terminology of Rust, a lifetime often means a time range where a borrow is active.
  To simplify the discussions, however, we in this paper use the term lifetime to refer to a \emph{time point when a borrow ends}.
}
which is statically managed by \emph{lifetime variables} \(\alpha\).
A lifetime variable can be a \emph{lifetime parameter} that a function takes or a \emph{local lifetime variable} introduced within a function.
We have three lifetime-related ghost instructions:
\(\intro \alpha\) introduces a new local lifetime variable,
\(\now \alpha\) sets a local lifetime variable to the current moment and eliminates it,
and \(\alpha \le \beta\) asserts the ordering on local lifetime variables.

\Subsubsection{Expressivity and Limitations}\label{index:cor-syntax-discuss}

COR can express most borrow patterns in the core of Rust.
The set of moments when a borrow is active forms a continuous time range, even under \emph{non-lexical lifetimes} \cite{NLL}.\footnote{%
  Strictly speaking, this property is broken by recently adopted implicit two-phase borrows \cite{TwoPhaseBorrows,NestedMethodCalls}.
  However, by shallow syntactical reordering, a program with implicit two-phase borrows can be fit into usual borrow patterns.
}

A major limitation of COR is that it does not support \emph{unsafe code blocks} and also lacks \emph{type traits and closures}.
Still, our idea can be combined with unsafe code and closures, as discussed in \cref{index:chc-discuss}.
Another limitation of COR is that, unlike Rust and \lambdaRust, we \emph{cannot directly modify/borrow a fragment of a variable} (e.g. an element of a pair).
Still, we can eventually modify/borrow a fragment by borrowing the whole variable and \emph{splitting pointers} (e.g. `\(\Let \,(*y_0,*y_1) = *x\)').
This borrow-and-split strategy, nevertheless, yields a subtle obstacle when we extend the calculus for advanced data types (e.g. \rusti{get_default} in `Problem Case \#3' from \cite{NLL}).
For future work, we pursue a more expressive calculus modeling Rust and extend our verification method to it.

\begin{example}[COR Program]\label{example:cor-program}
The following program expresses the functions \rusti{take_max} and \rusti{inc_max} presented in \cref{index:intro-ours}.
We shorthand sequential executions by `\(;^L\)' (e.g. \(L_0\colon I_0;^{L_1} I_1; \goto L_2\) stands for \(L_0\colon I_0;\, \goto L_1\ L_1\colon I_1;\, \goto L_2\)).\footnote{%
  The first character of each variable indicates the pointer kind (\(o\)/\(m\) corresponds to \(\own\)/\(\mut_\alpha\)).
  We swap the branches of the \(\keyword{match}\) statement in \(\takemax\), to fit the order to C/Rust's \rusti{if}.
}
\begingroup\small
\begin{align*}
  &
  \fn \takemax\, \angled{\alpha}\,
    (\nvar{ma}\colon \mut_\alpha \Int,\, \var{mb}\colon \mut_\alpha \Int)
    \to \mut_\alpha \Int\, \{
  \br[-.4em] & \quad
    \entry\colon
      \Let *\var{ord} = *\var{ma} \bge *\var{mb};^{\nL{1}}
      \match{*\var{ord}}{\inj_1 *\var{ou}\to \goto \nL{2},\ \inj_0 *\var{ou}\to \goto \nL{5}}
  \br[-.4em] & \quad
    \nL{2}\colon
      \drop \var{ou};^{\nL{3}}
      \drop \var{mb};^{\nL{4}}
      \return \var{ma} \quad
    \nL{5}\colon
      \drop \var{ou};^{\nL{6}}
      \drop \var{ma};^{\nL{7}}
      \return \var{mb}
  \br[-.4em] &
  \}
  \br[.1em] &
  \fn \incmax
    (\nvar{oa}\colon \own \Int,\, \var{ob}\colon \own \Int)
    \to \own \bool\, \{
  \br[-.4em] & \quad
    \entry\colon
      \intro \alpha;^{\nL{1}}
      \Let \var{ma} = \mutbor_\alpha \var{oa};^{\nL{2}}
      \Let \var{mb} = \mutbor_\alpha \var{ob};^{\nL{3}}
  \br[-.4em] & \quad\
      \Let \var{mc} = \takemax\angled{\alpha}(ma,mb);^{\nL{4}}
      \Let *\var{o1} = 1;^{\nL{5}}
      \Let *\var{oc'} = *\var{mc} + *\var{o1};^{\nL{6}}
      \drop \var{o1};^{\nL{7}}
  \br[-.4em] & \quad\
      \swap (\nvar{mc},\var{oc'});^{\nL{8}}
      \drop \var{oc'};^{\nL{9}}
      \drop \var{mc};^{\nL{10}}
      \now \alpha;^{\nL{11}}
      \Let *\var{or} = *\var{oa} \bne *\var{ob};^{\nL{12}}
  \br[-.4em] & \quad\
      \drop \var{oa};^{\nL{13}}
      \drop \var{ob};^{\nL{14}}
      \return \var{or}
  \br[-.4em] &
  \}
\end{align*}
\endgroup
In \(\takemax\), conditional branching is performed by \(\keyword{match}\) and its \(\goto\) directions (at \(\nL{1}\)).
In \(\incmax\), increment on the mutable reference \(\var{mc}\) is performed by calculating the new value (at \(\nL{4},\nL{5}\)) and updating the data by swap (at \(\nL{7}\)).

The following is the corresponding Rust program, with ghost annotations (marked italic and dark green, e.g. \rusti{@(drop ma)@}) on lifetimes and releases of mutable references.
\begin{rust}
fn take_max<'a>(ma: &'a mut i32, mb: &'a mut i32) -> &'a mut i32 {
  if *ma >= *mb { @(drop mb;)@ ma } else { @(drop ma;)@ mb }
}
fn inc_max(mut a: i32, mut b: i32) -> bool {
  { @(intro 'a;)@
    let mc = take_max@(<'a>)@(&@('a)@ mut a, &@('a)@ mut b); *mc += 1;
  @(drop mc;)@ @(now 'a;)@ }
  a != b
}
\end{rust}
\end{example}

\subsection{Type System}
\label{index:cor-type-system}

The type system of COR assigns to each label a \emph{whole context} \((\cGamma,\cA)\).
We define below the whole context and the typing judgments.

\Paragraph{Context}
A \emph{variable context} \(\cGamma\) is a finite set of items of form \(x\colonu\ac T\),
where \(T\) should be a complete \emph{pointer} type and \(\ac\) (which we call \emph{activeness}) is of form `\(\Active\)' or `\(\dagger\alpha\)' (\emph{frozen} until lifetime \(\alpha\)).
We abbreviate \(x\colonu{\Active} T\) as \(x\colon T\).
A variable context should not contain two items on the same variable.
A \emph{lifetime context} \(\cA = (A,R)\) is a finite preordered set of lifetime variables, where \(A\) is the underlying set and \(R\) is the preorder.
We write \(\lvert\cA\rvert\) and \(\le_\cA\) to refer to \(A\) and \(R\).
Finally, a \emph{whole context} \((\cGamma,\cA)\) is a pair of a variable context \(\cGamma\) and a lifetime context \(\cA\) such that every lifetime variable in \(\cGamma\) is contained in \(\cA\).

\Paragraph{Notations}
The set operation \(A + B\) (or more generally \(\sum_\lambda A_\lambda\)) denotes the disjoint union, i.e. the union defined only if the arguments are disjoint.
The set operation \(A - B\) denotes the set difference defined only if \(A \supseteq B\).
For a natural number \(n\), \([n]\) denotes the set \(\{0, \dots, n \!-\! 1\}\).

Generally, an auxiliary definition for a rule can be presented just below, possibly in a dotted box.

\Paragraph{Program and Function}
The rules for typing programs and functions are presented below.
They assign to each label a whole context \((\cGamma,\cA)\).
`\(S \colond{\varPi,f} (\cGamma,\cA) \mid (\cGamma_L,\cA_L)_L \mid U\)' is explained later.
\begingroup\small
\begin{gather*}
  \frac{
    \text{for any}\,F\, \text{in}\, \varPi,\
    F\colond\varPi (\cGamma_{\name(F),L},\cA_{\name(F),L})_{L \in \Label_F}
  }{
    \varPi\colon (\cGamma_{f,L},\cA_{f,L})_{(f,L)\, \in\, \FnLabel_\varPi}
  } \br[-.1em]
  \text{
    \(\name(F)\): the function name of \(F\) \quad
    \(\Label_F\): the set of labels in \(F\)
  } \br[-.2em]
  \text{
    \(\FnLabel_\varPi\): the set of pairs \((f,L)\) such that a function \(f\) in \(\varPi\) has a label \(L\)
  } \br[.2em]
  \frac{\begin{gathered}
    F = \fn f\angled{\alpha_0,\dots,\alpha_{m-1}
    \!\mid\! \alpha_{a_0} \!\le\! \alpha_{b_0}, \dots, \alpha_{a_{l-1}} \!\le\! \alpha_{b_{l-1}}} (x_0\colon T_0, \dots, x_{n-1}\colon T_{n-1}) \!\to\! U\, \{{\cdots}\} \\[-.3em]
    \cGamma_\entry = \{x_i\colon T_i \!\mid\! i \!\in\! [n]\} \quad
    A = \{\alpha_j \!\mid\! j \!\in\! [m]\} \ \
    \cA_\entry = \bigl(A,\bigl(\Id_A \!\cup\! \{(\alpha_{a_k},\alpha_{b_k}) \!\mid\! k \!\in\! [l]\}\bigr)^+\bigr) \\[-.2em]
    \text{for any}\ L'\colon S \in \LabelStmt_F,\
    S \colond{\varPi,f} (\cGamma_{L'},\cA_{L'}) \mid (\cGamma_L,\cA_L)_{L \in \Label_F} \mid U
  \end{gathered}}{\begin{aligned}
    F \colond\varPi\ (\cGamma_L,\cA_L)_{L \in \Label_F}
  \end{aligned}} \br[-.1em]
  \text{
    \(\LabelStmt_F\): the set of labeled statements in \(F\)
  } \br[-.4em]
  \text{
    \(\Id_A\): the identity relation on \(A\) \quad
    \(R^+\): the transitive closure of \(R\)
  }
\end{gather*}
\endgroup

On the rule for the function,
the initial whole context at \(\entry\) is specified (the second and third preconditions) and also the contexts for other labels are checked (the fourth precondition).
The context for each label (in each function) can actually be determined in the order by the distance in the number of \(\goto\) jumps from \(\entry\), but that order is not very obvious because of \emph{unstructured control flows}.

\Paragraph{Statement}
`\(S \colond{\varPi,f} (\cGamma,\cA) \mid (\cGamma_L,\cA_L)_L \mid U\)' means that running the statement \(S\) (under \(\varPi,f\)) with the whole context \((\cGamma,\cA)\) results in a jump to a label with the whole contexts specified by \((\cGamma_L,\cA_L)_L\) or a return of data of type \(U\).
Its rules are presented below.
`\(I\colond{\varPi,f} (\cGamma,\cA) \to (\cGamma',\cA')\)' is explained later.
\begingroup\small
\begin{gather*}
  \frac{
    I\colond{\varPi,f} (\cGamma,\cA) \to (\cGamma_{L_0},\cA_{L_0})
  }{
    I;\, \goto {L_0}\colond{\varPi,f} (\cGamma,\cA) \mid (\cGamma_L,\cA_L)_L \mid U
  } \quad
  \frac{
    \cGamma = \{x\colon U\} \quad
    \lvert\cA\rvert = A_{\ex\, \varPi,f}
  }{
    \return x\colond{\varPi,f} (\cGamma,\cA) \mid (\cGamma_L,\cA_L)_L \mid U
  } \br[-.1em]
  \text{
    \(A_{\ex\, \varPi,f}\): the set of lifetime parameters of \(f\) in \(\varPi\)
  } \br[.2em]
  \frac{\begin{gathered}
    x\colon P\,(T_0 \!+\! T_1) \in \cGamma \\[-.4em]
    \text{for}\ i = 0,1,\
    (\cGamma_{L_i},\cA_{L_i}) = (\cGamma \!-\! \{x\colon P\,(T_0 \!+\! T_1)\} \!+\! \{y_i\colon P\, T_i\},\, \cA)
  \end{gathered}}{
    \match{*x}{\inj_0 *y_0 \to \goto L_0,\ \inj_1 *y_1 \to \goto L_1}\colond{\varPi,f} (\cGamma,\cA) \mid (\cGamma_L,\cA_L)_L \mid U
  }
\end{gather*}
\endgroup
The rule for the \(\return\) statement ensures that there remain no extra variables and local lifetime variables.

\Paragraph{Instruction}
`\(I\colond{\varPi,f} (\cGamma,\cA) \to (\cGamma',\cA')\)' means that running the instruction \(I\) (under \(\varPi,f\)) updates the whole context \((\cGamma,\cA)\) into \((\cGamma',\cA')\).
The rules are designed so that, for any \(I\), \(\varPi\), \(f\), \((\cGamma,\cA)\), there exists at most one \((\cGamma',\cA')\) such that \(I\colond{\varPi,f} (\cGamma,\cA) \to (\cGamma',\cA')\) holds.
Below we present some of the rules;
the complete rules are presented in \shortorfull{the full paper}{\cref{index:appx-cor-typing-instructions}}.
The following is the typing rule for mutable (re)borrow.
\begingroup\small
\begin{gather*}
  \frac{
    \alpha \notin A_{\ex\, \varPi,f} \quad
    P = \own, \mut_\beta \quad
    \text{for any}\ \gamma \in \Lifetime_{P\, T},\
    \alpha \le_\cA \gamma
  }{
    \Let y = \mutbor_\alpha x \, \colond{\varPi,f} (\cGamma \!+\! \{x\colon P\, T\},\, \cA) \to (\cGamma \!+\! \{y\colon \mut_\alpha T,\, x\colonu{\dagger\alpha} P\, T\},\, \cA)
  } \br[-.1em]
  \text{
    \(\Lifetime_T\): the set of lifetime variables occurring in \(T\)
  }
\end{gather*}
\endgroup
After you mutably (re)borrow an owning pointer / mutable reference \(x\) until \(\alpha\), \(x\) is \emph{frozen} until \(\alpha\).
Here, \(\alpha\) should be a local lifetime variable\footnote{%
  In COR, a reference that lives after the return from the function should be created by splitting a reference (e.g. `\(\Let \,(*y_0,*y_1) = *x\)') given in the inputs;
  see also \nameref{index:cor-syntax-discuss}.
} (the first precondition) that does not live longer than the data of \(x\) (the third precondition).
Below are the typing rules for local lifetime variable introduction and elimination.
\begingroup\small
\begin{gather*}
  \intro \alpha \, \colond{\varPi,f} \bigl(\cGamma,(A,R)\bigr) \to \bigl(\cGamma,(\{\alpha\} \!+\! A,\, \{\alpha\} \!\times\! (\{\alpha\} \!+\! A_{\ex\, \varPi,f}) \!+\! R)\bigr) \br[.3em]
  \frac{
    \alpha \notin A_{\ex\, \varPi,f}
  }{
    \now \alpha \, \colond{\varPi,f} \bigl(\cGamma,(\{\alpha\} \!+\! A,\,R)\bigr)
    \to \bigl(\{\thaw_\alpha(x\colonu{\ac} T) \mid x\colonu{\ac} T \!\in\! \cGamma\},\,(A,\{(\beta,\gamma) \!\in\! R \mid \beta \!\ne\! \alpha\})\bigr)
  } \br[-.1em]
  \thaw_\alpha(x\colonu\ac T) \defeq \begin{cases}
    x\colon T &\!\!\! (\ac = \dagger\alpha) \\[-.3em]
    x\colonu\ac T &\!\!\! (\text{otherwise})
  \end{cases}
\end{gather*}
\endgroup
On \(\intro \alpha\),
it just ensures the new local lifetime variable to be earlier than any lifetime parameters (which are given by exterior functions).
On \(\now \alpha\),
the variables frozen with \(\alpha\) get active again.
Below is the typing rule for dereference of a pointer to a pointer, which may be a bit interesting.
\begingroup\small
\begin{gather*}
  \Let y = *x \, \colond{\varPi,f} (\cGamma \!+\! \{x\colon P\,P'\,T\},\, \cA) \to (\cGamma \!+\! \{y\colon (P \!\circ\! P')\,T\},\, \cA) \br[-.2em]
  \dbox{\(
    P \circ {\own} = {\own} \circ P \defeq P \quad
    R_\alpha \circ R'_\beta \defeq R''_\alpha\
    \text{where}\ R'' = \begin{cases}
      \mut &\!\! (R = R' = \mut) \\[-.3em]
      \immut &\!\! (\text{otherwise})
    \end{cases}
  \)}
\end{gather*}
\endgroup
The third precondition of the typing rule for \(\mutbor\) justifies taking just \(\alpha\) in the rule `\(R_\alpha \circ R'_\beta \defeq R''_\alpha\)'.

\vspace{\baselineskip}

Let us interpret \(\varPi\colon (\cGamma_{f,L},\cA_{f,L})_{(f,L)\, \in\, \FnLabel_\varPi}\) as ``the program \(\varPi\) has the type \((\cGamma_{f,L},\cA_{f,L})_{(f,L)\, \in\, \FnLabel_\varPi}\)''.
The type system ensures that any program has at most one type (which may be a bit unclear because of unstructured control flows).
Hereinafter, we implicitly assume that a program has a type.

\subsection{Concrete Operational Semantics}
\label{index:cor-cos}

We introduce for COR \emph{concrete operational semantics}, which handles a concrete model of the heap memory.

The basic item, \emph{concrete configuration} \(\cC\), is defined as follows.
\begin{gather*}
  \cS \sdef \End \mathrel{\, \bigm|\,} [f,L]\, x,\cF;\, \cS \quad
  \ltag{concrete configuration} \cC \sdef [f,L]\, \cF;\, \cS \mid \cH
\end{gather*}

Here, \(\cH\) is a \emph{heap}, which maps addresses (represented by integers) to integers (data).
\(\cF\) is a \emph{concrete stack frame}, which maps variables to addresses.
The stack part of \(\cC\) is of form `\([f,L]\, \cF;\, [f',L']\, x,\cF';\, \cdots;\, \End\)' (we may omit the terminator `\(;\, \End\)').
\([f,L]\) on each stack frame indicates the program point.
`\(x,\)' on each non-top stack frame is the receiver of the value returned by the function call.

Concrete operational semantics is characterized by the one-step transition relation \(\cC \to_\varPi \cC'\) and the termination relation \(\final_\varPi(\cC)\), which can be defined straightforwardly.
Below we show the rules for mutable (re)borrow, swap, function call and return from a function;
the complete rules and an example execution are presented in \shortorfull{the full paper}{\cref{index:appx-cor-cos}}.
\(S_{\varPi,f,L}\) is the statement for the label \(L\) of the function \(f\) in \(\varPi\).
\(\Ty_{\varPi,f,L}(x)\) is the type of variable \(x\) at the label.
\begingroup\small
\begin{gather*}
  \frac{
    S_{\varPi,f,L} = \Let y = \mutbor_\alpha x;\, \goto L' \quad \cF(x) = a
  }{
    [f,L]\, \cF;\, \cS \mid \cH
    \ \to_\varPi\ [f,L']\, \cF \!+\! \{(y,a)\};\, \cS \mid \cH
  } \br[.3em]
  \frac{
    S_{\varPi,f,L} = \swap(*x,*y);\, \goto L' \quad
    \Ty_{\varPi,f,L}(x) = P\, T \quad
    \cF(x) = a \quad
    \cF(y) = b
  }{
    \begin{aligned}
      & [f,L]\, \cF;\, \cS \mid \cH \!+\! \{(a \!+\! k, m_k) \!\mid\! k \!\in\! [\#T]\} \!+\! \{(b \!+\! k, n_k) \!\mid\! k \!\in\! [\#T]\} \\[-.3em]
      & \ \to_\varPi\ [f,L']\, \cF;\, \cS \mid \cH \!+\! \{(a \!+\! k, n_k) \!\mid\! k \!\in\! [\#T]\} \!+\! \{(b \!+\! k, m_k) \!\mid\! k \!\in\! [\#T]\}
    \end{aligned}
  } \br[.3em]
  \frac{\begin{gathered}
    S_{\varPi,f,L} = \Let y = g\angled{\cdots}(x_0,\dots,x_{n-1});\, \goto L' \\[-.2em]
    \varSigma_{\varPi,g} = \angled{\cdots}(x'_0\colon T_0, \dots, x'_{n-1}\colon T_{n-1}) \to U
  \end{gathered}}{
    [f,L]\, \cF \!+\! \{(x_i,a_i) \!\mid\! i \!\in\! [n]\};\, \cS \mid \cH \to_\varPi\ [g,\entry]\, \{(x'_i,a_i) \!\mid\! i \!\in\! [n]\};\,
    [f,L]\,y,\cF;\, \cS \mid \cH
  } \br[.3em]
  \frac{
    S_{\varPi,f,L} = \return x
  }{
    [f,L]\, \{(x,a)\}; [g,L']\, x',\cF'; \cS \mid \cH
    \to_\varPi [g,L']\, \cF' \!+\! \{(x',a)\}; \cS \mid \cH
  } \br[.3em]
  \frac{
    S_{\varPi,f,L} = \return x \quad
  }{
    \final_\varPi\bigl([f,L]\, \{(x,a)\} \mid \cH\bigr)
  }
\end{gather*}
\endgroup
Here we introduce `\(\#T\)', which represents how many memory cells the type \(T\) takes (at the outermost level).
\(\#T\) is defined for every \emph{complete} type \(T\), because every occurrence of type variables in a complete type is guarded by a pointer constructor.
\begingroup\small
\begin{gather*}
  \#(T_0 \!+\! T_1) \defeq 1 + \max\{\#T_0,\#T_1\} \quad
  \#(T_0 \!\times\! T_1) \defeq \#T_0 + \#T_1 \br[-.1em]
  \#\, \mu X.T \defeq \#\,T[\mu X.T/X] \quad
  \#\, \Int = \#\,P\, T \defeq 1 \quad
  \#\, \unit = 0
\end{gather*}
\endgroup

\section{CHC Representation of COR Programs}
\label{index:chc}

To formalize the idea discussed in \cref{index:intro}, we give a translation from COR programs to CHC systems,
which precisely characterize the input-output relations of the COR programs.
We first define the logic for CHCs (\cref{index:chc-logic}).
We then formally describe our translation (\cref{index:chc-body}) and prove its correctness (\cref{index:chc-correct}).
Also, we examine effectiveness of our approach with advanced examples (\cref{index:chc-examples}) and discuss how our idea can be extended and enhanced (\cref{index:chc-discuss}).

\subsection{Multi-sorted Logic for Describing CHCs}
\label{index:chc-logic}

To begin with, we introduce a first-order multi-sorted logic for describing the CHC representation of COR programs.

\Subsubsection{Syntax}

The syntax is defined as follows.
\begingroup\small
\begin{gather*}
  \ltag{CHC} \varPhi \sdef
    \forall\,x_0\colon \sigma_0, \dots, x_{m-1}\colon \sigma_{m-1}.\ \
    \check\varphi \impliedby \psi_0 \land \cdots \land \psi_{n-1} \br[-.2em]
  \top \defeq \text{the nullary conjunction of formulas} \br[-.2em]
  \ltag{formula} \varphi,\psi \sdef f(t_0,\dots,t_{n-1}) \quad
  \ltag{elementary formula} \check\varphi \sdef f(p_0,\dots,p_{n-1}) \br[-.2em]
  \ltag{term} t \sdef x \sor \angled{t} \sor \angled{t_*, t_\0} \sor \inj_i t \sor (t_0,t_1) \sor {*}t \sor \0t \sor t.i \sor \const \sor t \op t' \br[-.2em]
  \ltag{value} v, w \sdef \angled{v} \sor \angled{v_*, v_\0} \sor \inj_i v \sor (v_0,v_1) \sor \const \br[-,2em]
  \ltag{pattern} p,q \sdef x \sor \angled{p} \sor \angled{p_*, p_\0} \sor \inj_i p \sor (p_0,p_1) \sor \const \br[-.2em]
  \ltag{sort} \sigma,\tau \sdef X \sor \mu X.\sigma \sor C\, \sigma \sor \sigma_0 + \sigma_1 \sor \sigma_0 \times \sigma_1 \sor \Int \sor \unit \br[-.2em]
  \ltag{container kind} C \sdef \Box \sor \mut \quad
  \const \sdef \text{same as COR} \quad \op \sdef \text{same as COR} \br[-.2em]
  \bool \defeq \unit + \unit \quad
  \true \defeq \inj_1 {()} \quad
  \false \defeq \inj_0 {()} \br[-.2em]
  X \sdef \dtag{sort variable} \quad
  x,y \sdef \dtag{variable} \quad
  f \sdef \dtag{predicate variable}
\end{gather*}
\endgroup

We introduce \(\Box \sigma\) and \(\mut \sigma\), which correspond to \(\own T\)/\(\immut_\alpha T\) and \(\mut_\alpha T\) respectively.
\(\angled{t}\)/\(\angled{t_*,t_\0}\) is the constructor for \(\Box \sigma\)/\(\mut \sigma\).
\(*t\) takes the body/first value of \(\angled{-}\)/\(\angled{-,\! -}\) and \(\0t\) takes the second value of \(\angled{-,\! -}\).
We restrict the form of CHCs here to simplify the proofs later.
Although the logic does not have a primitive for equality, we can define the equality in a CHC system (e.g. by adding \(\forall\,x\colon \sigma.\, \var{Eq}(x,x) \impliedby \top\)).

A \emph{CHC system} \((\cPhi,\cXi)\) is a pair of a finite set of CHCs \(\cPhi = \{\varPhi_0,\dots,\varPhi_{n-1}\}\) and \(\cXi\),
where \(\cXi\) is a finite map from predicate variables to tuples of sorts (denoted by \(\varXi\)), specifying the sorts of the input values.
Unlike the informal description in \cref{index:intro}, we add \(\cXi\) to a CHC system.

\Subsubsection{Sort System}

`\(t \colond\cDelta \sigma\)' (the term \(t\) has the sort \(\sigma\) under \(\cDelta\)) is defined as follows.
Here, \(\cDelta\) is a finite map from variables to sorts.
\(\sigma \sim \tau\) is the congruence on sorts induced by \(\mu X.\sigma \sim \sigma[\mu X.\sigma/X]\).
\begingroup\small
\begin{gather*}
  \frac{
    \cDelta(x) = \sigma
  }{
    x\colond\cDelta \sigma
  } \quad
  \frac{
    t\colond\cDelta \sigma
  }{
    \angled{t}\colond\cDelta \Box \sigma
  } \quad
  \frac{
    t_*, t_\0 \colond\cDelta \sigma
  }{
    \angled{t_*, t_\0}\colond\cDelta \mut \sigma
  } \quad
  \frac{
    t\colond\cDelta \sigma_i
  }{
    \inj_i t \colond\cDelta \sigma_0 + \sigma_1
  } \quad
  \frac{
    t_0\colond\cDelta \sigma_0 \quad
    t_1\colond\cDelta \sigma_1
  }{
    (t_0,t_1) \colond\cDelta \sigma_0 \times \sigma_1
  } \br[-.1em]
  \frac{
    t\colond\cDelta C\, \sigma
  }{
    *t\colond\cDelta \sigma
  } \quad
  \frac{
    t\colond\cDelta \mut \sigma
  }{
    \0t\colond\cDelta \sigma
  } \quad
  \frac{
    t\colond\cDelta \sigma_0 + \sigma_1
  }{
    t.i\colond\cDelta \sigma_i
  } \quad
  \const \colond\cDelta \sigma_\const \quad
  \frac{
    t, t' \colond\cDelta \Int
  }{
    t \op t' \colond\cDelta \sigma_{\op}
  } \quad
  \frac{
    t\colond\cDelta \sigma \quad \sigma \sim \tau
  }{
    t\colond\cDelta \tau
  } \br[-.1em]
  \text{
    \(\sigma_\const\): the sort of \(\const\) \quad
    \(\sigma_{\op}\): the output sort of \(\op\)
  }
\end{gather*}
\endgroup

`\(\wellSorted_{\cDelta,\cXi}(\varphi)\)' and `\(\wellSorted_\cXi(\cPhi)\)', the judgments on well-sortedness of formulas and CHCs, are defined as follows.
\begingroup\small
\begin{gather*}
  \frac{
    \cXi(f) = (\sigma_0,\dots,\sigma_{n-1}) \quad
    \text{for any}\ i \in [n],\
    t_i \colond\cDelta \sigma_i
  }{
    \wellSorted_{\cDelta,\cXi}(f(t_0,\dots,t_{n-1}))
  } \br[.3em]
  \frac{\begin{gathered}
    \cDelta = \{(x_i,\sigma_i) \mid i \!\in\! [m]\} \quad
    \wellSorted_{\cDelta,\cXi}(\check\varphi) \quad
    \text{for any}\ j \in [n],\,
    \wellSorted_{\cDelta,\cXi}(\psi_j)
  \end{gathered}}{
    \wellSorted_\cXi\bigl(
      \forall x_0\colon \sigma_0, \dots, x_{m-1}\colon \sigma_{m-1}.\ \
      \check\varphi \impliedby \psi_0 \land \cdots \land \psi_{n-1}
    \bigr)
  }
\end{gather*}
\endgroup
The CHC system \((\cPhi,\cXi)\) is said to be well-sorted if \(\wellSorted_\cXi(\varPhi)\) holds for any \(\varPhi \in \cPhi\).

\Subsubsection{Semantics}

`\(\Bracked{t}_\cI\)', the interpretation of the term \(t\) as a value under \(\cI\), is defined as follows.
Here, \(\cI\) is a finite map from variables to values.
Although the definition is partial, the interpretation is defined for all well-sorted terms.
\begingroup\small
\begin{gather*}
  \Bracked{x}_\cI \defeq \cI(x) \quad
  \Bracked{\angled{t}}_\cI \defeq \angled{\Bracked{t}_\cI} \quad
  \Bracked{\angled{t_*, t_\0}}_\cI \defeq \angled{\Bracked{t_*}_\cI,\Bracked{t_\0}_\cI} \quad
  \Bracked{\inj_i t}_\cI \defeq \inj_i \Bracked{t}_\cI \br[-.3em]
  \Bracked{(t_0,t_1)}_\cI \defeq (\Bracked{t_0}_\cI,\Bracked{t_1}_\cI) \quad
  \Bracked{*t}_\cI\defeq \begin{cases}
    v &\!\!\! (\Bracked{t}_\cI = \angled{v}) \\[-.3em]
    v_* &\!\!\! (\Bracked{t}_\cI = \angled{v_*, v_\0})
  \end{cases} \quad
  \Bracked{\0t}_\cI\defeq v_\0\ \text{if}\ \Bracked{t}_\cI = \angled{v_*, v_\0} \br[-.3em]
  \Bracked{t.i}_\cI\defeq v_i\ \text{if}\ \Bracked{t}_\cI = (v_0,v_1) \quad
  \Bracked{\const}_\cI \defeq \const \quad
  \Bracked{t \op t'}_\cI\defeq \Bracked{t}_\cI \mathop{\Bracked{\op}} \Bracked{t'}_\cI \br[-.0em]
  \text{
    \(\Bracked{\op}\): the binary operation on values corresponding to \(\op\)
  }
\end{gather*}
\endgroup

A \emph{predicate structure} \(\cM\) is a finite map from predicate variables to (concrete) predicates on values.
\(\cM,\cI \models f(t_0,\dots,t_{n-1})\) means that \(\cM(f)(\Bracked{t_0}_\cI,\dots, \Bracked{t_{n-1}}_\cI)\) holds.
\(\cM \models \varPhi\) is defined as follows.
\begingroup\small
\begin{gather*}
  \frac{
    \text{for any}\ \cI\ \text{s.t.}\ \forall\,i \!\in\! [m].\, \cI(x_i)\colond\emp \sigma_i,\
    \cM,\cI \models \psi_0, \dots, \psi_{n-1}\ \text{implies}\ \cM,\cI \models \check\varphi
  }{
    \cM \, \models\, \forall\,x_0\colon \sigma_0, \dots, x_{m-1}\colon \sigma_{m-1}.\ \
    \check\varphi \impliedby \psi_0 \land \cdots \land \psi_{n-1}
  }
\end{gather*}
\endgroup
Finally, \(\cM \models (\cPhi,\cXi)\) is defined as follows.
\begingroup\small
\begin{gather*}
  \frac{\begin{gathered}
    \text{for any}\ (f,(\sigma_0,\dots,\sigma_{n-1})) \in \cXi,\
    \text{\(\cM(f)\) is a predicate on values of sort \(\sigma_0,\dots,\sigma_{n-1}\)} \\[-.4em]
    \dom \cM = \dom \cXi \quad
    \text{for any}\ \varPhi \in \cPhi,\,
    \cM \models \varPhi
  \end{gathered}}{
    \cM \models (\cPhi,\cXi)
  }
\end{gather*}
\endgroup

When \(\cM \models (\cPhi,\cXi)\) holds, we say that \(\cM\) is a \emph{model} of \((\cPhi,\cXi)\).
Every well-sorted CHC system \((\cPhi,\cXi)\) has the \emph{least model} on the point-wise ordering (which can be proved based on the discussions in \cite{HornProgramSemantics}), which we write as \(\cM^\least_{(\cPhi,\cXi)}\).

\subsection{Translation from COR Programs to CHCs}
\label{index:chc-body}

Now we formalize our translation of Rust programs into CHCs.
We define \(\Parened{\varPi}\), which is a CHC system that represents the input-output relations of the functions in the COR program \(\varPi\).

Roughly speaking, the least model \(\cM^\least_{\Parened{\varPi}}\) for this CHC system should satisfy:
for any values \(v_0, \dots, v_{n-1}, w\), \(\cM^\least_{\Parened{\varPi}} \models f_\entry(v_0, \dots, v_{n-1}, w)\) holds exactly if, in COR, a function call \(f(v_0,\dots,v_{n-1})\) can return \(w\).
Actually, in concrete operational semantics, such values should be read out from the heap memory.
The formal description and proof of this expected property is presented in \cref{index:chc-correct}.

\Subsubsection{Auxiliary Definitions}

The sort corresponding to the type \(T\), \(\Parened{T}\), is defined as follows.
\(\check P\) is a meta-variable for a non-mutable-reference pointer kind, i.e. \(\own\) or \(\immut_\alpha\).
Note that the information on lifetimes is all stripped off.
\begingroup\small
\begin{gather*}
  \Parened{X} \defeq X \quad
  \Parened{\mu X.T} = \mu X.\Parened{T} \quad
  \Parened{\check P\, T} \defeq \Box\, \Parened{T} \quad
  \Parened{\mut_\alpha T} \defeq \mut\, \Parened{T} \br[-.1em]
  \Parened{\Int} \defeq \Int \quad
  \Parened{\unit} \defeq \unit \quad
  \Parened{T_0 \!+\! T_1} \defeq \Parened{T_0} + \Parened{T_1} \quad
  \Parened{T_0 \!\times\! T_1} \defeq \Parened{T_0} \times \Parened{T_1}
\end{gather*}
\endgroup

We introduce a special variable \(\res\) to represent the result of a function.\footnote{%
  For simplicity, we assume that the parameters of each function are sorted respecting \emph{some fixed order} on variables (with \(\res\) coming at the last), and we enumerate various items in this fixed order.
}
For a label \(L\) in a function \(f\) in a program \(\varPi\),
we define \(\check\varphi_{\varPi,f,L}\), \(\varXi_{\varPi,f,L}\) and \(\cDelta_{\varPi,f,L}\) as follows,
if the items in the variable context for the label are enumerated as \(x_0\colonu{\ac_0} T_0,\dots,x_{n-1}\colonu{\ac_{n-1}} T_{n-1}\) and the return type of the function is \(U\).
\begin{gather*}
  \check\varphi_{\varPi,f,L} \defeq f_L(x_0,\dots,x_{n-1},\res) \quad
  \varXi_{\varPi,f,L} \defeq (\Parened{T_0},\dots,\Parened{T_{n-1}},\Parened{U}) \br[-.1em]
  \cDelta_{\varPi,f,L} \defeq \{(x_i,\Parened{T_i}) \mid i \in [n]\} + \{(\res,\Parened{U})\}
\end{gather*}
\(\forall(\cDelta)\) stands for \(\forall\,x_0\colon \sigma_0,\, \dots,\,x_{n-1}\colon \sigma_{n-1}\), where the items in \(\cDelta\) are enumerated as \((x_0,\sigma_0),\dots,(x_{n-1},\sigma_{n-1})\).

\Subsubsection{CHC Representation}

Now we introduce `\(\Parened{L\colon S}_{\varPi,f}\)', the set (in most cases, singleton) of CHCs modeling the computation performed by the labeled statement \(L\colon S\) in \(f\) from \(\varPi\).
Unlike informal descriptions in \cref{index:intro}, we turn to \emph{pattern matching} instead of equations, to simplify the proofs\shortorfull{}{ in \cref{index:appx-proof-sldc}}.
Below we show some of the rules; the complete rules are presented in \shortorfull{the full paper}{\cref{index:appx-chc}}.
The variables marked green (e.g. \(\fresh{x_\0}\)) should be fresh.
The following is the rule for mutable (re)borrow.
\begingroup\small
\begin{gather*}
  \begin{aligned}
    & \Parened{L\colon \Let y = \mutbor_\alpha x;\ \goto L'}_{\varPi,f} \\[-.4em]
    & \defeq \begin{cases}
      \left\{\, \begin{aligned}
        & \forall(\cDelta_{\varPi,f,L} \!+\! \{(\fresh{x_\0},\Parened{T})\}). \\[-.1em]
        & \ \check\varphi_{\varPi,f,L}
        \!\impliedby\! \check\varphi_{\varPi,f,L'}[\angled{*x,\fresh{x_\0}}/y,\angled{\fresh{x_\0}}/x]
      \end{aligned} \,\right\}
      & (\Ty_{\varPi,f,L}(x) = \own T) \\[.8em]
      \left\{\, \begin{aligned}
        & \forall(\cDelta_{\varPi,f,L} \!+\! \{(\fresh{x_\0},\Parened{T})\}). \\[-.1em]
        & \ \check\varphi_{\varPi,f,L}
        \!\impliedby\! \check\varphi_{\varPi,f,L'}[\angled{*x,\fresh{x_\0}}/y,\angled{\fresh{x_\0},\0x}/x]
      \end{aligned} \,\right\}
      & (\Ty_{\varPi,f,L}(x) = \mut_\alpha T)
    \end{cases}
  \end{aligned}
\end{gather*}
\endgroup
The value at the end of borrow is represented as a newly introduced variable \(x_\0\).
Below is the rule for release of a variable.
\begingroup\small
\begin{gather*}
  \begin{aligned}
    & \Parened{L\colon \drop x;\ \goto L'}_{\varPi,f} \\[-.4em]
    & \defeq \begin{cases}
      \bigl\{\,
        \forall(\cDelta_{\varPi,f,L}).\
        \check\varphi_{\varPi,f,L} \!\impliedby\! \check\varphi_{\varPi,f,L'}
       \,\bigr\}
      & (\Ty_{\varPi,f,L}(x) = \check P\, T) \\[.3em]
      \left\{\, \begin{aligned}
        & \forall(\cDelta_{\varPi,f,L} \!-\! \{(x, \mut\, \Parened{T})\} \!+\! \{(\fresh{x_*}, \Parened{T})\}). \\[-.1em]
        & \ \check\varphi_{\varPi,f,L}[\angled{\fresh{x_*},\fresh{x_*}}/x] \!\impliedby\! \check\varphi_{\varPi,f,L'}
      \end{aligned} \,\right\}
      & (\Ty_{\varPi,f,L}(x) = \mut_\alpha T)
    \end{cases}
  \end{aligned}
\end{gather*}
\endgroup
When a variable \(x\) of type \(\mut_\alpha T\) is dropped/released, we check the prophesied value at the end of borrow.
Below is the rule for a function call.
\begingroup\small
\begin{gather*}
  \begin{aligned}
    & \Parened{L\colon \Let y = g\angled{{\cdots}}(x_0,\dots,x_{n-1});\ \goto L'}_{\varPi,f} \\[-.3em]
    & \ \defeq\ \{\forall(\cDelta_{\varPi,f,L} \!+\! \{(y,\Parened{\Ty_{\varPi,f,L'}(y)})\}).\ \check\varphi_{\varPi,f,L} \!\impliedby\! g_\entry(x_0,\dots,x_{n-1},y) \land \check\varphi_{\varPi,f,L'}\}
  \end{aligned}
\end{gather*}
\endgroup
The body (the right-hand side of \(\!\impliedby\!\)) of the CHC contains two formulas, which yields a kind of call stack at the level of CHCs.
Below is the rule for a return from a function.
\begingroup\small
\begin{gather*}
  \Parened{L\colon \return x}_{\varPi,f}\ \defeq\ \bigl\{\,
    \forall(\cDelta_{\varPi,f,L}).\
    \check\varphi_{\varPi,f,L}[x/\res] \!\impliedby\! \top
   \,\bigr\}
\end{gather*}
\endgroup
The variable \(\res\) is forced to be equal to the returned variable \(x\).

Finally, \(\Parened{\varPi}\), the CHC system that represents the COR program \(\varPi\) (or the \emph{CHC representation} of \(\varPi\)), is defined as follows.
\begingroup\small
\begin{gather*}
  \Parened{\varPi} \defeq \bigl(
    \textstyle
    \sum_{F\, \text{in}\, \varPi,\,L\colon S\, \in\, \LabelStmt_F}
    \ \Parened{L\colon S}_{\varPi,\name(F)},\
    (\varXi_{\varPi,f,L})_{f_L\, \text{s.t.}\,(f,L)\, \in\, \FnLabel_\varPi}
  \bigr)
\end{gather*}
\endgroup

\begin{example}[CHC Representation]
We present below the CHC representation of \(\takemax\) described in \cref{index:cor-syntax}.
We omit CHCs on \(\incmax\) here.
We have also excluded the variable binders `\(\forall\, \cdots\)'.
\begingroup\small
\begin{gather*}
  \takemax_\entry(\nvar{ma},\var{mb},\res) \impliedby \takemax_{\nL{1}}(\nvar{ma},\var{mb},\angled{*\var{ma} \!\bge\! *\var{mb}},\res) \br[-.1em]
  \takemax_{\nL{1}}(\nvar{ma},\var{mb},\angled{\inj_1 \var{ord_{*!}}},\res) \impliedby \takemax_{\nL{2}}(\nvar{ma},\var{mb},\angled{\var{ord_{*!}}},\res) \br[-.1em]
  \takemax_{\nL{1}}(\nvar{ma},\var{mb},\angled{\inj_0 \var{ord_{*!}}},\res) \impliedby \takemax_{\nL{5}}(\nvar{ma},\var{mb},\angled{\var{ord_{*!}}},\res) \br[-.1em]
  \takemax_{\nL{2}}(\nvar{ma},\var{mb},\var{ou},\res) \impliedby \takemax_{\nL{3}}(\nvar{ma},\var{mb},\res) \br[-.1em]
  \takemax_{\nL{3}}(\nvar{ma},\angled{\nvar{mb}_*,\!\var{mb}_*},\res) \impliedby \takemax_{\nL{4}}(\nvar{ma},\res) \br[-.1em]
  \takemax_{\nL{4}}(\nvar{ma},\var{ma}) \impliedby \top \br[-.1em]
  \takemax_{\nL{5}}(\nvar{ma},\var{mb},\var{ou},\res) \impliedby \takemax_{\nL{6}}(\nvar{ma},\var{mb},\res) \br[-.1em]
  \takemax_{\nL{6}}(\angled{\nvar{ma}_*,\!\var{ma}_*},\var{mb},\res) \impliedby \takemax_{\nL{7}}(\nvar{mb},\res) \br[-.1em]
  \takemax_{\nL{7}}(\nvar{mb},\var{mb}) \impliedby \top
\end{gather*}
\endgroup
The fifth and eighth CHC represent release of \(\var{mb}\)/\(\var{ma}\).
The sixth and ninth CHC represent the determination of the return value \(\res\).
\end{example}

\subsection{Correctness of the CHC Representation}
\label{index:chc-correct}

Now we formally state and prove the correctness of the CHC representation.

\Paragraph{Notations}
We use \(\Braced{\cdots}\) (instead of \(\{\cdots\}\)) for multisets.
\(A \oplus B\) (or more generally \(\bigoplus_\lambda A_\lambda\)) denotes the multiset sum.
For example, \(\Braced{0,1} \oplus \Braced{1} = \Braced{0,1,1} \ne \Braced{0,1}\).

\Subsubsection{Readout and Safe Readout}

We introduce a few judgments to formally describe how read out data from the heap.

First, the judgment `\(\readout_\cH(*a\Colon T \mid v;\, \aM)\)' (the data at the address \(a\) of type \(T\) can be read out from the heap \(\cH\) as the value \(v\), yielding the memory footprint \(\aM\)) is defined as follows.\footnote{%
  Here we can ignore mutable/immutable references, because we focus on what we call \emph{simple} functions, as explained later.
}
Here, a \emph{memory footprint} \(\aM\) is a finite multiset of addresses, which is employed for monitoring the memory usage.
\begingroup\small
\begin{gather*}
  \frac{
    \cH(a) = a' \quad
    \readout_\cH(*a'\Colon T \mid v;\, \aM)
  }{
    \readout_\cH(*a\colon \own T \mid \angled{v};\, \aM \!\oplus\! \Braced{a})
  } \quad
  \frac{
    \readout_\cH(*a\Colon T[\mu X.T/X] \mid v;\, \aM)
  }{
    \readout_\cH(*a\Colon \mu X.T/X \mid v;\, \aM)
  } \br[.1em]
  \frac{
    \cH(a) = n
  }{
    \readout_\cH(*a\Colon \Int \mid n;\, \Braced{a})
  } \quad
  \readout_\cH(*a\Colon \unit \mid ();\, \emp) \br[.1em]
  \frac{\begin{gathered}
    \cH(a) = i \in [2] \quad
    \text{for any}\ k \!\in\! [(\#T_{1 \!-\! i} \!-\! \#T_i)_{\ge 0}],\ \cH(a \!+\! 1 \!+\! \#T_i \!+\! k) = 0 \\[-.3em]
    \readout_\cH(*(a \!+\! 1)\Colon T_i \mid v;\, \aM)
  \end{gathered}}{
    \readout_\cH\bigl(*a\Colon T_0 \!+\! T_1 \mid \inj_i v;\, \aM \!\oplus\! \Braced{a} \!\oplus\! \Braced{a \!+\! 1 \!+\! \#T_i \!+\! k \mid k \!\in\! [(\#T_{1 \!-\! i} \!-\! \#T_i)_{\ge 0}]}\bigr)
  } \br[-.2em]
  (n)_{\ge 0} \defeq \max \{n, 0\} \br[.1em]
  \frac{
    \readout_\cH\bigl(*a\Colon T_0 \mid v_0;\, \aM_0\bigr) \quad
    \readout_\cH\bigl(*(a \!+\! \#T_0)\Colon T_1 \mid v_1;\, \aM_1\bigr)
  }{
    \readout_\cH\bigl(*a\Colon T_0 \!\times\! T_1 \mid (v_0,v_1);\, \aM_0 \!\oplus\! \aM_1)
  }
\end{gather*}
\endgroup
For example, `\(\readout_{\{(100,7),(101,5)\}}(*100\Colon \Int \times \Int \mid (7,5);\, \Braced{100,101})\)' holds.

Next, `\(\readout_\cH(\cF\Colon \cGamma \mid \aF;\, \aM)\)' (the data of the stack frame \(\cF\) respecting the variable context \(\cGamma\) can be read out from \(\cH\) as \(\aF\), yielding \(\aM\)) is defined as follows.
\(\dom \cGamma\) stands for \(\{x \mid x\colonu\ac T \!\in\! \cGamma\}\).
\begingroup\small
\begin{gather*}
  \frac{\begin{gathered}
    \dom \cF = \dom \cGamma \quad
    \text{for any}\ x\colon \own T \in \cGamma,\
    \readout_\cH(*\cF(x)\Colon T \mid v_x;\, \aM_x)
  \end{gathered}}{
    \readout_\cH(\cF\Colon \cGamma \mid \{(x, \angled{v_x}) \!\mid\! x \!\in\! \dom \cF\};\, \textstyle \bigoplus_{x \in \dom \cF} \aM_x)
  }
\end{gather*}
\endgroup

Finally, `\(\safe_\cH(\cF\Colon \cGamma \mid \aF)\)' (the data of \(\cF\) respecting \(\cGamma\) can be \emph{safely} read out from \(\cH\) as \(\aF\)) is defined as follows.
\begingroup\small
\begin{gather*}
  \frac{\begin{gathered}
    \readout_\cH(\cF\Colon \cGamma \mid \aF;\, \aM) \quad
    \aM \ \text{has no duplicate items}
  \end{gathered}}{
    \safe_\cH(\cF\Colon \cGamma \mid \aF)
  }
\end{gather*}
\endgroup
Here, the `no duplicate items' precondition checks the safety on the ownership.

\Subsubsection{COS-based Model}

Now we introduce the \emph{COS-based model} (COS stands for concrete operational semantics) \(f^\COS_\varPi\) to formally describe the expected input-output relation.
Here, for simplicity, \(f\) is restricted to one that does not take lifetime parameters (we call such a function \emph{simple}; the input/output types of a simple function cannot contain references).
We define \(f^\COS_\varPi\) as the predicate (on values of sorts \(\Parened{T_0},\dots,\Parened{T_{n-1}},\Parened{U}\) if \(f\)'s input/output types are \(T_0,\dots,T_{n-1},U\)) given by the following rule.
\begingroup\small
\begin{gather*}
  \frac{\begin{gathered}
    \cC_0 \to_\varPi \cdots \to_\varPi \cC_N \quad
    \final_\varPi(\cC_N) \quad
    \cC_0 = [f,\entry]\, \cF \mid \cH \quad
    \cC_N = [f,L]\, \cF' \mid \cH' \\[-.2em]
    \safe_\cH\bigl(\cF\Colon \cGamma_{\varPi,f,\entry} \bigm| \{(x_i,v_i) \!\mid\! i \!\in\! [n]\}\bigr) \quad
    \safe_{\cH'}\bigl(\cF'\Colon \cGamma_{\varPi,f,L} \bigm| \{(y,w)\}\bigr)
  \end{gathered}}{
    f^\COS_\varPi(v_0,\dots,v_{n-1},w)
  } \br[.1em]
  \text{
    \(\cGamma_{\varPi,f,L}\): the variable context for the label \(L\) of \(f\) in the program \(\varPi\)
  }
\end{gather*}
\endgroup

\Subsubsection{Correctness Theorem}

Finally, the correctness (both soundness and completeness) of the CHC representation is simply stated as follows.
\begin{theorem}[Correctness of the CHC Representation]
\label{theorem:chc-correct}
  For any program \(\varPi\) and simple function \(f\) in \(\varPi\),
  \(f^\COS_\varPi\) is equivalent to \(\cM^\least_{\Parened{\varPi}}(f_\entry)\).
\end{theorem}
\begin{proof}
The details are presented in \shortorfull{the full paper}{\cref{index:appx-proof}}.
We outline the proof below.

First, we introduce \emph{abstract operational semantics}\shortorfull{}{ (\cref{index:appx-proof-aos})}, where we get rid of heaps and directly represent each variable in the program simply as a value with \emph{abstract variables}, which is strongly related to \emph{prophecy variables} (see \cref{index:related}).
An abstract variable represents the undetermined value of a mutable reference at the end of borrow.

Next, we introduce \emph{SLDC resolution}\shortorfull{}{ (\cref{index:appx-proof-sldc})} for CHC systems
and find a \emph{bisimulation} between abstract operational semantics and SLDC resolution\shortorfull{}{ (\cref{lemma:bisim-aos-chc})},
whereby we show that the \emph{AOS-based model}, defined analogously to the COS-based model, is \emph{equivalent} to the least model of the CHC representation\shortorfull{}{ (\cref{theorem:aos-chc-equivalent})}.
Moreover, we find a \emph{bisimulation} between concrete and abstract operational semantics\shortorfull{}{ (\cref{lemma:bisim-cos-aos})}
and prove that the COS-based model is \emph{equivalent} to the AOS-based model\shortorfull{}{ (\cref{theorem:cos-aos-equivalent})}.

Finally, combining the equivalences\shortorfull{}{ of \cref{theorem:aos-chc-equivalent} and \cref{theorem:cos-aos-equivalent}}, we achieve the proof for the correctness of the CHC representation.
\qed\end{proof}

Interestingly, as by-products of the proof, we have also shown the \emph{soundness of the type system} in terms of preservation and progression, in both concrete and abstract operational semantics.
\shortorfull{}{See \cref{index:appx-proof-aos-safe} and \cref{index:appx-proof-aos-chc} for details.}
Simplification and generalization of the proofs is left for future work.

\subsection{Advanced Examples}
\label{index:chc-examples}

We give advanced examples of pointer-manipulating Rust programs and their CHC representations.
For readability, we write programs in Rust (with ghost annotations) instead of COR.
In addition, CHCs are written in an informal style like \cref{index:intro}, preferring equalities to pattern matching.

\begin{example}\label{example:linger-dec}
Consider the following program, a variant of \rusti{just_rec} in \cref{index:intro-challenges}.
\begin{rust}
fn choose<'a>(ma: &'a mut i32, mb: &'a mut i32) -> &'a mut i32 {
  if rand() { @(drop mb;)@ ma } else { @(drop ma;)@ mb }
}
fn linger_dec<'a>(ma: &'a mut i32) -> bool {
  *ma -= 1; if rand() { @(drop ma;)@ return true; }
  let mut b = rand(); let old_b = b; @(intro 'b;)@ let mb = &@('b)@ mut b;
  let r2 = linger_dec@(<'b>)@(choose@(<'b>)@(ma, mb)); @(now 'b;)@
  r2 && old_b >= b
}
\end{rust}
Unlike \rusti{just_rec}, the function \rusti{linger_dec} can modify the local variable of an arbitrarily deep ancestor.
Interestingly, each recursive call to \rusti{linger_dec} can introduce a new lifetime \rusti{@('b)@}, which yields arbitrarily many layers of lifetimes.

Suppose we wish to verify that \rusti{linger_dec} never returns \rusti{false}.
If we use, like \(\JustRecR\) in \cref{index:intro-challenges}, a predicate taking the memory states \(h,h'\) and the stack pointer \(\var{sp}\), we have to discover the quantified invariant: \(\forall\,i \le \var{sp}.\, \select{h}{i} \ge \select{h'}{i}\).
In contrast, our approach reduces this verification problem to the following CHCs:
\begingroup\small
\begin{align*}
  & \Choose(\angled{a,a_\0},\angled{b,b_\0},r) \impliedby b_\0 = b \land r = \angled{a,a_\0} \br[-.2em]
  & \Choose(\angled{a,a_\0},\angled{b,b_\0},r) \impliedby a_\0 = a \land r = \angled{b,b_\0} \br[-.2em]
  & \LingerDec(\angled{a,a_\0}, r) \impliedby a' = a - 1 \land a_\0 = a' \land r = \true \br[-.2em]
  & \begin{aligned}
      \LingerDec(\angled{a,a_\0}, r) \impliedby & a' = a - 1 \land \var{oldb} = b \land \Choose(\angled{a',a_\0}, \angled{b,b_\0}, \var{mc}) \\[-.4em]
      & \land \LingerDec(\nvar{mc},r') \land r = (r' \band \var{oldb} \bge b_\0)
    \end{aligned} \br[-.2em]
  & r = \true \impliedby \LingerDec(\angled{a,a_\0}, r).
\end{align*}
\endgroup
This can be solved by many solvers since it has a very simple model:
\begin{align*}
  \Choose(\angled{a,a_\0},\angled{b,b_\0},r) &\defiff (b_\0 = b \land r = \angled{a,a_\0}) \vee (a_\0 = a \land r = \angled{b,b_\0}) \\[-.1em]
  \LingerDec(\angled{a,a_\0},r) &\defiff r = \true \land a \ge a_\0.
\end{align*}
\end{example}

\begin{example}
Combined with \emph{recursive data structures}, our method turns out to be more interesting.
Let us consider the following Rust code:\footnote{%
  In COR, \rusti{List} can be expressed as \(\mu X. \Int \times \own X + \unit\).
}
\begin{rust}
enum List { Cons(i32, Box<List>), Nil } use List::*;
fn take_some<'a>(mxs: &'a mut List) -> &'a mut i32 {
  match mxs {
    Cons(mx, mxs2) => if rand() { @(drop mxs2;)@ mx }
                           else { @(drop mx;)@ take_some@(<'a>)@(mxs2) }
    Nil => { take_some(mxs) }
  }
}
fn sum(xs: &List) -> i32 {
  match xs { Cons(x, xs2) => x + sum(xs2), Nil => 0 }
}
fn inc_some(mut xs: List) -> bool {
  let n = sum(&xs); @(intro 'a;)@ let my = take_some@(<'a>)@(&@('a)@ mut xs);
  *my += 1; @(drop my;)@ @(now 'a;)@ let m = sum(&xs); m == n + 1
}
\end{rust}
This is a program that manipulates singly linked integer lists, defined as a recursive data type.
\rusti{take_some} takes a mutable reference to a list and returns a mutable reference to some element of the list.
\rusti{sum} calculates the sum of the elements of a list.
\rusti{inc_some} increments some element of a list via a mutable reference and checks that the sum of the elements of the list has increased by \rusti{1}.

Suppose we wish to verify that \rusti{inc_some} never returns \rusti{false}.
Our method translates this verification problem into the following CHCs.\footnote{%
  \(\cons{x}{xs}\) is the cons made of the head \(x\) and the tail \(\var{xs}\).
  \(\nil\) is the nil.
  In our formal logic, they are expressed as \(\inj_0(x,\angled{xs})\) and \(\inj_1 ()\).
}
\begingroup\small
\begin{align*}
  & \TakeSome(\angled{\cons{x}{\var{xs'}}, \var{xs_\0}}, r) \impliedby \var{xs_\0} = \cons{x_\0}{\var{xs'_\0}} \land \var{xs'_\0} = \var{xs'} \land r = \angled{x, x_\0} \br[-.2em]
  & \TakeSome(\angled{\cons{x}{\var{xs'}}, \var{xs_\0}}, r) \impliedby \var{xs_\0} = \cons{x_\0}{\var{xs'_\0}} \land x_\0 = x \land \TakeSome(\angled{\var{xs'}, \var{xs'_\0}}, r) \br[-.2em]
  & \TakeSome(\angled{\nil, \var{xs_\0}}, r) \impliedby \TakeSome(\angled{\nil, \var{xs_\0}}, r) \br[-.0em]
  & \Sum(\angled{\cons{x}{\var{xs'}}}, r) \impliedby \Sum(\angled{\var{xs'}}, r') \land r = x + r' \br[-.2em]
  & \Sum(\angled{\nil}, r) \impliedby r = 0 \br[-.0em]
  & \begin{aligned}
      \IncSome(\var{xs}, r) \impliedby\ & \Sum(\angled{\var{xs}}, n) \land \TakeSome(\angled{\var{xs}, \nvar{xs_\0}}, \angled{y, y_\0}) \land y_\0 = y + 1 \\[-.4em]
      & \land \Sum(\angled{\var{xs}_\0}, m) \land r = (m \beq n \!+\! 1)
    \end{aligned} \br[-.2em]
  & r = \true \impliedby \IncSome(\var{xs}, r)
\end{align*}
\endgroup
A crucial technique used here is \emph{subdivision of a mutable reference}, which is achieved with the constraint \(\var{xs_\0} = \cons{x_\0}{\var{xs'_\0}}\).

We can give this CHC system a very simple model, using an auxiliary function \(\sumf\) (satisfying \(\sumf(\cons{x}{\var{xs'}}) \defeq x + \sumf(\nvar{xs'}),\ \sumf(\nil) \defeq 0\)):
\begin{align*}
  \TakeSome(\angled{\var{xs}, \var{xs_\0}}, \angled{y, y_\0}) &\defiff y_\0 - y = \sumf(\nvar{xs_\0}) - \sumf(\nvar{xs}) \\[-.2em]
  \Sum(\angled{\var{xs}}, r) &\defiff r = \sumf(\nvar{xs}) \\[-.2em]
  \IncSome(\nvar{xs}, r) &\defiff r = \true.
\end{align*}
Although the model relies on the function \(\sumf\), the validity of the model can be checked without induction on \(\sumf\)
(i.e. we can check the validity of each CHC just by properly unfolding the definition of \(\sumf\) a few times).

The example can be \emph{fully automatically and promptly} verified
by our approach using HoIce \cite{HoIce,HoIceMore} as the back-end CHC solver;
see \cref{index:expt}.
\end{example}

\subsection{Discussions}
\label{index:chc-discuss}

We discuss here how our idea can be extended and enhanced.

\Paragraph{Applying Various Verification Techniques}

Our idea can also be expressed as a translation of a pointer-manipulating Rust program into a program of a \emph{stateless functional programming language}, which allows us to use \emph{various verification techniques} not limited to CHCs.
Access to future information can be modeled using \emph{non-determinism}.
To express the value \(a_\0\) coming at the end of mutable borrow in CHCs, we just \emph{randomly guess} the value with non-determinism.
At the time we actually release a mutable reference, we just \emph{check} \ocamli{a' = a} and cut off execution branches that do not pass the check.

For example, \rusti{take_max}/\rusti{inc_max} in \cref{index:intro-ours}/\cref{example:cor-program} can be translated into the following OCaml program.
\begin{ocaml}
let rec assume b = if b then () else assume b
let take_max (a, a') (b, b') =
  if a >= b then (assume (b' = b); (a, a'))
            else (assume (a' = a); (b, b'))
let inc_max a b =
  let a' = Random.int(0) in let b' = Random.int(0) in
  let (c, c') = take_max (a, a') (b, b') in
  assume (c' = c + 1); not (a' = b')
let main a b = assert (inc_max a b)
\end{ocaml}
`\ocamli{let a' = Random.int(0)}' expresses a \emph{random guess} and `\ocamli{assume (a' = a)}' expresses a \emph{check}.
The original problem ``Does \cppi{inc_max} never return \cppi{false}?'' is reduced to the problem ``Does \ocamli{main} never fail at assertion?'' on the OCaml program.\footnote{%
  MoCHi \cite{MoCHiCEGAR}, a higher-order model checker for OCaml, successfully verified the safety property for the OCaml representation above.
  It also successfully and instantly verified a similar representation of \cppi{choose}/\cppi{linger_dec} at \cref{example:linger-dec}.
}

This representation allows us to use various verification techniques, including model checking (higher-order, temporal, bounded, etc.), semi-automated verification (e.g. on Boogie \cite{Boogie}) and verification on proof assistants (e.g. Coq \cite{Coq}).
The property to be verified can be not only partial correctness, but also total correctness and liveness.
Further investigation is left for future work.

\Paragraph{Verifying Higher-order Programs}

We have to care about the following points in modeling closures:
\textbf{(i)} A closure that encloses mutable references can be encoded as a pair of the main function and the `drop function' called when the closure is released;
\textbf{(ii)} A closure that updates enclosed data can be encoded as a function that returns, with the main return value, the updated version of the closure;
\textbf{(iii)} A closure that updates external data through enclosed mutable references can also be modeled by combination of (i) and (ii).
Further investigation on verification of higher-order Rust programs is left for future work.

\Paragraph{Libraries with Unsafe Code}

Our translation does not use lifetime information;
the correctness of our method is guaranteed by the nature of borrow.
Whereas lifetimes are used for \emph{static check} of the borrow discipline,
many libraries in Rust (e.g. \rusti{RefCell}) provide a mechanism for \emph{dynamic ownership check}.

We believe that such libraries with \emph{unsafe code} can be verified for our method by a separation logic such as Iris \cite{Iris,IrisGroundUp}, as RustBelt \cite{RustBelt} does.
A good news is that Iris has recently incorporated \emph{prophecy variables} \cite{FutureOurs}, which seems to fit well with our approach.
This is an interesting topic for future work.

After the libraries are verified, we can turn to our method.
For an easy example, \rusti{Vec} \cite{Vec} can be represented simply as a functional array; a mutable/immutable slice \rusti{&mut[T]/&[T]} can be represented as an array of mutable/immutable references.
For another example, to deal with \rusti{RefCell} \cite{RefCell}, we pass around an \emph{array} that maps a \rusti{RefCell<T>} address to data of type \rusti{T} equipped with an ownership counter;
\rusti{RefCell} itself is modeled simply as an address.\footnote{%
  To borrow a mutable/immutable reference from \rusti{RefCell}, we check and update the counter and take out the data from the array.
}\footnote{%
  In Rust, we can use \rusti{RefCell} to naturally encode data types with circular references (e.g. doubly-linked lists).
}
Importantly, \emph{at the very time we take a mutable reference \(\angled{a,a_\0}\) from a ref-cell, the data at the array should be updated into \(a_\0\)}.
Using methods such as pointer analysis \cite{PointerAnalysis}, we can possibly shrink the array.

Still, our method does not go quite well with \emph{memory leaks} \cite{MemoryLeak} caused for example by combination of \rusti{RefCell} and \rusti{Rc} \cite{Rc}, because they obfuscate the ownership release of mutable references.
We think that use of \rusti{Rc} etc. should rather be restricted for smooth verification.
Further investigation is needed.

\section{Implementation and Evaluation}
\label{index:expt}

We report on the implementation of our verification tool and the preliminary experiments conducted with small benchmarks to confirm the effectiveness of our approach.

\subsection{Implementation of RustHorn}
\label{index:expt-impl}

We implemented a prototype verification tool \emph{RustHorn} (available at \url{https://github.com/hopv/rust-horn}) based on the ideas described above.
The tool supports basic features of Rust supported in COR, including recursions and recursive types especially.

The implementation translates the MIR (Mid-level Intermediate Representation) \cite{MirBlog,MirRustc} of a Rust program into CHCs quite straightforwardly.\footnote{%
  In order to use the MIR, RustHorn's implementation depends on the unstable nightly version of the Rust compiler, which causes a slight portability issue.
}
Thanks to the nature of the translation, RustHorn can just rely on Rust's borrow check and forget about lifetimes.
For efficiency, the predicate variables are constructed by the granularity of the vertices in the control-flow graph in MIR, unlike the per-label construction of \cref{index:chc-body}.
Also, assertions in functions are taken into account unlike the formalization in \cref{index:chc-body}.

\subsection{Benchmarks and Experiments}
\label{index:expt-how}

To measure the performance of RustHorn and the existing CHC-based verifier SeaHorn \cite{SeaHorn}, we conducted preliminary experiments with benchmarks listed in \cref{table:expt}.
Each benchmark program is designed so that the Rust and C versions match.
Each benchmark instance consists of either one program or a pair of safe and unsafe programs that are very similar to each other.
The benchmarks and experimental results are accessible at \url{https://github.com/hopv/rust-horn}.

The benchmarks in the groups \bench{simple} and \bench{bmc} were taken from SeaHorn (\url{https://github.com/seahorn/seahorn/tree/master/test}), with the Rust versions written by us.
They have been chosen based on the following criteria: they
(i) consist of only features supported by core Rust,
(ii) follow Rust's ownership discipline, and (iii) are small enough to be amenable for
manual translation from C to Rust.

The remaining six benchmark groups are built by us and consist of programs featuring mutable references.
The groups \bench{inc-max}, \bench{just-rec} and \bench{linger-dec} are based on the examples that have appeared in \cref{index:intro} and \cref{index:chc-examples}.
The group \bench{swap-dec} consists of programs that perform repeated involved updates via mutable references to mutable references.
The groups \bench{lists} and \bench{trees} feature destructive updates on recursive data structures (lists and trees) via mutable references, with one interesting program of it explained in \cref{index:chc-examples}.

We conducted experiments on a commodity laptop (2.6GHz Intel Core i7 MacBook Pro with 16GB RAM).
First we translated each benchmark program by RustHorn and SeaHorn (version 0.1.0-rc3) \cite{SeaHorn} translate into CHCs in the SMT-LIB 2 format.
Both RustHorn and SeaHorn generated CHCs sufficiently fast (about 0.1 second for each program).
After that, we measured the time of CHC solving by Spacer \cite{Spacer} in Z3 (version 4.8.7) \cite{Z3} and HoIce (version 1.8.1) \cite{HoIce,HoIceMore} for the generated CHCs.
SeaHorn's outputs were not accepted by HoIce, especially because SeaHorn generates CHCs with arrays.
We also made modified versions for some of SeaHorn's CHC outputs, adding constraints on address freshness, to improve accuracy of representations and reduce false alarms.\footnote{%
  For \bench{base/3} and \bench{repeat/3} of \bench{inc-max}, the address-taking parts were already removed, probably by inaccurate pointer analysis.
}

\subsection{Experimental Results}
\label{index:expt-results}

\begin{table}
  \centering
  \scalebox{0.9}{\begin{tabular}{ccc|cc|cc}
    &&& \multicolumn{2}{c}{RustHorn} & \multicolumn{2}{c}{SeaHorn \emph{w/Spacer}} \\[-.25em]
    \emph{Group} & \emph{Instance} & \emph{Property} & \emph{w/Spacer} & \emph{w/HoIce} & \emph{as is} & \emph{modified} \\[-.15em]\hline\hline
    \group[4.8]{simple} & \instance{01} & safe & \bit & \bit & \bit \\[-.25em]
    & \instance{04-recursive} & safe & 0.5 & \TO & 0.8 \\[-.25em]
    & \instance{05-recursive} & unsafe & \bit & \bit & \bit \\[-.25em]
    & \instance{06-loop} & safe & \TO & 0.1 & \TO \\[-.25em]
    & \instance{hhk2008} & safe & \TO & 40.5 & \bit \\[-.25em]
    & \instance{unique-scalar} & unsafe & \bit & \bit & \bit \\[-.15em]\hline
    \group[8.2]{bmc} & \instance[1.8]{1} & safe & 0.2 & \bit & \bit \\[-.25em]
    && unsafe & 0.2 & \bit & \bit \\[-.25em]
    & \instance[1.8]{2} & safe & \TO & 0.1 & \bit \\[-.25em]
    && unsafe & \bit & \bit & \bit \\[-.25em]
    & \instance[1.8]{3} & safe & \bit & \bit & \bit \\[-.25em]
    && unsafe & \bit & \bit & \bit \\[-.25em]
    & \instance[1.8]{diamond-1} & safe & 0.1 & \bit & \bit \\[-.25em]
    && unsafe & \bit & \bit & \bit \\[-.25em]
    & \instance[1.8]{diamond-2} & safe & 0.2 & \bit & \bit \\[-.25em]
    && unsafe & \bit & \bit & \bit \\[-.15em]\hline
    \group[6.7]{inc-max} & \instance[1.8]{base} & safe & \bit & \bit & \FA & \bit \\[-.25em]
    && unsafe & \bit & \bit & \bit & \bit \\[-.25em]
    & \instance[1.8]{base/3} & safe & \bit & \bit & \FA \\[-.25em]
    && unsafe & 0.1 & \bit & \bit \\[-.25em]
    & \instance[1.8]{repeat} & safe & 0.1 & \TO & \FA & 0.1 \\[-.25em]
    && unsafe & \bit & 0.4 & \bit & \bit \\[-.25em]
    & \instance[1.8]{repeat/3} & safe & 0.2 & \TO & \bit \\[-.25em]
    && unsafe & \bit & 1.3 & \bit \\[-.15em]\hline
    \group[6.7]{swap-dec} & \instance[1.8]{base} & safe & \bit & \bit & \FA & \bit \\[-.25em]
    && unsafe & 0.1 & \TO & \bit & \bit \\[-.25em]
    & \instance[1.8]{base/3} & safe & 0.2 & \TO & \FA & \bit \\[-.25em]
    && unsafe & 0.4 & 0.9 & \bit & 0.1 \\[-.25em]
    & \instance[1.8]{exact} & safe & 0.1 & 0.5 & \FA & \TO \\[-.25em]
    && unsafe & \bit & 26.0 & \bit & \bit \\[-.25em]
    & \instance[1.8]{exact/3} & safe & \TO & \TO & \FA & \FA \\[-.25em]
    && unsafe & \bit & 0.4 & \bit & \bit \\[-.15em]\hline
    \group[1.8]{just-rec} & \instance[1.8]{base} & safe & \bit & \bit & \bit \\[-.25em]
    && unsafe & \bit & 0.1 & \bit \\[-.15em] \hline
    \group[6.7]{linger-dec} & \instance[1.8]{base} & safe & \bit & \bit & \FA \\[-.25em]
    && unsafe & \bit & 0.1 & \bit \\[-.25em]
    & \instance[1.8]{base/3} & safe & \bit & \bit & \FA \\[-.25em]
    && unsafe & \bit & 7.0 & \bit \\[-.25em]
    & \instance[1.8]{exact} & safe & \bit & \bit & \FA \\[-.25em]
    && unsafe & \bit & 0.2 & \bit \\[-.25em]
    & \instance[1.8]{exact/3} & safe & \bit & \bit & \FA \\[-.25em]
    && unsafe & \bit & 0.6 & \bit \\[-.15em] \hline
    \group[6.7]{lists} & \instance[1.8]{append} & safe & \TE & \bit & \FA \\[-.25em]
    && unsafe & \TE & 0.2 & 0.1 \\[-.25em]
    & \instance[1.8]{inc-all} & safe & \TE & \bit & \FA \\[-.25em]
    && unsafe & \TE & 0.3 & \bit \\[-.25em]
    & \instance[1.8]{inc-some} & safe & \TE & \bit & \FA \\[-.25em]
    && unsafe & \TE & 0.3 & 0.1 \\[-.25em]
    & \instance[1.8]{inc-some/2} & safe & \TE & \TO & \FA \\[-.25em]
    && unsafe & \TE & 0.3 & 0.4 \\[-.15em] \hline
    \group[6.7]{trees} & \instance[1.8]{append-t} & safe & \TE & \bit & \TO \\[-.25em]
    && unsafe & \TE & 0.3 & 0.1 \\[-.25em]
    & \instance[1.8]{inc-all-t} & safe & \TE & \TO & \TO \\[-.25em]
    && unsafe & \TE & 0.1 & \bit \\[-.25em]
    & \instance[1.8]{inc-some-t} & safe & \TE & \TO & \TO \\[-.25em]
    && unsafe & \TE & 0.3 & 0.1 \\[-.25em]
    & \instance[1.8]{inc-some/2-t} & safe & \TE & \TO & \FA \\[-.25em]
    && unsafe & \TE & 0.4 & 0.1 \\[-.15em] \hline
  \end{tabular}}
  \vspace{.8em}
  \caption{
    Benchmarks and experimental results on RustHorn and SeaHorn, with Spacer/Z3 and HoIce.
    ``\TO'' denotes timeout of 180 seconds;
    ``\FA'' means reporting `unsafe' for a safe program;
    ``\TE'' is a tool error of Spacer, which currently does not deal with recursive types well.
  }
  \label{table:expt}
\end{table}

\Cref{table:expt} shows the results of the experiments.

Interestingly, the combination of RustHorn and HoIce succeeded in verifying many programs with recursive data types (\bench{lists} and \bench{trees}), although it failed at difficult programs.\footnote{%
  For example, \bench{inc-some/2} takes two mutable references in a list and increments on them;
  \bench{inc-all-t} destructively increments all elements in a tree.
}
HoIce, unlike Spacer, can find models defined with primitive recursive functions for recursive data types.\footnote{%
  We used the latest version of HoIce, whose algorithm for recursive types is presented in the full paper of \cite{HoIceMore}.
}

False alarms of SeaHorn for the last six groups are mainly due to problematic approximation of SeaHorn for pointers and heap memories, as discussed in \cref{index:intro-challenges}.
On the modified CHC outputs of SeaHorn, five false alarms were erased and four of them became successful.
For the last four groups, unboundedly many memory cells can be allocated, which imposes a fundamental challenge for SeaHorn's array-based approach as discussed in \cref{index:intro-challenges}.\footnote{%
  We also tried on Spacer \(\JustRecR\), the stack-pointer-based accurate representation of \cppi{just_rec} presented in \cref{index:intro-challenges}, but we got timeout of 180 seconds.
}
The combination of RustHorn and HoIce took a relatively long time or reported timeout for some programs, including unsafe ones, because HoIce is still an unstable tool compared to Spacer;
in general, automated CHC solving can be rather unstable.

\section{Related Work}
\label{index:related}

\Paragraph{CHC-based Verification of Pointer-Manipulating Programs}

SeaHorn \cite{SeaHorn} is a representative existing tool for CHC-based verification of pointer-manipulating programs.
It basically represents the heap memory as an array.
Although some pointer analyses \cite{SeaHornContextSensitive} are used to optimize the array representation of the heap,
their approach suffers from the scalability problem discussed in \cref{index:intro-challenges}, as confirmed by the experiments in \cref{index:expt}.
Still, their approach is quite effective as automated verification, given that many real-world pointer-manipulating programs do not follow Rust-style ownership.

Another approach is taken by JayHorn \cite{JayHorn,JayHornQuantified}, which translates Java programs (possibly using object pointers) to CHCs.
They represent store invariants using special predicates \textit{pull} and \textit{push}.
Although this allows faster reasoning about the heap than the array-based approach, it can suffer from more false alarms.
We conducted a small experiment for JayHorn (0.6-alpha) on some of the benchmarks of \cref{index:expt-how};
unexpectedly, JayHorn reported `\texttt{UNKNOWN}' (instead of `\texttt{SAFE}' or `\texttt{UNSAFE}') for even simple programs such as the programs of the instance \bench{unique-scalar} in \bench{simple} and the instance \bench{basic} in \bench{inc-max}.

\Paragraph{Verification for Rust}

Whereas we have presented the first CHC-based (fully automated) verification method specially designed for Rust-style ownership, there have been a number of studies on other types of verification for Rust.

RustBelt \cite{RustBelt} aims to formally prove high-level safety properties for Rust libraries with unsafe internal implementation, using manual reasoning on the higher-order concurrent separation logic Iris \cite{Iris,IrisGroundUp} on the Coq Proof Assistant \cite{Coq}.
Although their framework is flexible, the automation of the reasoning on the framework is little discussed.
The language design of our COR is affected by their formal calculus \lambdaRust.

Electrolysis \cite{Electrolysis} translates some subset of Rust into a purely functional programming language to manually verify functional correctness on Lean Theorem Prover \cite{Lean}.
Although it clears out pointers to get simple models like our approach,
Electrolysis' applicable scope is quite limited, because it deals with mutable references by \emph{simple static tracking of addresses based on lenses} \cite{Lens},
not supporting even basic use cases such as dynamic selection of mutable references (e.g. \rusti{take_max} in \cref{index:intro-ours}) \cite{ElectrolysisReference}, which our method can easily handle.
Our approach covers \emph{all} usages of pointers of the safe core of Rust as discussed in \cref{index:chc}.

Some serial studies \cite{Rust2Viper,Prusti,RustGTT} conduct (semi-)automated verification on Rust programs using Viper \cite{Viper}, a verification platform based on separation logic with fractional ownership.
This approach can to some extent deal with unsafe code \cite{Rust2Viper} and type traits \cite{RustGTT}.
Astrauskas et al. \cite{Prusti} conduct semi-automated verification (manually providing pre/post-conditions and loop invariants) on many realistic examples.
Because Viper is based on \emph{fractional ownership}, however, their platforms have to use \emph{concrete indexing on the memory} for programs like \rusti{take_max}/\rusti{inc_max}.
In contrast, our idea leverages \emph{borrow-based ownership}, and it can be applied also to semi-automated verification as suggested in \cref{index:chc-discuss}.

Some researches \cite{Crust,SmackRust,NoPanic} employ bounded model checking on Rust programs, especially with unsafe code.
Our method can be applied to bounded model checking as discussed in \cref{index:chc-discuss}.

\Paragraph{Verification using Ownership}

Ownership has been applied to a wide range of verification.
It has been used for detecting race conditions on concurrent programs \cite{OwnershipRaceDeadlock,RaceLinear}
and analyzing the safety of memory allocation \cite{FractionalOwnership}.
Separation logic based on ownership is also studied well \cite{PermissionSeparationLogic,Viper,Iris}.
Some verification platforms \cite{VCC,SpecSharp,Why3Region} support simple ownership.
However, most prior studies on ownership-based verification are based on fractional or counting ownership.
Verification under \emph{borrow-based ownership} like Rust was little studied before our work.

\Paragraph{Prophecy Variables}

Our idea of taking a future value to represent a mutable reference is linked to the notion of \emph{prophecy variables} \cite{RefinementMappings,FineGrained,FutureOurs}.
Jung et al. \cite{FutureOurs} propose a new Hoare-style logic with prophecy variables.
In their logic, prophecy variables are not copyable, which is analogous to uncopyability of mutable references in Rust.
This logic can probably be used for generalizing our idea as suggested in \cref{index:chc-discuss}.

\section{Conclusion}
\label{index:concl}

We have proposed a novel method for CHC-based program verification,
which represents a mutable reference as a pair of values, the current value and the future value at the time of release.
We have formalized the method for a core language of Rust and proved its correctness.
We have implemented a prototype verification tool for a subset of Rust and confirmed the effectiveness of our approach.
We believe that this study establishes the foundation of verification leveraging borrow-based ownership.

\subsubsection{Acknowledgments}

This work was supported by JSPS KAKENHI Grant Number JP15H05706 and JP16K16004.
We are grateful to the anonymous reviewers for insightful comments.

\bibliographystyle{splncs04}
\bibliography{biblio}
\vfill

{\small\medskip\noindent{\bf Open Access} This chapter is licensed under the terms of the Creative Commons\break Attribution 4.0 International License (\url{http://creativecommons.org/licenses/by/4.0/}), which permits use, sharing, adaptation, distribution and reproduction in any medium or format, as long as you give appropriate credit to the original author(s) and the source, provide a link to the Creative Commons license and indicate if changes were made.}

{\small \spaceskip .28em plus .1em minus .1em The images or other third party material in this chapter are included in the chapter's Creative Commons license, unless indicated otherwise in a credit line to the material.~If material is not included in the chapter's Creative Commons license and your intended\break use is not permitted by statutory regulation or exceeds the permitted use, you will need to obtain permission directly from the copyright holder.}

\medskip\noindent\includegraphics{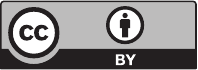}

\iffull
\appendix
\section{Complementary Definitions on COR}

\subsection{Complete Typing Rules for Instructions}
\label{index:appx-cor-typing-instructions}

The following is the complete rules for the typing judgment on instructions \(I\colond{\varPi,f} (\cGamma,\cA) \to (\cGamma',\cA')\).
The variables on the right-hand side of one instruction should be mutually distinct.
The rules for subtyping \(T \le_\cA U\) are explained later.
\begingroup\small
\begin{gather*}
  \frac{
    \alpha \notin A_{\ex\, \varPi,f} \quad
    P = \own, \mut_\beta \quad
    \text{for any}\ \gamma \in \Lifetime_{P\, T},\
    \alpha \le_\cA \gamma
  }{
    \Let y = \mutbor_\alpha x \, \colond{\varPi,f} (\cGamma \!+\! \{x\colon P\, T\},\, \cA) \to (\cGamma \!+\! \{y\colon \mut_\alpha T,\, x\colonu{\dagger\alpha} P\, T\},\, \cA)
  } \br[.3em]
  \frac{
    \text{if \(T\) is of form \(\own U\), every \(\own\) and \(\mut_\alpha\) in \(U\) is guarded by some \(\immut_\beta\)}
  }{
    \drop x \, \colond{\varPi,f} (\cGamma \!+\! \{x\colon T\},\, \cA) \to (\cGamma,\cA)
  } \br[.3em]
  \immut x \, \colond{\varPi,f} (\cGamma \!+\! \{x\colon \mut_\alpha T\},\, \cA) \to (\cGamma \!+\! \{x\colon \immut_\alpha T\},\, \cA) \br[.3em]
  \frac{
    x\colon \mut_\alpha T,\,y\colon P\, T \in \cGamma \quad
    P = \own,\mut_\beta
  }{
    \swap(*x,*y) \, \colond{\varPi,f} (\cGamma,\cA) \to (\cGamma,\cA)
  } \br[.3em]
  \Let *y = x \, \colond{\varPi,f} (\cGamma \!+\! \{x\colon T\},\, \cA) \to (\cGamma \!+\! \{y\colon \own T\},\, \cA) \br[.3em]
  \Let y = *x \, \colond{\varPi,f} (\cGamma \!+\! \{x\colon P\,P'\,T\},\, \cA) \to (\cGamma \!+\! \{y\colon (P \!\circ\! P')\,T\},\, \cA) \br[-.2em]
  \dbox{\(
    P \circ {\own} = {\own} \circ P \defeq P \quad
    R_\alpha \circ R'_\beta \defeq R''_\alpha\
    \text{where}\ R'' = \begin{cases}
      \mut &\!\! (R = R' = \mut) \\[-.3em]
      \immut &\!\! (\text{otherwise})
    \end{cases}
  \)} \br[.6em]
  \frac{
    x\colon P\, T \in \cGamma \quad T\colon {\Copy}
  }{
    \Let *y = \Copy *x \, \colond{\varPi,f} (\cGamma,\cA) \to (\cGamma \!+\! \{y\colon \own T\},\, \cA)
  } \br[-.1em]
  \dbox{\(\displaystyle
    \Int\colon {\Copy} \quad
    \unit\colon {\Copy} \quad
    \immut_\alpha T\colon {\Copy} \quad
    \frac{T\colon {\Copy}}{\mu X.T\colon {\Copy}} \quad
    \frac{T_0, T_1\, \colon {\Copy}}{T_0 \!+\! T_1\, \colon {\Copy}} \quad
    \frac{T_0, T_1\, \colon {\Copy}}{T_0 \!\times\! T_1\, \colon {\Copy}}
  \)} \br[.6em]
  \frac{
    T \le_\cA U
  }{
    x \as U \, \colond{\varPi,f} (\cGamma \!+\! \{x\colon T\},\, \cA) \to (\cGamma \!+\! \{x\colon U\},\, \cA)
  } \br[.3em]
  \frac{\begin{gathered}
    \varSigma_{\varPi,g} = \angled{\alpha'_0,\dots,\alpha'_{m-1} \mid \alpha'_{a_0} \le \alpha'_{b_0}, \dots, \alpha'_{a_{l-1}} \le \alpha'_{b_{l-1}}} (x'_0\colon T'_0, \dots, x'_{n-1}\colon T'_{n-1}) \to T'_n \\[-.3em]
    \text{for any}\ j \in [l],\
    \alpha_{a_j} \le_\cA \alpha_{b_j} \quad
    \text{for any}\ i \in [n \!+\! 1],\
    T_i = T'_i[\alpha_0/\alpha'_0,\dots,\alpha_{m-1}/\alpha'_{m-1}]
  \end{gathered}}{
    \Let y = g\angled{\alpha_0,\dots,\alpha_{m-1}}(x_0,\dots,x_{n-1}) \, \colond{\varPi,f} (\cGamma \!+\! \{x_i\colon T_i \mid i \in [n]\},\, \cA) \to (\cGamma \!+\! \{y\colon T_n\},\, \cA)
  } \br[-.0em]
  \text{
    \(\varSigma_{\varPi,f}\): the function signature of the function \(f\) in \(\varPi\)
  } \br[.3em]
  \intro \alpha \, \colond{\varPi,f} \bigl(\cGamma,(A,R)\bigr) \to \bigl(\cGamma,(\{\alpha\} \!+\! A,\, \{\alpha\} \!\times\! (\{\alpha\} \!+\! A_{\ex\, \varPi,f}) \!+\! R)\bigr) \br[.3em]
  \frac{
    \alpha \notin A_{\ex\, \varPi,f}
  }{
    \now \alpha \, \colond{\varPi,f} \bigl(\cGamma,(\{\alpha\} \!+\! A,\,R)\bigr)
    \to \bigl(\{\thaw_\alpha(x\colonu{\ac} T) \mid x\colonu{\ac} T \!\in\! \cGamma\},\,(A,\{(\beta,\gamma) \!\in\! R \mid \beta \!\ne\! \alpha\})\bigr)
  } \br[-.1em]
  \thaw_\alpha(x\colonu\ac T) \defeq \begin{cases}
    x\colon T &\!\!\! (\ac = \dagger\alpha) \\[-.3em]
    x\colonu\ac T &\!\!\! (\text{otherwise})
  \end{cases} \br[.5em]
  \frac{
    \alpha, \beta \notin A_{\ex\, \varPi,f}
  }{
    \alpha \le \beta \, \colond{\varPi,f} \bigl(\cGamma,(A,R)\bigr) \to \bigl(\cGamma,\,(A,\,(\{(\alpha,\beta)\} \cup R)^+)\bigr)
  } \br[.3em]
  \Let *y = \const \, \colond{\varPi,f} (\cGamma,\cA) \to (\cGamma \!+\! \{y\colon \own T_\const\},\, \cA) \br[-.1em]
  \text{
    \(T_\const\): the type of \(\const\) (\(\Int\) or \(\unit\))
  } \br[.3em]
  \frac{
    x\colon P\, \Int,\,x'\colon P'\,\Int \in \cGamma
  }{
    \Let *y = *x \op *x' \, \colond{\varPi,f} (\cGamma,\cA) \to (\cGamma \!+\! \{y\colon \own T_{\op}\},\, \cA)
  } \br[-.1em]
  \text{
    \(T_{\op}\): the output type of \(\op\) (\(\Int\) or \(\bool\))
  } \br[.3em]
  \Let *y = \rand() \, \colond{\varPi,f} (\cGamma,\cA) \to (\cGamma \!+\! \{y\colon \own \Int\},\, \cA) \br[.2em]
  \Let *y = \inj^{T_0 \!+\! T_1}_i *x \, \colond{\varPi,f} (\cGamma \!+\! \{x\colon \own T_i\},\, \cA) \to (\cGamma \!+\! \{y\colon \own\, (T_0 \!+\! T_1)\},\, \cA) \br[.3em]
  \Let *y = (*x_0,*x_1) \, \colond{\varPi,f} (\cGamma \!+\! \{x_0\colon \own T_0,\ x_1\colon \own T_1\},\, \cA) \to (\cGamma \!+\! \{y\colon \own\, (T_0 \!\times\! T_1)\},\, \cA) \br[.3em]
  \Let\, (*y_0,*y_1) = *x \, \colond{\varPi,f} (\cGamma \!+\! \{x\colon P\,(T_0 \!\times\! T_1)\},\, \cA) \to (\cGamma \!+\! \{y_0\colon P\, T_0,\ y_1\colon P\, T_1\},\, \cA)
\end{gather*}
\endgroup

\Subsubsection{Rule for Drop}

The precondition for the typing rule on \(\drop x\) is just for simplicity on formal definitions.
For concrete operational semantics, a non-guarded \(\own\) within \(\own U\) causes nested releases of memory cells.
For translation to CHCs, a non-guarded \(\mut\) within \(\own U\) would make value checks complicated.

This precondition does not weaken the expressivity, because we can divide pointers by dereference (\(\Let y = *x\)), pair destruction (\(\Let\, (*y_0,*y_1) = *x\)) and variant destruction (\(\match{*x}{{\cdots}}\)) (possibly using loops/recursions, for recursive types).

\Subsubsection{Rule for Swap}

We can omit swap between two owning pointers because it is essentially the same thing with just swapping the names of the pointers.
Note that an active (i.e. not frozen) owning pointer has no other alias at all.

\Subsubsection{Subtyping}

The subtyping judgment \(\Xi \vdash T \!\le_\cA\! U\) is defined as follows.
Here, \(\Xi\) is a set of assumptions of form \(T \!\le\! U\), which is used for subtyping on recursive types.
\(\varnothing \vdash T \!\le_\cA\! U\) can be shortened into \(T \!\le_\cA\! U\).
\begingroup\small
\begin{gather*}
  \frac{T \!\le\! U \in \Xi}{\Xi \vdash T \!\le_\cA\! U} \quad
  \frac{\Xi \vdash T \!\le_\cA\! U}{\Xi \vdash \check P\, T \!\le_\cA\! \check P\,U} \quad
  \frac{\Xi \vdash T \!\le_\cA\! U,\, U \!\le_\cA\! T}{\Xi \vdash \mut_\alpha T \!\le_\cA\! \mut_\alpha U} \quad
  \frac{\Xi \vdash \beta \!\le_\cA\! \alpha}{\Xi \vdash R_\alpha\,T \!\le_\cA\! R_\beta\,T} \br[.1em]
  \frac{\Xi \vdash T_0 \!\le_\cA\! U_0,\, T_1 \!\le_\cA\! U_1}{\Xi \vdash T_0 \!+\! T_1 \!\le_\cA\! U_0 \!+\! U_1} \quad
  \frac{\Xi \vdash T_0 \!\le_\cA\! U_0,\, T_1 \!\le_\cA\! U_1}{\Xi \vdash T_0 \!\times\! T_1 \!\le_\cA\! U_0 \!\times\! U_1} \br[.1em]
  \Xi \vdash \mu X.T \!\le_\cA\! T[\mu X.T/X],\, T[\mu X.T/X] \!\le_\cA\! \mu X.T \br[.1em]
  \frac{\text{\(X',Y'\) are fresh in \(\Xi\)} \quad \Xi + \{X' \!\le\! Y'\} \vdash T[X'/X] \!\le_\cA\! U[Y'/Y]}{\Xi \vdash \mu X.T \!\le_\cA\! \mu Y.U} \br[.1em]
  \frac{\begin{gathered}
    \text{\(X',Y'\) are fresh in \(\Xi\)} \\[-.4em] \Xi + \{X' \!\le\! Y', Y' \!\le\! X'\} \vdash T[X'/X] \!\le_\cA\! U[Y'/Y],\, U[Y'/Y] \!\le_\cA\! T[X'/X]
  \end{gathered}}{\Xi \vdash \mu X.T \!\le_\cA\! \mu Y.U,\, \mu Y.U \!\le_\cA\! \mu X.T} \br[.1em]
  \Xi \vdash T \!\le_\cA\! T \quad
  \frac{\Xi \vdash T \!\le_\cA\! T',\, T' \!\le_\cA\! T''}{\Xi \vdash T \!\le_\cA\! T''}
\end{gather*}
\endgroup

\subsection{Complete Rules and an Example Execution for Concrete Operational Semantics}
\label{index:appx-cor-cos}

The following is the complete rules for the judgments \(\cC \to_\varPi \cC'\) and \(\final_\varPi(\cC)\).
\begingroup\small
\begin{gather*}
  \frac{
    S_{\varPi,f,L} = \Let y = \mutbor_\alpha x;\, \goto L' \quad \cF(x) = a
  }{
    [f,L]\, \cF;\, \cS \mid \cH
    \ \to_\varPi\ [f,L']\, \cF \!+\! \{(y,a)\};\, \cS \mid \cH
  } \br[.3em]
  \frac{
    S_{\varPi,f,L} = \drop x;\, \goto L' \quad
    \Ty_{\varPi,f,L}(x) = \own T
  }{
    [f,L]\, \cF \!+\! \{(x,a)\};\, \cS \mid \cH \!+\! \{(a \!+\! k, n_k) \!\mid\! k \!\in\! [\#T]\}
    \ \to_\varPi\ [f,L']\, \cF;\, \cS \mid \cH
  } \br[.3em]
  \frac{
    S_{\varPi,f,L} = \drop x;\, \goto L' \quad
    \Ty_{\varPi,f,L}(x) = R_\alpha\,T
  }{
    [f,L]\, \cF \!+\! \{(x,a)\};\, \cS \mid \cH
    \ \to_\varPi\ [f,L']\, \cF;\, \cS \mid \cH
  } \br[.3em]
  \frac{
    S_{\varPi,f,L} = \immut x;\, \goto L'
  }{
    [f,L]\, \cF;\, \cS \mid \cH
    \ \to_\varPi\ [f,L']\, \cF;\, \cS \mid \cH
  } \br[.3em]
  \frac{
    S_{\varPi,f,L} = \swap(*x,*y);\, \goto L' \quad
    \Ty_{\varPi,f,L}(x) = P\, T \quad
    \cF(x) = a \quad
    \cF(y) = b
  }{
    \begin{aligned}
      & [f,L]\, \cF;\, \cS \mid \cH \!+\! \{(a \!+\! k, m_k) \!\mid\! k \!\in\! [\#T]\} \!+\! \{(b \!+\! k, n_k) \!\mid\! k \!\in\! [\#T]\} \\[-.3em]
      & \ \to_\varPi\ [f,L']\, \cF;\, \cS \mid \cH \!+\! \{(a \!+\! k, n_k) \!\mid\! k \!\in\! [\#T]\} \!+\! \{(b \!+\! k, m_k) \!\mid\! k \!\in\! [\#T]\}
    \end{aligned}
  } \br[.3em]
  \frac{
    S_{\varPi,f,L} = \Let *y = x;\, \goto L'
  }{
    [f,L]\, \cF \!+\! \{(x,a')\};\, \cS \mid \cH
    \ \to_\varPi\ [f,L']\, \cF \!+\! \{(y,a)\};\, \cS \mid \cH \!+\! \{(a,a')\}
  } \br[.3em]
  \frac{
    S_{\varPi,f,L} = \Let y = *x;\, \goto L' \quad
    \Ty_{\varPi,f,L}(x) = \own P\, T
  }{
    [f,L]\, \cF \!+\! \{(x,a)\};\, \cS \mid \cH \!+\! \{(a,a')\}
    \ \to_\varPi\ [f,L']\, \cF \!+\! \{(y,a')\};\, \cS \mid \cH
  } \br[.3em]
  \frac{
    S_{\varPi,f,L} = \Let y = *x;\, \goto L' \quad
    \Ty_{\varPi,f,L}(x) = R_\alpha\,P\, T \quad
    \cH(a) = a'
  }{
    [f,L]\, \cF \!+\! \{(x,a)\};\, \cS \mid \cH
    \ \to_\varPi\ [f,L']\, \cF \!+\! \{(y,a')\};\, \cS \mid \cH
  } \br[.3em]
  \frac{
    S_{\varPi,f,L} = \Let *y = \Copy *x;\, \goto L' \quad
    \Ty_{\varPi,f,L}(x) = P\, T \quad
    \cF(x) = a
  }{
    [f,L]\, \cF;\, \cS \mid \cH
    \ \to_\varPi\ [f,L']\, \cF \!+\! \{(y,b)\};\, \cS \mid \cH \!+\! \{(b \!+\! k,\cH(a \!+\! k)) \!\mid\! k \!\in\! [\#T]\}
  } \br[.3em]
  \frac{
    S_{\varPi,f,L} = I;\, \goto L' \quad
    I = x \as T,\, \intro \alpha,\, \now \alpha,\, \alpha \le \beta
  }{
    [f,L]\, \cF;\, \cS \mid \cH
    \ \to_\varPi\ [f,L']\, \cF;\, \cS \mid \cH
  } \br[.3em]
  \frac{\begin{gathered}
    S_{\varPi,f,L} = \Let y = g\angled{\cdots}(x_0,\dots,x_{n-1});\, \goto L' \\[-.2em]
    \varSigma_{\varPi,g} = \angled{\cdots}(x'_0\colon T_0, \dots, x'_{n-1}\colon T_{n-1}) \to U
  \end{gathered}}{
    [f,L]\, \cF \!+\! \{(x_i,a_i) \!\mid\! i \!\in\! [n]\};\, \cS \mid \cH \to_\varPi\ [g,\entry]\, \{(x'_i,a_i) \!\mid\! i \!\in\! [n]\};\,
    [f,L]\,y,\cF;\, \cS \mid \cH
  } \br[.3em]
  \frac{
    S_{\varPi,f,L} = \return x
  }{
    [f,L]\, \{(x,a)\}; [g,L']\, x',\cF'; \cS \mid \cH
    \to_\varPi [g,L']\, \cF' \!+\! \{(x',a)\}; \cS \mid \cH
  } \br[.3em]
  \frac{
    S_{\varPi,f,L} = \return x \quad
  }{
    \final_\varPi\bigl([f,L]\, \{(x,a)\} \mid \cH\bigr)
  } \br[.3em]
  \frac{
    S_{\varPi,f,L} = \Let *y = \const;\, \goto L' \quad
    \cH' = \begin{cases}
      \{(a,n)\} &\!\! (\const = n) \\[-.3em]
      \emp &\!\! (\const = ())
    \end{cases}
  }{
    [f,L]\, \cF;\, \cS \mid \cH
    \ \to_\varPi\ [f,L']\, \cF \!+\! \{(y,a)\};\, \cS \mid \cH \!+\! \cH'
  } \br[.3em]
  \frac{
    S_{\varPi,f,L} = \Let *y = *x \op *x';\, \goto L' \quad
    \cF(x) = a \quad
    \cF(x') = a'
  }{
    [f,L]\, \cF;\, \cS \mid \cH
    \ \to_\varPi\ [f,L']\, \cF \!+\! \{(y,b)\};\, \cS \mid \cH \!+\! \{(b,\, \cH(a) \mathop{\angled{\op}} \cH(a'))\}
  } \br[-.0em]
  \text{
    \(\angled{\op}\): \(\op\) as a binary operation on integers, with \(\true\)/\(\false\) encoded as \(1\)/\(0\)
  } \br[.3em]
  \frac{
    S_{\varPi,f,L} = \Let *y = \rand();\, \goto L'
  }{
    [f,L]\, \cF;\, \cS \mid \cH
    \ \to_\varPi\ [f,L']\, \cF \!+\! \{(y,a)\};\, \cS \mid \cH \!+\! \{(a,n)\}
  } \br[.3em]
  \frac{
    S_{\varPi,f,L} = \Let *y = \inj^{T_0 \!+\! T_1}_i *x;\, \goto L' \quad
    \cH_0 = \{(a' \!+\! 1 \!+\! \#T_i \!+\! k,\,0) \mid k \!\in\! [(\#T_{1 \!-\! i} \!-\! \#T_i)_{\ge 0}]\}
  }{\begin{aligned}
    & [f,L]\, \cF \!+\! \{(x,a)\};\, \cS \mid \cH \!+\! \{(a \!+\! k, m_k) \!\mid\! k \!\in\! [\#T_i]\} \\[-.3em]
    & \ \to_\varPi\ [f,L']\, \cF \!+\! \{(y,a')\};\, \cS \mid
    \cH \!+\! \{(a',i)\} \!+\! \{(a' \!+\! 1 \!+\! k, m_k) \!\mid\! k \!\in\! [\#T_i]\} \!+\! \cH_0
  \end{aligned}} \br[.3em]
  \frac{\begin{gathered}
    S_{\varPi,f,L} = \match{*x}{\inj_0 *y_0 \to \goto L'_0,\ \inj_1 *y_1 \to \goto L'_1} \\[-.3em]
    \Ty_{\varPi,f,L}(x) = \own\, (T_0 \!+\! T_1) \quad
    i \in [2] \quad
    \cH_0 = \{(a \!+\! 1 \!+\! \#T_i \!+\! k,\,0) \mid k \in [(\#T_{1 \!-\! i} \!-\! \#T_i)_{\ge 0}]\}
  \end{gathered}}{\begin{aligned}
    & [f,L]\, \cF \!+\! \{(x,a)\};\, \cS \mid
    \cH \!+\! \{(a,i)\} \!+\! \{(a \!+\! 1 \!+\! k, m_k) \mid k \!\in\! [\#T_i]\} \!+\! \cH_0 \\[-.3em]
    & \ \to_\varPi\ [f,L'_i]\, \cF \!+\! \{(y_i, a \!+\! 1)\};\, \cS \mid \cH \!+\! \{(a \!+\! 1 \!+\! k, m_k) \mid k \!\in\! [\#T_i]\}
  \end{aligned}} \br[.3em]
  \frac{\begin{gathered}
    S_{\varPi,f,L} = \match{*x}{\inj_0 *y_0 \to \goto L'_0,\ \inj_1 *y_1 \to \goto L'_1} \\[-.1em]
    \Ty_{\varPi,f,L}(x) = R_\alpha\,(T_0 \!+\! T_1) \quad
    \cH(a) = i \in [2]
  \end{gathered}}{
    [f,L]\, \cF \!+\! \{(x,a)\};\, \cS \mid \cH
    \ \to_\varPi\ [f,L'_i]\, \cF \!+\! \{(y_i, a \!+\! 1)\};\, \cS \mid \cH
  } \br[.3em]
  \frac{
    S_{\varPi,f,L} = \Let *y = (*x_0,*x_1);\, \goto L' \quad
    \text{for each}\,i \in [2],\
    \Ty_{\varPi,f,L}(x_i) = \own T_i
  }{\begin{aligned}
    & [f,L]\, \cF \!+\! \{(x_0,a_0),(x_1,a_1)\};\, \cS \mid \cH \!+\! \{(a_i \!+\! k, m_{ik}) \!\mid\! i \!\in\! [2], k \!\in\! [\#T_i]\} \\[-.4em]
    & \ \to_\varPi\ [f,L']\, \cF \!+\! \{(y,a')\};\, \cS \mid \cH \!+\! \{(a' \!+\! i \#T_0 \!+\! k,\, m_{ik}) \!\mid\! i \!\in\! [2], k \!\in\! [\#T_i]\}
  \end{aligned}} \br[.3em]
  \frac{
    S_{\varPi,f,L} = \Let\, (*y_0,*y_1) = *x;\, \goto L' \quad
    \Ty_{\varPi,f,L}(x) = P\,(T_0 \!\times\! T_1)
  }{
    [f,L]\, \cF \!+\! \{(x,a)\};\, \cS \mid \cH
    \ \to_\varPi\ [f,L']\, \cF \!+\! \{(y_0,a),(y_1,a \!+\! \#T_0)\};\, \cS \mid \cH
  }
\end{gather*}
\endgroup

\begin{example}[Execution on Concrete Operational Semantics]\label{example:cos-execution}
The following is an example execution for the COR program of \cref{example:cor-program}.
\(\spade, \heart, \dia, \club\) represent some distinct addresses (e.g. \(100,101,102,103\)).
\(\to_\varPi\) is abbreviated as \(\to\).
\begingroup\small
\begin{align*}
  &
    [\incmax,\entry]\, \{(\var{oa},\spade),(\var{ob},\heart)\} \mid \{(\spade,4),(\heart,3)\}
  \br[-.2em] &
    \to [\incmax,\nL{1}]\, \{(\var{oa},\spade),(\var{ob},\heart)\} \mid \{(\spade,4),(\heart,3)\}
  \br[-.2em] &
    \to^+ [\incmax,\nL{3}]\, \{(\var{ma},\spade),(\var{mb},\heart),(\var{oa},\spade),(\var{ob},\heart)\} \mid \{(\spade,4),(\heart,3)\}
  \br[-.2em] &
    \begin{aligned}
      & \to [\takemax,\entry]\, \{(\var{ma},\spade),(\var{mb},\heart)\}; \\[-.3em]
      & \hspace{6em} [\incmax,\nL{4}]\, \var{mc},\{(\var{oa},\spade),(\var{ob},\heart)\} \mid \{(\spade,4),(\heart,3)\}
    \end{aligned}
  \br[-.2em] &
    \begin{aligned}
      & \to [\takemax,\nL{1}]\, \{(\var{ord},\dia),(\var{ma},\spade),(\var{mb},\heart)\}; \\[-.3em]
      & \hspace{6em} [\incmax,\nL{4}]\, \var{mc},\!\{(\var{oa},\spade),(\var{ob},\heart)\} \mid \{(\spade,4),(\heart,3),(\dia,1)\}
    \end{aligned}
  \br[-.2em] &
    \begin{aligned}
      & \to [\takemax,\nL{2}]\, \{(\var{ou},\dia\!+\!1),(\var{ma},\spade),(\var{mb},\heart)\}; \\[-.3em]
      & \hspace{6em} [\incmax,\nL{4}]\, \var{mc},\{(\var{oa},\spade),(\var{ob},\heart)\} \mid \{(\spade,4),(\heart,3)\}
    \end{aligned}
  \br[-.2em] &
    \begin{aligned}
      & \to^+ [\takemax,\nL{4}]\, \{(\var{ma},\spade)\}; \\[-.3em]
      & \hspace{7em} [\incmax,\nL{4}]\, \var{mc},\{(\var{oa},\spade),(\var{ob},\heart)\} \mid \{(\spade,4),(\heart,3)\}
    \end{aligned}
  \br[-.2em] &
    \to [\incmax,\nL{4}]\, \{(\var{mc},\spade),(\var{oa},\spade),(\var{ob},\heart)\} \mid \{(\spade,4),(\heart,3)\}
  \br[-.2em] &
    \to [\incmax,\nL{5}]\, \{(\var{o1},\dia),(\var{mc},\spade),(\var{oa},\spade),(\var{ob},\heart)\} \mid \{(\spade,4),(\heart,3),(\dia,1)\}
  \br[-.2em] &
    \to^+ [\incmax,\nL{7}]\, \{(\var{oc'},\club),(\var{mc},\spade),(\var{oa},\spade),(\var{ob},\heart)\} \mid \{(\spade,4),(\heart,3),(\club,5)\}
  \br[-.2em] &
    \to [\incmax,\nL{8}]\, \{(\var{oc'},\club),(\var{mc},\spade),(\var{oa},\spade),(\var{ob},\heart)\} \mid \{(\spade,5),(\heart,3),(\club,4)\}
  \br[-.2em] &
    \to^+ [\incmax,\nL{10}]\, \{(\var{oa},\spade),(\var{ob},\heart)\} \mid \{(\spade,5),(\heart,3)\}
  \br[-.2em] &
    \to [\incmax,\nL{11}]\, \{(\var{oa},\spade),(\var{ob},\heart)\} \mid \{(\spade,5),(\heart,3)\}
  \br[-.2em] &
    \to^+ [\incmax,\nL{14}]\, \{(\var{ores},\dia)\} \mid \{(\dia,1)\}
\end{align*}
\endgroup

The execution is quite straightforward.
Recall that every variable is a pointer and holds just an address.
Most of the data is stored in the heap.
\end{example}

\section{Complete Rules for Translation from Labeled Statements to CHCs}
\label{index:appx-chc}

We present below the complete rules for \(\Parened{L\colon S}_{\varPi,f}\).
\begingroup\small
\begin{gather*}
  \begin{aligned}
    & \Parened{L\colon \Let y = \mutbor_\alpha x;\ \goto L'}_{\varPi,f} \\[-.4em]
    & \defeq \begin{cases}
      \left\{\, \begin{aligned}
        & \forall(\cDelta_{\varPi,f,L} \!+\! \{(\fresh{x_\0},\Parened{T})\}). \\[-.1em]
        & \ \check\varphi_{\varPi,f,L}
        \!\impliedby\! \check\varphi_{\varPi,f,L'}[\angled{*x,\fresh{x_\0}}/y,\angled{\fresh{x_\0}}/x]
      \end{aligned} \,\right\}
      & (\Ty_{\varPi,f,L}(x) = \own T) \\[.8em]
      \left\{\, \begin{aligned}
        & \forall(\cDelta_{\varPi,f,L} \!+\! \{(\fresh{x_\0},\Parened{T})\}). \\[-.1em]
        & \ \check\varphi_{\varPi,f,L}
        \!\impliedby\! \check\varphi_{\varPi,f,L'}[\angled{*x,\fresh{x_\0}}/y,\angled{\fresh{x_\0},\0x}/x]
      \end{aligned} \,\right\}
      & (\Ty_{\varPi,f,L}(x) = \mut_\alpha T)
    \end{cases}
  \end{aligned} \br[.3em]
  \begin{aligned}
    & \Parened{L\colon \drop x;\ \goto L'}_{\varPi,f} \\[-.4em]
    & \defeq \begin{cases}
      \bigl\{\,
        \forall(\cDelta_{\varPi,f,L}).\
        \check\varphi_{\varPi,f,L} \!\impliedby\! \check\varphi_{\varPi,f,L'}
       \,\bigr\}
      & (\Ty_{\varPi,f,L}(x) = \check P\, T) \\[.3em]
      \left\{\, \begin{aligned}
        & \forall(\cDelta_{\varPi,f,L} \!-\! \{(x, \mut\, \Parened{T})\} \!+\! \{(\fresh{x_*}, \Parened{T})\}). \\[-.1em]
        & \ \check\varphi_{\varPi,f,L}[\angled{\fresh{x_*},\fresh{x_*}}/x] \!\impliedby\! \check\varphi_{\varPi,f,L'}
      \end{aligned} \,\right\}
      & (\Ty_{\varPi,f,L}(x) = \mut_\alpha T)
    \end{cases}
  \end{aligned} \br[.3em]
  \begin{aligned}
    & \Parened{L\colon \immut x;\ \goto L'}_{\varPi,f} \\[-.2em]
    & \defeq\ \left\{\, \begin{aligned}
      & \forall(\cDelta_{\varPi,f,L} \!-\! \{(x, \mut\ \Parened{T})\} \!+\! \{(\fresh{x_*}, \Parened{T})\}). \\[-.3em]
      & \ \check\varphi_{\varPi,f,L}[\angled{\fresh{x_*},\fresh{x_*}}/x] \!\impliedby\! \check\varphi_{\varPi,f,L'}[\angled{\fresh{x_*}}/x]
    \end{aligned} \,\right\}
    \quad (\Ty_{\varPi,f,L}(x) = \mut_\alpha T)
  \end{aligned} \br[.3em]
  \begin{aligned}
    & \Parened{L\colon \swap(*x,*y);\ \goto L'}_{\varPi,f} \\[-.5em]
    & \defeq \begin{cases}
      \{\,
        \forall(\cDelta_{\varPi,f,L}).\
        \check\varphi_{\varPi,f,L} \!\impliedby\! \check\varphi_{\varPi,f,L'}[\angled{*y,\0x}/x,\angled{*x}/y]
      \, \}
      & (\Ty_{\varPi,f,L}(y) = \own T) \\[-.2em]
      \bigl\{\,
        \forall(\cDelta_{\varPi,f,L}).\
        \check\varphi_{\varPi,f,L} \!\impliedby\! \check\varphi_{\varPi,f,L'}[\angled{*y,\0x}/x,\angled{*x,\0y}/y]
       \,\bigr\}
      & (\Ty_{\varPi,f,L}(y) = \mut_\alpha T)
    \end{cases}
  \end{aligned} \br[.3em]
  \Parened{L\colon \Let *y = x;\ \goto L'}_{\varPi,f}
  \defeq\ \bigl\{\,
    \forall(\cDelta_{\varPi,f,L}).\
    \check\varphi_{\varPi,f,L} \!\impliedby\! \check\varphi_{\varPi,f,L'}[\angled{x}/y]
   \,\bigr\} \br[.3em]
  \begin{aligned}
    & \Parened{L\colon \Let y = *x;\ \goto L'}_{\varPi,f} \\[-.4em]
    & \defeq \begin{cases}
      \bigl\{\,
        \forall(\cDelta_{\varPi,f,L}).\
        \check\varphi_{\varPi,f,L} \!\impliedby\! \check\varphi_{\varPi,f,L'}[*x/y]
       \,\bigr\}
      & \!\!\!(\Ty_{\varPi,f,L}(x) = \own P\, T) \\[.4em]
      \bigl\{\,
        \forall(\cDelta_{\varPi,f,L}).\
        \check\varphi_{\varPi,f,L} \!\impliedby\! \check\varphi_{\varPi,f,L'}[\angled{{**}x}/y]
       \,\bigr\}
      & \!\!\!(\Ty_{\varPi,f,L}(x) = \immut_\alpha P\, T) \\[.4em]
      \{\,
        \forall(\cDelta_{\varPi,f,L}).\
        \check\varphi_{\varPi,f,L} \!\impliedby\! \check\varphi_{\varPi,f,L'}[\angled{{**}x,*\0x}/y]
      \,\}
      & \!\!\!(\Ty_{\varPi,f,L}(x) = \mut_\alpha \own T) \\[.8em]
      \left\{\, \begin{aligned}
        & \forall(\cDelta_{\varPi,f,L} \!-\! \{(x, \mut \Box\, \Parened{T})\} \!+\! \{(\fresh{x_*},\Box\, \Parened{T})\}). \\
        & \ \ \check\varphi_{\varPi,f,L}[\angled{\fresh{x_*},\fresh{x_*}}/x] \!\impliedby\! \check\varphi_{\varPi,f,L'}[\fresh{x_*}/y]
      \end{aligned} \,\right\}
      & \!\!\!(\Ty_{\varPi,f,L}(x) = \mut_\alpha \immut_\beta T) \\[1.3em]
      \left\{\, \begin{aligned}
        & \forall(\cDelta_{\varPi,f,L} \!-\! \{(x,\mut \mut\, \Parened{T})\} \\[-.1em]
        & \quad \!+\! \{(\fresh{x_{**}},\Parened{T}), (\fresh{x_{*\0}},\Parened{T}), (\fresh{x_{\0*}},\Parened{T})\}). \\
        & \ \ \check\varphi_{\varPi,f,L}[\angled{\angled{\fresh{x_{**}}, \fresh{x_{*\0}}}, \angled{\fresh{x_{\0*}}, \fresh{x_{*\0}}}}/x] \\[-.1em]
        & \qquad \!\impliedby\! \check\varphi_{\varPi,f,L'}[\angled{\fresh{x_{**}},\fresh{x_{\0*}}}/y]
      \end{aligned} \,\right\}
      & \!\!\!(\Ty_{\varPi,f,L}(x) = \mut_\alpha \mut_\beta T)
    \end{cases}
  \end{aligned} \br[.3em]
  \Parened{L\colon \Let *y = \Copy *x;\ \goto L'}_{\varPi,f}
  \ \defeq\ \bigl\{\,
    \forall(\cDelta_{\varPi,f,L}).\
    \check\varphi_{\varPi,f,L} \!\impliedby\! \check\varphi_{\varPi,f,L'}[\angled{*x}/y]
   \,\bigr\} \br[.3em]
  \Parened{L\colon x \as T;\ \goto L'}_{\varPi,f}
  \ \defeq\ \bigl\{\,
    \forall(\cDelta_{\varPi,f,L}).\
    \check\varphi_{\varPi,f,L} \!\impliedby\! \check\varphi_{\varPi,f,L'}
   \,\bigr\} \br[.3em]
  \begin{aligned}
    & \Parened{L\colon \Let y = g\angled{{\cdots}}(x_0,\dots,x_{n-1});\ \goto L'}_{\varPi,f} \\[-.3em]
    & \ \defeq\ \{\forall(\cDelta_{\varPi,f,L} \!+\! \{(y,\Parened{\Ty_{\varPi,f,L'}(y)})\}).\ \check\varphi_{\varPi,f,L} \!\impliedby\! g_\entry(x_0,\dots,x_{n-1},y) \land \check\varphi_{\varPi,f,L'}\}
  \end{aligned} \br[.3em]
  \Parened{L\colon \return x}_{\varPi,f}\ \defeq\ \bigl\{\,
    \forall(\cDelta_{\varPi,f,L}).\
    \check\varphi_{\varPi,f,L}[x/\res] \!\impliedby\! \top
   \,\bigr\} \br[.3em]
  \begin{aligned}
    & \Parened{L\colon \intro \alpha;\ \goto L'}_{\varPi,f} =
    \Parened{L\colon \now \alpha;\ \goto L'}_{\varPi,f} =
    \Parened{L\colon \alpha \le \beta;\ \goto L'}_{\varPi,f} \\[-.4em]
    & \ \defeq\ \bigl\{\,
      \forall(\cDelta_{\varPi,f,L}).\
      \check\varphi_{\varPi,f,L} \!\impliedby\! \check\varphi_{\varPi,f,L'}
     \,\bigr\}
  \end{aligned} \br[.3em]
  \Parened{L\colon \Let *y = \const;\ \goto L'}_{\varPi,f}\ \defeq\ \bigl\{\,
    \forall(\cDelta_{\varPi,f,L}).\
    \check\varphi_{\varPi,f,L} \!\impliedby\! \check\varphi_{\varPi,f,L'}[\angled{\const}/y]
   \,\bigr\} \br[.3em]
  \Parened{L\colon \Let *y = *x \op *x';\ \goto L'}_{\varPi,f}\ \defeq\ \bigl\{\,
    \forall(\cDelta_{\varPi,f,L}).\
    \check\varphi_{\varPi,f,L} \!\impliedby\! \check\varphi_{\varPi,f,L'}[\angled{*x \op *x'}/y]
   \,\bigr\} \br[.3em]
  \Parened{L\colon \Let *y = \rand();\ \goto L'}_{\varPi,f}\ \defeq\ \bigl\{\,
    \forall(\cDelta_{\varPi,f,L'}).\
    \check\varphi_{\varPi,f,L} \!\impliedby\! \check\varphi_{\varPi,f,L'}
   \,\bigr\} \br[.3em]
  \Parened{L\colon \Let *y = \inj^{T_0 \!+\! T_1}_i *x;\ \goto L'}_{\varPi,f}\ \defeq\ \bigl\{\,
    \forall(\cDelta_{\varPi,f,L}).\
    \check\varphi_{\varPi,f,L} \!\impliedby\! \check\varphi_{\varPi,f,L'}[\angled{\inj_i *x}/y]
   \,\bigr\} \br[.3em]
  \begin{aligned}
    & \Parened{L\colon \match{*x}{\inj_0 *y_0 \to \goto L_0,\ \inj_1 *y_1 \to \goto L_1}}_{\varPi,f}\ \ \\[-.3em]
    & \ \defeq\ \left\{\,\begin{aligned}
      & \forall(\cDelta_{\varPi,f,L_i} - \{(x, \mut (\Parened{T_0} \!+\! \Parened{T_1}))\} + \{(x_{*!}, \Parened{T_i})\}). \\[-.1em]
      & \hspace{7em} \check\varphi_{\varPi,f,L}[\angled{\inj_i x_{*!}} / x] \!\impliedby\! \check\varphi_{\varPi,f,L_i}[\angled{x_{*!}} / y_i]
    \end{aligned}
    \ \middle|\
      i \in [2]
    \right\} \\
    & \ \text{if}\ \ \Ty_{\varPi,f,L}(x) = \check P\, (T_0 \!+\! T_1)
  \end{aligned} \br[.3em]
  \begin{aligned}
    & \Parened{L\colon \match{*x}{\inj_0 *y_0 \to \goto L_0,\ \inj_1 *y_1 \to \goto L_1}}_{\varPi,f}\ \ \\[-.3em]
    & \ \defeq\ \left\{\,\begin{aligned}
      & \forall(\cDelta_{\varPi,f,L_i} - \{(x, \mut (\Parened{T_0} \!+\! \Parened{T_1}))\} + \{(x_{*!}, \Parened{T_i}), (x_{\0!}, \Parened{T_i})\}). \\[-.1em]
      & \hspace{7em} \check\varphi_{\varPi,f,L}[\angled{\inj_i x_{*!}, \inj_i x_{\0i}} / x] \!\impliedby\! \check\varphi_{\varPi,f,L_i}[\angled{x_{*!}, x_{\0!}} / y_i]
    \end{aligned}
    \ \middle|\
      i \in [2]
    \right\} \\
    & \ \text{if}\ \ \Ty_{\varPi,f,L}(x) = \mut_\alpha (T_0 \!+\! T_1)
  \end{aligned} \br[.3em]
  \begin{aligned}
    & \Parened{L\colon \Let *y = (*x_0,*x_1);\ \goto L'}_{\varPi,f} \\[-.3em]
    & \ \defeq\ \bigl\{\,
      \forall(\cDelta_{\varPi,f,L}).\
      \check\varphi_{\varPi,f,L} \!\impliedby\! \check\varphi_{\varPi,f,L'}[\angled{(*x_0,*x_1)}/y]
     \,\bigr\}
  \end{aligned} \br[.3em]
  \begin{aligned}
    & \Parened{L\colon \Let\, (*y_0,*y_1) = *x;\ \goto L'}_{\varPi,f} \\[-.3em]
    & \defeq \begin{cases}
      \left\{\, \begin{aligned}
        & \forall(\cDelta_{\varPi,f,L}).\
        \check\varphi_{\varPi,f,L} \\
        & \quad \!\impliedby\! \check\varphi_{\varPi,f,L'}[\angled{(*x).0}/y_0, \angled{(*x).1}/y_1]
      \end{aligned} \,\right\}
      & (\Ty_{\varPi,f,L}(x) = \check P\, T) \\[.7em]
      \left\{\, \begin{aligned}
        & \forall(\cDelta_{\varPi,f,L}).\
        \check\varphi_{\varPi,f,L} \!\impliedby\! \\
        & \quad \check\varphi_{\varPi,f,L'}[\angled{(*x).0,(\0x).0}/y_0, \angled{(*x).1,(\0x).1}/y_1]
      \end{aligned} \,\right\}
      & (\Ty_{\varPi,f,L}(x) = \mut_\alpha T)
    \end{cases}
  \end{aligned}
\end{gather*}
\endgroup

\Subsubsection{Rule for Dereference}

The rule for dereference (\(\Let y = *x\)) may seem complicated at a glance. It is however just because this single instruction can cause multiple events (dereference and release of a mutable reference).

\section{Proof of the Correctness of the CHC Representation}
\label{index:appx-proof}

\subsection{Abstract Operational Semantics}
\label{index:appx-proof-aos}

We introduce \emph{abstract operation semantics} for COR, as a mediator between concrete operational semantics and the logic.
In abstract operational semantics, we get rid of heaps and directly represent each variable as a value with such future values expressed as \emph{abstract variables} \(\ab{x}\) (marked bold and light blue), which is strongly related to \emph{prophecy variables}.
An abstract variable represents the undetermined value of a mutable reference at the end of borrow.

Formally, we introduce a \emph{pre-value}, which is defined as follows:
\begin{gather*}
  \ltag{pre-value} \hat v, \hat w
  \sdef \angled{\hat v}
  \sor \angled{\hat v_*, \hat v_\0}
  \sor \inj_i \hat v
  \sor (\hat v_0, \hat v_1)
  \sor \const
  \sor \ab{x}.
\end{gather*}

Abstract operational semantics is described as transition on program states encoded as an \emph{abstract configuration} \(\aC\), which is defined as follows.
Here, an \emph{abstract stack frame} \(\aF\) maps variables to pre-values.
We may omit the terminator `\(;\, \End\)'.
\begin{gather*}
  \aS \sdef \End \mathrel{\, \bigm|\,} [f,L]_\cTheta\,x,\aF;\, \aS \quad
  \ltag{abstract configuration} \aC \sdef [f,L]_\cTheta\, \aF;\, \aS \mid_\cA
\end{gather*}

In order to facilitate proofs later, we append lifetime-related ghost information to \(\aC\), which does not directly affect the execution.
\(\cA\) is a \emph{global lifetime context}, which is the lifetime context of all local lifetime variables from all stack frames;
we add a \emph{tag} on a local lifetime variable (e.g. \(\alpha^{(i)}\) instead of \(\alpha\)) to clarify which stack frame it belongs to.
\(\cTheta\) is a \emph{lifetime parameter context}, which maps the lifetime variables in the (local) lifetime context for a stack frame to the corresponding \emph{tagged} lifetime variables in the global lifetime context.

Just as concrete operational semantics,
abstract operational semantics is characterized by the one-step transition relation \(\aC \to_\varPi \aC'\) and the termination relation \(\final_\varPi(\aC)\), which are defined by the following rules.
\(\aC[\hat v/\ab{x}]\) is \(\aC\) with every \(\ab{x}\) in its abstract stack frames replaced with \(\hat v\).
`\(\val\)' maps both \(\angled{\hat v}\) and \(\angled{\hat v, \ab{x_\0}}\) to \(\hat v\).
\begingroup\small
\begin{gather*}
  \frac{
    S_{\varPi,f,L} = \Let y = \mutbor_\alpha x;\, \goto L' \quad
    \ab{x_\0}\ \text{is fresh}
  }{
    [f,L]_\cTheta\, \aF \!+\! \{(x,\angled{\hat v_*})\};\, \aS \mid_\cA
    \ \to_\varPi\ [f,L']_\cTheta\, \aF \!+\! \{(y,\angled{\hat v_*,\ab{x_\0}}),(x,\angled{\ab{x_\0}})\};\, \aS \mid_\cA
  } \br[.3em]
  \frac{
    S_{\varPi,f,L} = \Let y = \mutbor_\alpha x;\, \goto L' \quad
    \ab{x_\0}\ \text{is fresh}
  }{
    [f,L]_\cTheta\, \aF \!+\! \{(x,\angled{\hat v_*,\ab{x'_\0}})\};\, \aS \mid_\cA
    \ \to_\varPi\ [f,L']_\cTheta\, \aF \!+\! \{(y,\angled{\hat v_*,\ab{x_\0}}),(x,\angled{\ab{x_\0},\ab{x'_\0}})\};\, \aS \mid_\cA
  } \br[.3em]
  \frac{
    S_{\varPi,f,L} = \drop x;\, \goto L' \quad
    \Ty_{\varPi,f,L}(x) = \check P\, T
  }{
    [f,L]_\cTheta\, \aF \!+\! \{(x, \hat v)\};\, \aS \mid_\cA
    \ \to_\varPi\ [f,L']_\cTheta\, \aF;\, \aS \mid_\cA
  } \br[.3em]
  \frac{
    S_{\varPi,f,L} = \drop x;\, \goto L' \quad
    \Ty_{\varPi,f,L}(x) = \mut_\alpha T
  }{
    [f,L]_\cTheta\, \aF \!+\! \{(x,\angled{\hat v_*,\ab{x_\0}})\};\, \aS \mid_\cA
    \ \to_\varPi\ \bigl([f,L']_\cTheta\, \aF;\, \aS \mid_\cA\bigr)\bigl[\hat v_*/\ab{x_\0}\bigr]
  } \br[.3em]
  \frac{
    S_{\varPi,f,L} = \immut x;\, \goto L'
  }{
    [f,L]_\cTheta\, \aF \!+\! \{(x,\angled{\hat v_*,\ab{x_\0}})\};\, \aS \mid_\cA
    \ \to_\varPi\ \bigl([f,L']_\cTheta\, \aF \!+\! \{(x,\angled{\hat v_*})\};\, \aS \mid_\cA\bigr)\bigl[\hat v_*/\ab{x_\0}\bigr]
  } \br[.3em]
  \frac{
    S_{\varPi,f,L} = \swap(*x,*y);\, \goto L' \quad
    \Ty_{\varPi,f,L}(y) = \own T
  }{\begin{aligned}
    & [f,L]_\cTheta\, \aF \!+\! \{(x,\angled{\hat v_*, \ab{x_\0}}),(y,\angled{\hat w_*})\};\, \aS \mid_\cA \\[-.2em]
    & \hspace{6em} \to_\varPi\ [f,L']_\cTheta\, \aF \!+\! \{(x,\angled{\hat w_*, \ab{x_\0}}),(y,\angled{\hat v_*})\};\, \aS \mid_\cA
  \end{aligned}} \br[.3em]
  \frac{
    S_{\varPi,f,L} = \swap(*x,*y);\, \goto L' \quad
    \Ty_{\varPi,f,L}(y) = \mut_\alpha T
  }{\begin{aligned}
    & [f,L]_\cTheta\, \aF \!+\! \{(x, \angled{\hat v_*, \ab{x_\0}}),(y,\angled{\hat w_*, \ab{y_\0}})\};\, \aS \mid_\cA \\[-.2em]
    & \hspace{6em} \to_\varPi\ [f,L']_\cTheta\, \aF \!+\! \{(x, \angled{\hat w_*, \ab{x_\0}}),(y, \angled{\hat v_*, \ab{y_\0}})\};\, \aS \mid_\cA
  \end{aligned}} \br[.3em]
  \frac{
    S_{\varPi,f,L} = \Let *y = x;\, \goto L'
  }{
    [f,L]_\cTheta\, \aF \!+\! \{(x, \hat v)\};\, \aS \mid_\cA
    \ \to_\varPi\ [f,L']_\cTheta\, \aF \!+\! \{(y,\angled{\hat v})\};\, \aS \mid_\cA
  } \br[.3em]
  \frac{
    S_{\varPi,f,L} = \Let y = *x;\, \goto L' \quad
    \Ty_{\varPi,f,L}(x) = \own P\, T
  }
  {[f,L]_\cTheta\, \aF \!+\! \{(x,\angled{\hat v_*})\};\, \aS \mid_\cA
    \ \to_\varPi\ [f,L']_\cTheta\, \aF \!+\! \{(y, \hat v_*)\};\, \aS \mid_\cA} \br[.3em]
  \frac{
    S_{\varPi,f,L} = \Let y = *x;\, \goto L' \quad
    \Ty_{\varPi,f,L}(x) = \immut_\alpha P\, T
  }{
    [f,L]_\cTheta\, \aF \!+\! \{(x,\angled{\hat v_*})\};\, \aS \mid_\cA
    \ \to_\varPi\ [f,L']_\cTheta\, \aF \!+\! \{(y,\angled{\val(\hat v_*)})\};\, \aS \mid_\cA
  } \br[.3em]
  \frac{
    S_{\varPi,f,L} = \Let y = *x;\, \goto L' \quad
    \Ty_{\varPi,f,L}(x) = \mut_\alpha \own T \quad
    \ab{x_{\0*}}\ \text{is fresh}
  }{
    [f,L]_\cTheta\, \aF \!+\! \{(x,\angled{\angled{\hat v_{**}},\ab{x_\0}})\};\, \aS \mid_\cA
    \ \to_\varPi\ \bigl([f,L']_\cTheta\, \aF \!+\! \{(y,\angled{\hat v_{**},\ab{x_{\0*}}})\};\, \aS \mid_\cA\bigr)\bigl[\angled{\ab{x_{\0*}}}/\ab{x_\0}\bigr]
  } \br[.3em]
  \frac{
    S_{\varPi,f,L} = \Let y = *x;\, \goto L' \quad
    \Ty_{\varPi,f,L}(x) = \mut_\alpha \immut_\beta T
  }{
    [f,L]_\cTheta\, \aF \!+\! \{(x,\angled{\angled{\hat v_{**}},\ab{x_\0} })\};\, \aS \mid_\cA
    \ \to_\varPi\ \bigl([f,L']_\cTheta\, \aF \!+\! \{(y,\angled{\hat v_{**}})\};\, \aS \mid_\cA\bigr)\bigl[\angled{\hat v_{**}}/\ab{x_\0}\bigr]
  } \br[.3em]
  \frac{
    S_{\varPi,f,L} = \Let y = *x;\, \goto L' \quad
    \Ty_{\varPi,f,L}(x) = \mut_\alpha \mut_\beta T \quad
    \ab{x_{*\0}}\ \text{is fresh}
  }{\begin{aligned}
    & [f,L]_\cTheta\, \aF \!+\! \{(x,\angled{\angled{\hat v_{**},\ab{x'_{*\0}}},\ab{x_\0}})\};\, \aS \mid_\cA \\[-.2em]
    & \hspace{7em} \to_\varPi\ \bigl([f,L']_\cTheta\, \aF \!+\! \{(y,\angled{\hat v_{**},\ab{x_{*\0}}})\};\, \aS \mid_\cA\bigr)\bigl[\angled{\ab{x_{*\0}},\ab{x'_{*\0}}}/\ab{x_\0}\bigr]
  \end{aligned}
  } \br[.3em]
  \frac{
    S_{\varPi,f,L} = \Let *y = \Copy *x;\, \goto L'
  }{
    [f,L]_\cTheta\, \aF;\, \aS \mid_\cA
    \ \to_\varPi\ [f,L']_\cTheta\, \aF \!+\! \{(y,\angled{\val(\aF(x))})\};\, \aS \mid_\cA
  } \br[.3em]
  \frac{
    S_{\varPi,f,L} = x \as T;\, \goto L'
  }{
    [f,L]_\cTheta\, \aF;\, \aS \mid_\cA
    \ \to_\varPi\ [f,L']_\cTheta\, \aF;\, \aS \mid_\cA
  } \br[.3em]
  \frac{\begin{gathered}
    S_{\varPi,f,L} = \Let y = g\angled{\alpha_0,\dots,\alpha_{m-1}}(x_0,\dots,x_{n-1});\, \goto L' \\[-.2em]
    \varSigma_{\varPi,g} = \angled{\alpha'_0,\dots,\alpha'_{m-1} \mid \cdots}\,(x'_0\colon T_0,\dots,x'_{n-1}\colon T_{n-1}) \quad
    \cTheta' = \{(\alpha'_j,\alpha_j\cTheta) \mid j \!\in\! [m]\}
  \end{gathered}}{
    [f,L]_\cTheta\, \aF \!+\! \{(x_i, \hat v_i) \!\mid\! i \!\in\! [n]\};\, \aS \mid_\cA
    \ \to_\varPi\ [g,\entry]_{\cTheta'}\, \{(x'_i, \hat v_i) \!\mid\! i \!\in\! [n]\};\, [f,L']_\cTheta\,y,\aF;\, \aS \mid_\cA
  } \br[.3em]
  \frac{
    S_{\varPi,f,L} = \return x \quad
  }{
    [f,L]_\cTheta\, \{(x, \hat v)\}; [g,L']_{\cTheta'}\,x',\aF';\, \aS \mid_\cA
    \ \to_\varPi\ [g,L']_{\cTheta'}\, \aF' \!+\! \{(x', \hat v)\};\, \aS \mid_\cA
  } \br[.3em]
  \frac{
    S_{\varPi,f,L} = \return x \quad
  }{
    \final_\varPi\bigl([f,L]_\cTheta\, \{(x, \hat v)\} \mid_\cA\bigr)
  } \br[.3em]
  \frac{\begin{gathered}
    S_{\varPi,f,L} = \intro \alpha;\, \goto L' \quad
    \aS\ \text{has}\ n\ \text{layers} \quad
    A_\ex = \{\alpha^{(k)} \!\in\! A \mid k \!<\! n\}
  \end{gathered}}{
    [f,L]_\cTheta\, \aF;\, \aS \mid_{(A,R)}
    \ \to_\varPi\ [f,L']_{\cTheta + \{(\alpha,\alpha^{(n)})\}}\, \aF;\, \aS \mid_{(\{\alpha^{(n)}\} + A,\, \{\alpha^{(n)}\} \times (\{\alpha^{(n)}\} + A_\ex) + R)}
  } \br[.3em]
  \frac{
    S_{\varPi,f,L} = \now \alpha;\, \goto L'
  }{
    [f,L]_{\{(\alpha,\alpha^{(n)})\} + \cTheta}\, \aF;\, \aS \mid_{(\{\alpha^{(n)}\} + A,R)}
    \ \to_\varPi\ [f,L']_\cTheta\, \aF;\, \aS \mid_{(A,\, \{(\beta^{(k)},\gamma^{(l)}) \in R \,\mid\, \beta^{(k)} \ne \alpha^{(n)}\})}\,
  } \br[.3em]
  \frac{\begin{gathered}
    S_{\varPi,f,L} = \alpha \le \beta;\, \goto L'
  \end{gathered}}{
    [f,L]_\cTheta\, \aF;\, \aS \mid_{(A,R)}
    \ \to_\varPi\ [f,L']_\cTheta\, \aF;\, \aS \mid_{(A,\, (\{(\cTheta(\alpha),\cTheta(\beta))\}+ R)^+)}
  } \br[.3em]
  \frac{
    S_{\varPi,f,L} = \Let *y = \const;\, \goto L'
  }{
    [f,L]_\cTheta\, \aF;\, \aS \mid_\cA
    \ \to_\varPi\ [f,L']_\cTheta\, \aF \!+\! \{(y,\angled{\const})\};\, \aS \mid_\cA
  } \br[.3em]
  \frac{
    S_{\varPi,f,L} = \Let *y = *x \op *x';\, \goto L'
  }{
    [f,L]_\cTheta\, \aF;\, \aS \mid_\cA
    \ \to_\varPi\ [f,L']_\cTheta\, \aF \!+\! \{(y,\angled{\val(\aF(x)) \mathop{\lBrack\op\rBrack} \val(\aF(x'))})\};\, \aS \mid_\cA
  } \br[.3em]
  \frac{
    S_{\varPi,f,L} = \Let *y = \rand();\, \goto L'
  }{
    [f,L]_\cTheta\, \aF;\, \aS \mid_\cA
    \ \to_\varPi\ [f,L']_\cTheta\, \aF \!+\! \{(y,\angled{n})\};\, \aS \mid_\cA
  } \br[.3em]
  \frac{
    S_{\varPi,f,L} = \Let *y = \inj^{T_0 \!+\! T_1}_i *x;\, \goto L'
  }{
    [f,L]_\cTheta\, \aF \!+\! \{(x,\angled{\hat v_*})\};\, \aS \mid_\cA
    \ \to_\varPi\ [f,L']_\cTheta\, \aF \!+\! \{(y,\angled{\inj_i \hat v_*})\};\, \aS \mid_\cA
  } \br[.3em]
  \frac{\begin{gathered}
    S_{\varPi,f,L} = \match{*x}{\inj_0 *y_0 \to \goto L'_0,\ \inj_1 *y_1 \to \goto L'_1} \\[-.4em]
    \Ty_{\varPi,f,L}(x) = \check P\, (T_0 \!+\! T_1)
  \end{gathered}}{
    [f,L]_\cTheta\, \aF \!+\! \{(x,\angled{\inj_i \hat v_{*!}})\};\, \aS \mid_\cA
    \ \to_\varPi\ [f,L'_i]_\cTheta\, \aF \!+\! \{(y_i,\angled{\hat v_{*!}})\};\, \aS \mid_\cA
  } \br[.3em]
  \frac{\begin{gathered}
    S_{\varPi,f,L} = \match{*x}{\inj_0 *y_0 \to \goto L'_0,\ \inj_1 *y_1 \to \goto L'_1} \\[-.2em]
    \Ty_{\varPi,f,L}(x) = \mut_\alpha (T_0 \!+\! T_1) \quad
    \ab{x_{\0!}}\ \text{is fresh}
  \end{gathered}}{
    [f,L]_\cTheta\, \aF \!+\! \{(x,\angled{\inj_i \hat v_{*!},\ab{x_\0}})\};\, \aS \mid_\cA
    \ \to_\varPi\ \bigl([f,L'_i]_\cTheta\, \aF \!+\! \{(y_i,\angled{\hat v_{*!}, \ab{x_{\0!}}})\};\, \aS \mid_\cA\bigr)\bigl[\inj_i \ab{x_{\0!}}/\ab{x_\0}\bigr]
  } \br[.3em]
  \frac{
    S_{\varPi,f,L} = \Let *y = (*x_0,*x_1);\, \goto L'
  }{
    [f,L]_\cTheta\, \aF \!+\! \{(x_0, \angled{\hat v_{*0}}), (x_1, \angled{\hat v_{*1}})\};\, \aS \mid_\cA
    \ \to_\varPi\ [f,L']_\cTheta\, \aF \!+\! \{(y,\angled{(\hat v_{*0}, \hat v_{*1})})\};\, \aS \mid_\cA
  } \br[.3em]
  \frac{
    S_{\varPi,f,L} = \Let\, (*y_0,*y_1) = *x;\, \goto L'
  }{
    [f,L]_\cTheta\, \aF \!+\! \{(x,\angled{(\hat v_{*0},\hat v_{*1})})\};\, \aS \mid_\cA
    \ \to_\varPi\ [f,L']_\cTheta\, \aF \!+\! \{(y_0,\angled{\hat v_{*0}}), (y_1,\angled{\hat v_{*1}})\};\, \aS \mid_\cA
  } \br[.3em]
  \frac{
    S_{\varPi,f,L} = \Let\, (*y_0,*y_1) = *x;\, \goto L' \quad
    \ab{x_{\00}},\ab{x_{\01}}\ \text{are fresh}
  }{\begin{aligned}
    & [f,L]_\cTheta\, \aF \!+\! \{(x,\angled{(\hat v_{*0},\hat v_{*1}), \ab{x_\0}})\};\, \aS \mid_\cA \\[-.3em]
    & \ \to_\varPi\ \bigl([f,L']_\cTheta\, \aF \!+\! \{(y_0,\angled{\hat v_{*0}, \ab{x_{\00}}}), (y_1,\angled{\hat v_{*1}, \ab{x_{\01}}})\};\, \aS \mid_\cA\bigr)\bigl[(\ab{x_{\00}},\ab{x_{\01}})/\ab{x_\0}\bigr]
  \end{aligned}}
\end{gather*}
\endgroup

\begin{example}[Execution on Abstract Operaitonal Semantics]
The following is an example execution on abstract operational semantics for \cref{example:cor-program}.
It corresponds to \cref{example:cos-execution}, the example execution on concrete operational semantics.

Here, \(\cA \defeq (\{\alpha\}, \Id_{\{\alpha\}})\) and \(\cTheta \defeq \{\alpha, \alpha^{(0)}\}\).
\begingroup\small
\begin{align*}
  &
    [\incmax,\entry]_\emp\, \{(\nvar{oa},\angled{4}),(\nvar{ob},\angled{3})\} \mid_{(\!\emp,\emp\!)}
  \br[-.2em] &
    \to [\incmax,\nL{1}]_\cTheta\, \{(\nvar{oa},\angled{4}),(\nvar{ob},\angled{3})\} \mid_\cA
  \br[-.2em] &
    \to^+ [\incmax,\nL{3}]_\cTheta\, \{(\nvar{ma},\angled{4,\ab{a_\0}}),(\nvar{mb},\angled{3,\ab{b_\0}}),(\nvar{oa},\angled{\ab{a_\0}}),(\nvar{ob},\angled{\ab{b_\0}})\} \mid_\cA
  \br[-.2em] &
    \begin{aligned}
      & \to [\takemax,\entry]_\cTheta\, \{(\nvar{ma},\angled{4,\ab{a_\0}}),(\nvar{mb},\angled{3,\ab{b_\0}})\}; \\[-.3em]
      & \hspace{4em} [\incmax,\nL{4}]_\cTheta\, \var{mc},\{(\nvar{oa},\angled{\ab{a_\0}}),(\nvar{ob},\angled{\ab{b_\0}})\} \mid_\cA
    \end{aligned}
  \br[-.2em] &
    \begin{aligned}
      & \to [\takemax,\nL{1}]_\cTheta\, \{(\nvar{ord},\angled{\inj_1 {()}}),(\nvar{ma},\angled{4,\ab{a_\0}}),(\nvar{mb},\angled{3,\ab{b_\0}})\}; \\[-.3em]
      & \hspace{4em} [\incmax,\nL{4}]_\cTheta\, \var{mc},\{(\nvar{oa},\angled{\ab{a_\0}}),(\nvar{ob},\angled{\ab{b_\0}})\} \mid_\cA
    \end{aligned}
  \br[-.2em] &
    \begin{aligned}
      & \to [\takemax,\nL{2}]_\cTheta\, \{(\nvar{ou},\angled{()}),(\nvar{ma},\angled{4,\ab{a_\0}}),(\nvar{mb},\angled{3,\ab{b_\0}})\}; \\[-.3em]
      & \hspace{4em} [\incmax,\nL{4}]_\cTheta\, \var{mc},\{(\nvar{oa},\angled{\ab{a_\0}}),(\nvar{ob},\angled{\ab{b_\0}})\} \mid_\cA
    \end{aligned}
  \br[-.2em] &
    \begin{aligned}
      & \to^+ [\takemax,\nL{4}]_\cTheta\, \{(\nvar{ma},\angled{4,\ab{a_\0}})\}; \\[-.3em]
      & \hspace{4em} [\incmax,\nL{4}]_\cTheta\, \var{mc},\{(\nvar{oa},\angled{\ab{a_\0}}),(\nvar{ob},\angled{3})\} \mid_\cA
    \end{aligned}
  \br[-.2em] &
    \to [\incmax,\nL{4}]_\cTheta\, \{(\nvar{mc},\angled{4,\ab{a_\0}}),(\nvar{oa},\angled{\ab{a_\0}}),(\nvar{ob},\angled{3})\} \mid_\cA
  \br[-.2em] &
    \to [\incmax,\nL{5}]_\cTheta\, \{(\nvar{o1},\angled{1}),(\nvar{mc},\angled{4,\ab{a_\0}}),(\nvar{oa},\angled{\ab{a_\0}}),(\nvar{ob},\angled{3})\} \mid_\cA
  \br[-.2em] &
    \to^+ [\incmax,\nL{7}]_\cTheta\, \{(\nvar{oc'},\angled{5}),(\nvar{mc},\angled{4,\ab{a_\0}}),(\nvar{oa},\angled{\ab{a_\0}}),(\nvar{ob},\angled{3})\} \mid_\cA
  \br[-.2em] &
    \to [\incmax,\nL{8}]_\cTheta\, \{(\nvar{oc'},\angled{4}),(\nvar{mc},\angled{5,\ab{a_\0}}),(\nvar{oa},\angled{\ab{a_\0}}),(\nvar{ob},\angled{3})\} \mid_\cA
  \br[-.2em] &
    \to^+ [\incmax,\nL{10}]_\cTheta\, \{(\nvar{oa},\angled{5}),(\nvar{ob},\angled{3})\} \mid_\cA
  \br[-.2em] &
    \to [\incmax,\nL{11}]_\emp\, \{(\nvar{oa},\angled{5}),(\nvar{ob},\angled{3})\} \mid_{(\!\emp,\emp\!)}
  \br[-.2em] &
    \to^+ [\incmax,\nL{14}]_\emp\, \{(\nvar{or},\angled{\inj_1 {()}})\} \mid_{(\!\emp,\emp\!)}
\end{align*}
\endgroup
The abstract variables \(\ab{a_\0}\) and \(\ab{b_\0}\) are introduced for mutable borrow of \(\nvar{oa}\) and \(\nvar{ob}\).
By the call of \(\takemax\), \(\nvar{mb}\) is released, whereby the variable \(\ab{b_\0}\) is set to the value \(3\),
and the variable \(\ab{a_\0}\) is passed to \(\nvar{mc}\).
After the increment is performed, \(\nvar{mc}\) is released, and thereby \(\ab{a_\0}\) is set to the updated value \(5\).
\end{example}

\subsection{Safety on Abstract Configurations}
\label{index:appx-proof-aos-safe}

It is natural to require for an abstract configuration that each variable is shared by the borrower and the lender and is not used elsewhere.\footnote{%
  We should take care of the cases where a mutable reference is immutably borrowed (e.g. \(\immut_\alpha \mut_\beta T\)), because immutable references can be unrestrictedly copied.
  Later when we define `\(\summary\)` judgments, we get over this problem using \emph{access modes}.
}
A stack of borrows (caused by reborrows) can be described as a chain of abstract variables (e.g. \(\angled{v, \ab{x}}, \angled{\ab{x}, \ab{y}}, \angled{\ab{y}}\)).

To describe such restrictions, we define the \emph{safety} on an abstract configuration `\(\safe_\varPi(\aC)\)'.
We also show \emph{progression and preservation} regarding safety on \emph{abstract operational semantics}, as a part of soundness of COR's type system.

\Subsubsection{Summary}
An \emph{abstract variable summary} \(\aX\) is a finite multiset of items of form `\(\give_\alpha(\ab{x}\Colon T)\)' or `\(\take^\alpha(\ab{x}\Colon T)\)'.

Now, `\(\summary_D^\ac(\hat v\Colon T \mid \aX)\)'
(the pre-value \(\hat v\) of type \(T\) yields an abstract variable summary \(\aX\), under the access mode \(D\) and the activeness \(\ac\))
is defined as follows.
Here, an \emph{access mode} \(D\) is either of form `\(\hot\)' or `\(\cold\)'.
\begingroup\small
\begin{gather*}
  \summary_D^{\dagger\alpha}(\ab{x}\Colon T \mid \{\take^\alpha(\ab{x}\Colon T)\}) \quad
  \frac{
    \summary_{D \cdot \check P}^\ac(\hat v\Colon T \mid \aX)
  }{
    \summary_D^\ac(\angled{\hat v}\Colon \check P\, T \mid \aX)
  } \br[-.1em]
  \dbox{\(
    D \cdot \own \defeq D \quad
    D \cdot \immut_\beta \defeq \cold
  \)} \br[.1em]
  \frac{
    \summary_{\hot}^\ac(\hat v\Colon T \mid \aX)
  }{
    \summary_{\hot}^\ac(\angled{\hat v, \ab{x}}\Colon \mut_\beta T \mid \aX \oplus \{\give_\beta(\ab{x}\Colon T)\})
  } \quad
  \frac{
    \summary_{\cold}^\ac(\hat v\Colon T \mid \aX)
  }{
    \summary_{\cold}^\ac(\angled{\hat v, \hat w}\Colon \mut_\beta T \mid \aX)
  } \br[.1em]
  \frac{
    \summary_D^\ac(\hat v\Colon T[\mu X.T/X] \mid \aX)
  }{
    \summary_D^\ac(\hat v\Colon \mu X.T/X \mid \aX)
  } \quad
  \summary_D^\ac(\const\Colon T \mid \emp) \br[.1em]
  \frac{
    \summary_D^\ac(\hat v\Colon T_i \mid \aX)
  }{
    \summary_D^\ac\bigl(\inj_i \hat v\Colon T_0 \!+\! T_1 \bigm| \aX\bigr)
  } \quad
  \frac{
    \summary_D^\ac(\hat v_0\Colon T_0 \mid \aX_0) \quad
    \summary_D^\ac(\hat v_1\Colon T_1 \mid \aX_1)
  }{
    \summary_D^\ac\bigl((\hat v_0, \hat v_1)\Colon T_0 \!\times\! T_1 \bigm| \aX_0 \oplus \aX_1\bigr)
  }
\end{gather*}
\endgroup

`\(\summary_\cTheta(\aF\Colon \cGamma \mid \aX)\)'
(the abstract stack frame \(\aF\) respecting the variable context \(\cGamma\) yields \(\aX\), under the lifetime parameter context \(\cTheta\))
is defined as follows.
\begingroup\small
\begin{gather*}
  \frac{
    \dom \cF = \dom \cGamma \quad
    \text{for any}\ x\colonu\ac T \in \cGamma,\
    \summary_{\hot}^\ac\bigl(\aF(x)\Colon T\, \cTheta \, \mid\, \aX_x\bigr)
  }{
    \summary_\cTheta\bigl(\aF\Colon \cGamma \bigm| \bigoplus_{x\colonu\ac T \in \cGamma} \aX_x\bigr)
  }
\end{gather*}
\endgroup

Finally, `\(\summary_\varPi(\aC \mid \aX)\)' (the abstract configuration \(\aC\) yields \(\aX\) under the program \(\varPi\)) is defined as follows.
\begingroup\small
\begin{gather*}
  \frac{
    \text{for any}\ i \in [n + 1],\
    \summary_{\cTheta_i}(\aF_i\Colon \cGamma_{\varPi,f_i,L_i} \mid \aX_i)
  }{
    \summary_\varPi\bigl(
      [f_0,L_0]_{\cTheta_0}\, \aF_0;\,
      [f_1,L_1]_{\cTheta_1}\, x_1,\aF_1;\, \cdots;\,
      [f_n,L_n]_{\cTheta_n}\, x_n,\aF_n \mid_\cA \,\bigm|\, \bigoplus_{i=0}^n\!\aX_i \,\bigr)
  }
\end{gather*}
\endgroup

\Subsubsection{Lifetime Safety}

`\(\lifetimeSafe_i(\cA_\Global,\cTheta \mid \cA_\local,A_\ex)\)'
(the global lifetime context \(\cA_\Global\) with the lifetime parameter context \(\cTheta\) is safe on lifetimes
with respect to the (local) lifetime context \(\cA_\local\) from the type system and the set of lifetime parameters \(A_\ex\) under the stack frame index \(i\))
is defined as follows.
\begingroup\small
\begin{gather*}
  \frac{\begin{gathered}
    \dom \cTheta = \lvert\cA_\local\rvert \qquad
    \text{for any}\ \alpha \!\in\! A_\ex,\ \text{letting}\ \beta^{(k)} = \cTheta(\alpha),\ k < i\ \text{holds} \\[-.5em]
    \text{for any}\ \alpha \!\in\! \lvert\cA_\local\rvert \!-\! A_\ex,\
    \cTheta(\alpha) = \alpha^{(i)} \\[-.3em]
    \text{for any}\ (\alpha,\beta) \!\in\! \lvert\cA_\local\rvert^2 \!-\! A_\ex^2,\ \,
    \alpha \!\le_{\cA_\local}\! \beta \iff \cTheta(\alpha) \!\le_{\cA_\Global}\! \cTheta(\beta) \\[-.3em]
    \text{for any}\ \alpha,\beta \!\in\! A_\ex^2,\
    \alpha \!\le_{\cA_\local}\! \beta \implies \cTheta(\alpha) \!\le_{\cA_\Global}\! \cTheta(\beta)
  \end{gathered}}{
    \lifetimeSafe_i(\cA_\Global,\cTheta \mid \cA_\local,A_\ex)
  }
\end{gather*}
\endgroup

`\(\lifetimeSafe_\varPi\bigl(\cA_\Global,(f_i,L_i,\cTheta_i)_{i=0}^n\bigr)\)'
(\(\cA_\Global\) with the finite sequence of function names, labels and lifetime parameter contexts \((f_i,L_i,\cTheta_i)_{i=0}^n\) is safe on lifetimes under the program \(\varPi\))
is defined as follows.
\begingroup\small
\begin{gather*}
  \frac{\begin{gathered}
    \text{for any}\ i \!\in\! [n \!+\! 1],\
    \lifetimeSafe_i(\cA_\Global,\cTheta_i \mid \cA_{\varPi,f_i,L_i},A_{\ex\, \varPi,f_i}) \\[-.3em]
    \textstyle \card\, \lvert\cA_\Global\rvert\,=\, \sum_{i=0}^{n} \card\,(\lvert\cA_{\varPi,f_i,L_i}\rvert \!-\! A_{\ex\, \varPi,f_i})
  \end{gathered}}{
    \lifetimeSafe_\varPi\bigl(\cA_\Global,(f_i,L_i,\cTheta_i)_{i=0}^n\bigr)
  } \br[-.1em]
  \text{
    \(\cA_{\varPi,f,L}\): the lifetime context for the label \(L\) of \(f\) in \(\varPi\) \quad
    \(\card X\): the cardinality of \(X\)
  }
\end{gather*}
\endgroup

Finally, `\(\lifetimeSafe_\varPi(\aC)\)' (the abstract configuration \(\aC\) is safe on lifetimes under the program \(\varPi\)) is defined as follows.
\begingroup\small
\begin{gather*}
  \frac{
    \lifetimeSafe_\varPi\bigl(\cA_\Global,(f_i,L_i,\cTheta_i)_{i=0}^n\bigr)
  }{
    \lifetimeSafe_\varPi\bigl(
      [f_n,L_n]_{\cTheta_n}\, \aF_n;
      [f_{n-1},L_{n-1}]_{\cTheta_{n-1}}\,x_{n-1},\aF_{n-1};\, \cdots;\,
      [f_0,L_0]_{\cTheta_0}\,x_0,\aF_0 \mid_{\cA_\Global}
    \bigr)
  }
\end{gather*}
\endgroup

\Subsubsection{Safety}

We first define the safety on abstract variable summaries.
`\(\safe_\cA(\ab{x},\aX)\)' is defined as follows.
Here, \(T \sim_\cA U\) means \(T \le_\cA U \land U \le_\cA T\) (the \emph{type equivalence}).
\begingroup\small
\begin{gather*}
  \frac{
    \aX(\ab{x}) = \Braced{\give_\alpha(\ab{x}\Colon T),\, \take^\beta(\ab{x}\Colon T')} \quad
    T \sim_\cA T' \quad
    \alpha \le_\cA \beta
  }{
    \safe_\cA(\ab{x},\aX)
  } \quad
  \frac{
    \aX(\ab{x}) = \emp
  }{
    \safe_\cA(\ab{x},\aX)
  } \br[-.1em]
  \text{
    \(\aX(\ab{x})\): the multiset of the items of form `\(\give_\gamma(\ab{x}\Colon U)\)'/`\(\take^\gamma(\ab{x}\Colon U)\)' in \(\aX\)
  }
\end{gather*}
\endgroup
`\(\safe_\cA(\aX)\)' means that \(\safe_\cA(\ab{x},\aX)\) holds for any \(\ab{x}\).

Finally, `\(\safe_\varPi(\aC)\)' is defined as follows.
\begingroup\small
\begin{gather*}
  \frac{\begin{gathered}
    \summary_\varPi(\aC \mid \aX) \quad
    \lifetimeSafe_\varPi(\aC) \quad
    \aC = {\cdots} \mid_\cA \quad
    \safe_\cA(\aX)
  \end{gathered}}{
    \safe_\varPi(\aC)
  }
\end{gather*}
\endgroup

\begin{property}[Safety on an Abstract Configuration Ensures Progression]
  For any \(\varPi\) and \(\aC\) such that \(\safe_\varPi(\aC)\) holds and \(\final_\varPi(\aC)\) does not hold,
  there exists \(\aC'\) satisfying \(\aC \to_\varPi \aC'\).
\end{property}
\begin{proof}
  Clear.
  The important guarantee the safety on an abstract configuration provides is that,
  in the pre-value assigned to each \emph{active} variable, abstract variables do not appear except in the form \(\angled{\hat v, \ab{x}}\).
\qed\end{proof}

\begin{lemma}[Safety on the Abstract Configuration is Preserved]
\label{lemma:preserve-aos}
  For any \(\varPi\) and \(\aC,\aC'\) such that \(\safe_\varPi(\aC)\) and \(\aC \to_\varPi \aC'\) hold, \(\safe_\varPi(\aC')\) is satisfied.
\end{lemma}
\begin{proof}
  Straightforward.
  Preservation of safety on the abstract variable summary is the point.
  Below we check some tricky cases.

  \Paragraph{Type Weakening}
  Type weakening (\(x \as T\)) essentially only changes lifetimes on types.
  A lifetime on a type can become earlier if it is \emph{not} guarded by any \(\mut_\alpha\).
  Thus only the following changes happen on the abstract variable summary:
  (i) for an item of form `\(\give_\alpha(\ab{x}\Colon T)\)', \(\alpha\) can get earlier and \(T\) can be weakened;
  and (ii) for an item of form `\(\take^\alpha(\ab{x}\Colon T)\)', \(\alpha\) do not change and \(T\) can be weakened.

  \Paragraph{Mutable (Re)borrow}
  When we perform \(\Let \var{my} = \mutbor_\alpha \var{px}\),
  the abstract variable summary just gets two new items `\(\give_\alpha(\ab{x_\0}\Colon T)\)' and `\(\take^\alpha(\ab{x_\0}\Colon T)\)', for some \(\ab{x_\0}\) and \(T\).

  \Paragraph{Release of a Mutable Reference}
  When we release a mutable reference \(\var{mx}\), whose pre-value is of form \(\angled{\hat v, \ab{x_\0}}\),
  only the following changes happen on the abstract variable summary:
  (i) the items of form `\(\give_\alpha(\ab{x_\0}\Colon T)\)' and `\(\take^\beta(\ab{x_\0}\Colon T')\)' are removed;
  and (ii) since \(\hat v\) moves to another variable, the type of each abstract variable in \(\hat v\) may change into an equivalent type.

  \Paragraph{Ownership Weakening}
  Similar to a release of a mutable reference.

  \Paragraph{Swap}
  Swap (\(\swap(*x,*y)\)) actually does not alter the abstract variable summary.

  \Paragraph{Copying}
  When data of type \(T\) is copied, \(T\colon \Copy\) holds, which ensures that each mutable reference \(\mut_\alpha U\) in \(T\) is guarded by some immutable reference.
  Therefore the abstract variable summary does not change.

  \Paragraph{Subdivision of a Mutable Reference}
  A mutable reference is subdivided in the following forms: pair destruction `\(\Let\, (*\var{mx}_0,*\var{mx}_1) = *\var{mx}\)', variant destruction `\(\match{*\var{mx}}{\inj_0 *\var{my} \!\to\! \goto L_0,\, \cdots}\)', and dereference `\(\Let \var{mx} = *\var{mpx}\)'.
  When a mutable reference \(\var{mx}\) with a pre-value \(\angled{\hat v, \ab{x}}\) is subdivided,
  the two items of form \(\give_\alpha(\ab{x}\Colon T)\) and \(\take^\beta(\ab{x}\Colon T')\) are accordingly `subdivided' in the abstract variable summary.
  With a close look, the safety turns out to be preserved.

  \Paragraph{Elimination of a Local Lifetime Variable}
  Just after we eliminate a local lifetime variable \(\alpha\) (`\(\now \alpha\)'), since there remains no lifetime variable earlier than \(\alpha\) in the lifetime context, the abstract variable summary has no item of form `\(\give_{\alpha^{(n)}}(\ab{x}\Colon T)\)' (for appropriate \(n\)).
  Therefore, just before (and just after) the lifetime elimination, the abstract variable summary has no item of form `\(\take^{\alpha^{(n)}}(\ab{x}\Colon T')\)'.
\qed\end{proof}

\subsection{SLDC Resolution}
\label{index:appx-proof-sldc}

For CHC representation of a COR program,
we introduce a variant of SLD resolution, which we call \emph{SLDC resolution} (Selective Linear Definite clause Calculative resolution).
Interpreting each CHC as a deduction rule, SLDC resolution can be understood as a \emph{top-down} construction of a proof tree from the left-hand side.
SLDC resolution is designed to be complete with respect to the logic (\cref{lemma:sldc-complete}).

A \emph{resolutive configuration} \(\aK\) and a \emph{pre-resolutive configuration} \(\hat\aK\) have the following form.
\begin{gather*}
  \ltag{resolutive configuration} \aK \sdef \check\varphi_0, \dots, \check\varphi_{n-1} \mid q \br[-.2em]
  \ltag{pre-resolutive configuration} \hat\aK \sdef \varphi_0, \dots, \varphi_{n-1} \mid q
\end{gather*}
The elementary formulas in a resolutive configuration can be understood as a model of a \emph{call stack}.
\(q\) is a pattern that represents the \emph{returned value}.
This idea is later formalized in \cref{index:appx-proof-aos-chc}.

\(\aK \to_{(\cPhi,\cXi)} \aK'\)
(\(\aK\) can change into \(\aK'\) by one step of SLDC resolution on \((\cPhi,\cXi)\))
is defined by the following non-deterministic transformation from \(\aK\) to \(\aK'\).
\begin{enumerate}
  \item
  The `stack' part of \(\aK\) should be non-empty. Let \(\aK = f(p_0,\dots,p_{m-1}),\check\varphi_1, \dots,\allowbreak \check\varphi_n \mid q\).

  Take from \(\cPhi\) any CHC that unifies with the head of the stack of \(\aK\).
  That is, \(\cPhi\) is of form
  \(\forall x_0\colon \sigma_0, \dots, x_{l-1}\colon \sigma_{l-1}.\ f(p'_0,\dots,p'_{m-1}) \!\impliedby\! \psi_0 \!\land\! \cdots \!\land\! \psi_{k-1}\)
  and \(p'_0,\dots,p'_{m-1}\) unify with \(p_0,\dots,p_{m-1}\).
  Let us take the most \emph{general} unifier \((\theta, \theta')\) such that \(p_0\theta = p'_0\theta', \dots, p_{m-1}\theta = p'_{m-1}\theta'\) hold.
  Here, \(\theta\) maps variables to patterns.

  Now we have a pre-resolutive configuration \(\hat\aK = \psi'_0, \dots, \psi'_{k-1}, \check\varphi'_1, \dots, \check\varphi'_n \mid q'\), where \(\psi'_i \defeq \psi_i\theta'\), \(\check\varphi'_j \defeq \check\varphi_j\theta\) and \(q' \defeq q\theta\).

  \item
  We `calculate' \(\hat\aK\) into a resolutive configuration.
  That is, we repeat the following operations to update \((\hat\aK\) until \(\psi'_0, \dots, \psi'_{k-1}\) all become elementary.
  \(\aK'\) is set to the final version of \(\hat\aK\).
  \begin{itemize}
    \item
    We substitute variables conservatively until there do not remain terms of form \(*x,\, \0x,\, x.i,\, x \,\op\, t / t \,\op\, x\);
    for each case, we replace \(x\) with \(\angled{x_*}\)/\(\angled{x_*, x_\0}\) (depending on the sort), \(\angled{x_*, x_\0}\), \((x_0,x_1)\), \(n\), taking fresh variables.

    \item
    We replace each \(*\angled{t_*}/*\angled{t_*,t_\0},\, \0\angled{t_*,t_\0},\, (t_0,t_1).i,\, n \,\op\, n'\) with \(t_*,\, t_\0,\, t_i,\,\allowbreak n \,\Bracked{\op}\, n'\).

    \item
    If there exists a variable \(x\) that occurs only once in the pre-resolutive configuration \(\hat\aK\), then replace it with any value of the suitable sort.\footnote{%
      We use this peculiar rule to handle the `\(\Let *y = \rand()\)' instruction later for \cref{lemma:bisim-aos-chc}.
    }
  \end{itemize}
\end{enumerate}
We have carefully designed SLDC resolution to match it with abstract operational semantics, which assists the proof of \cref{theorem:aos-chc-equivalent}.

\begin{lemma}[Completeness of SLDC Resolution]\label{lemma:sldc-complete}
  For any \((\cPhi,\cXi)\) and \(f \in \dom \cXi\),
  the following are equivalent for any values \(v_0,\dots,v_{n-1},w\) of the appropriate sorts.
  \begin{enumerate}
    \item
    \(\cM^\least_{(\cPhi,\cXi)}(f)(v_0,\dots,v_{n-1},w)\) holds.
    \item
    There exists a sequence \(\aK_0, \dots, \aK_N\) such that
    \(\aK_0 = f(v_0, \dots, v_{n-1}, r) \mid r\),
    \(\aK_N = \ \mid p\),
    \(\aK_0 \to_{(\cPhi,\cXi)} \cdots \to_{(\cPhi,\cXi)} \aK_N\) and
    \(p\) can be refined into \(w\) by instantiating variables.
  \end{enumerate}
\end{lemma}
\begin{proof}
  Clear by thinking of derivation trees (which can be defined in a natural manner) on CHC system \((\cPhi,\cXi)\).
\qed\end{proof}

\subsection{Equivalence of the AOS-based Model and the CHC Representation}
\label{index:appx-proof-aos-chc}

We first show a bisimulation between abstract operational semantics and SLDC resolution (\cref{lemma:bisim-aos-chc}).
Using the bisimulation, we can easily show the equivalence of the AOS-based model and (the least model of) the CHC representation.

\Subsubsection{Bisimulation Lemma}

Interestingly, there is a \emph{bisimulation} between the transition system of abstract operational semantics and the process of SLDC resolution.

\(\aF \leadsto^\theta_{f,L,\ab{r}} \check\varphi\)
(the abstract stack frame \(\aF\) can be translated into the elementary formula \(\check\varphi\), under \(\theta\), \(f\), \(L\) and \(\ab{r}\)) is defined as follows.
Here, \(\theta\) maps abstract variables to (normal) variables.
\(\hat v\theta\) is the value made from \(\hat v\) by replacing each \(\ab{x}\) with \(\theta(\ab{x})\).
\(\ab{r}\) is the abstract variable for taking the result.
\begingroup\small
\begin{gather*}
  \frac{
    \text{the items of \(\aF\) are enumerated as \((x_0, \hat v_0), \dots, (x_{n-1}, \hat v_{n-1})\)}
  }{
    \aF \leadsto^\theta_{f,L,\ab{r}} f_L(\hat v_0\theta_0,\dots,\hat v_{n-1}\theta,\ab{r}\theta)
  }
\end{gather*}
\endgroup

Now, \(\aC \leadsto_\varPi \aK\) is defined as follows.
\begingroup\small
\begin{gather*}
  \frac{\begin{gathered}
    \safe_\varPi(\aC) \quad
    \aC = [f_0,L_0]_{\cTheta_0}\, \aF_0;\,
      [f_1,L_1]_{\cTheta_1}\,x_1,\aF_1;\, \cdots;\,
      [f_n,L_n]_{\cTheta_n}\,x_n,\aF_n \mid_\cA \\[-.2em]
    \text{\(\ab{r_0},\dots,\ab{r_n}\) are fresh in \(\aC\)} \\[-.4em]
    \aF_0 \leadsto^\theta_{f_0,L_0,\ab{r_0}} \check\varphi_0 \quad
    \text{for any}\ i \!\in\! [n],\
    \aF_{i+1} \!+\! \{(x_{i+1},\ab{r_i})\} \leadsto^\theta_{f_{i+1},L_{i+1},\ab{r_{i+1}}} \check\varphi_{i+1}
  \end{gathered}}{
    \aC \leadsto_\varPi \check\varphi_0, \dots, \check\varphi_n \mid \theta(\ab{r_n})
  }
\end{gather*}
\endgroup

\begin{lemma}[Bisimulation between Abstract Operational Semantics and SLDC Resolution]\label{lemma:bisim-aos-chc}
  Take any \(\varPi\), \(\aC\) and \(\aK\) satisfying \(\aC \leadsto_\varPi \aK\).

  For any \(\aC'\) satisfying \(\aC \to_\varPi \aC'\),
  there exists some \(\aK'\) satisfying \(\aK \to_{\Parened{\varPi}} \aK'\) and \(\aC' \leadsto_\varPi \aK'\).
  Likewise, for any \(\aK'\) satisfying \(\aK \to_{\Parened{\varPi}} \aK'\),
  there exists some \(\aC'\) satisfying \(\aC \to_\varPi \aC'\) and \(\aC' \leadsto_\varPi \aK'\).
\end{lemma}
\begin{proof}
  Straightforward.
\qed\end{proof}

\Subsubsection{AOS-based Model and the Equivalence Theorem}

Take any \(\varPi\) and simple \(f\).
The \emph{AOS-based model} (AOS stands for abstract operational semantics) for \(f\), denoted by \(f^\AOS\), is the predicate defined by the following rule.
\begingroup\small
\begin{gather*}
  \frac{\begin{gathered}
    \aC_0 \to_\varPi \cdots \to_\varPi \aC_N \quad
    \final_\varPi(\aC_N) \quad
    \safe_\varPi(\aC_0) \\[-.3em]
    \aC_0 = [f,\entry]_\emp\, \{(x_i,v_i) \!\mid\! i \!\in\! [n]\} \mid_{(\!\emp,\emp\!)} \quad
    \aC_N = [f,L']_\emp\, \{(y,w)\} \mid_{(\!\emp,\emp\!)}
  \end{gathered}}{
    f^\AOS_\varPi(v_0,\dots,v_{n-1},w)
  }
\end{gather*}
\endgroup

Now we can prove the following theorem.

\begin{theorem}[Equivalence of the AOS-based Model and the CHC Representation]
\label{theorem:aos-chc-equivalent}
  For any \(\varPi\) and simple \(f\) in \(\varPi\),
  \(f^\AOS_\varPi\) is equivalent to \(\cM_{\Parened{\varPi}}(f_\entry)\).
\end{theorem}
\begin{proof}
  Clear from completeness of SLDC resolution (\cref{lemma:sldc-complete}) and the bisimulation between abstract operational semantics and SLDC resolution (\cref{lemma:bisim-aos-chc}).
\qed\end{proof}

\subsection{Bisimulation between Concrete and Abstract Operational Semantics}

Extending `\(\safe_\cH(\cF\Colon \cGamma \mid \aF)\)' introduced in \cref{index:chc-correct}, we define the \emph{safe readout} `\(\safe_\varPi(\cC \mid \aC)\)' of an abstract configuration from a concrete configuration.
Interestingly, the safe readout is a \emph{bisimulation} between concrete and abstract operational semantics (\cref{lemma:bisim-cos-aos}).
We also establish \emph{progression and preservation} regarding the safe readout, as a part of soundness of COR's type system in terms of \emph{concrete operational semantics}, extending the soundness shown for abstract operational semantics in \cref{index:appx-proof-aos-safe}.

\Subsubsection{Auxiliary Notions}

An \emph{extended abstract variable summary} \(\hat\aX\) is a finite multiset of items of form `\(\give_\alpha(*a; \ab{x} \Colon\allowbreak T)\)' or `\(\take^\alpha(*a; \ab{x} \Colon T)\)', where \(a\) is an address.
An \emph{extended access mode} \(\hat D\) is of form either `\(\hot\)' or `\(\cold_\alpha\)'.
An \emph{extended memory footprint} \(\hat\aM\) is a finite multiset of items of form `\(\hot^\ac(a)\)' or `\(\cold_\alpha(a)\)', where \(a\) is an address.

\Subsubsection{Readout}

First, `\(\readout_{\cH, \hat D}^\ac(a\Colon T \mid \hat v;\, \hat\aX, \hat\aM)\)'
and `\(\readout_{\cH, \hat D}^\ac(*a\Colon T \mid \hat v;\, \hat\aX, \hat\aM)\)'
(the pointer of the address \(a\) / the data at \(a\), typed \(T\), can be read out from the heap \(\cH\) as a pre-value \(\hat v\), yielding an extended abstract variable summary \(\hat\aX\) and an extended memory footprint \(\hat\aM\), under the extended access mode \(\hat D\) and the activeness \(\ac\))
are defined by the following rules.
\begingroup\small
\begin{gather*}
  \frac{\begin{gathered}
    \readout_{\cH, \hat D \circ \check P}^\ac(*a\Colon T \mid \hat v;\, \hat\aX, \hat\aM)
  \end{gathered}}{
    \readout_{\cH, \hat D}^\ac\bigl(
      a\Colon \check P\, T \bigm| \angled{\hat v};\,
      \hat\aX, \hat\aM
    \bigr)
  } \br[-.1em]
  \dbox{\(
    \hat D \circ \own \defeq \hat D \quad
    \hot \circ \immut_\beta \defeq \cold_\beta \quad
    \cold_\alpha \circ \immut_\beta \defeq \cold_\alpha
  \)} \br[.1em]
  \frac{
    \readout_{\cH,\hot}^\ac(*a\Colon T \mid \hat v;\, \hat\aX, \hat\aM)
  }{
    \readout_{\cH,\hot}^\ac\bigl(
      a\Colon \mut_\beta T \bigm| \angled{\hat v, \ab{x}};\,
      \hat\aX \!\oplus\! \Braced{\give_\beta(*a; \ab{x}\Colon T)},\,
      \hat\aM
    \bigr)
  } \br[.1em]
  \frac{
    \readout_{\cH,\cold_\beta}^\ac(*a\Colon T \mid \hat v;\, \hat\aX, \hat\aM)
  }{
    \readout_{\cH,\cold_\beta}^\ac\bigl(
      a\Colon \mut_\beta T \bigm| \angled{\hat v, \hat w};\,
      \hat\aX, \hat\aM
    \bigr)
  } \br[.1em]
  \readout_{\cH, \hat D}^{\dagger\alpha}(*a\Colon T \mid \ab{x};\, \Braced{\take^\alpha(*a; \ab{x}\Colon T)}, \emp) \br[.0em]
  \frac{
    \cH(a) = a' \quad
    \readout_{\cH, \hat D}^\ac(a'\Colon P\, T \mid \hat v;\, \hat\aX, \hat\aM)
  }{
    \readout_{\cH, \hat D}^\ac(*a\Colon P\, T \mid \hat v;\, \hat\aX, \hat\aM \!\oplus\! \Braced{{\hat D}^\ac(a)})
  } \br[-.1em]
  \dbox{\(
    \hat D^\ac(a) \defeq \begin{cases}
      \hot^\ac(a) & (\hat D = \hot) \\[-.3em]
      \cold_\beta(a) & (\hat D = \cold_\beta)
    \end{cases}
  \)} \br[.1em]
  \frac{
    \readout_{\cH, \hat D}^\ac(*a\Colon T[\mu X.T/X] \mid \hat v;\, \hat\aX, \hat\aM)
  }{
    \readout_{\cH, \hat D}^\ac(*a\Colon \mu X.T \mid \hat v;\, \hat\aX, \hat\aM)
  } \br[.0em]
  \frac{
    \cH(a) = n
  }{
    \readout_{\cH, \hat D}^\ac(*a\Colon \Int \mid n;\, \emp, \Braced{{\hat D}^\ac(a)})
  } \quad
  \readout_{\cH, \hat D}^\ac(*a\Colon \unit \mid ();\, \emp,\emp) \br[.0em]
  \frac{\begin{gathered}
    \cH(a) = i \in [2] \quad
    \readout_{\cH, \hat D}^\ac(*(a \!+\! 1)\Colon T_i \mid \hat v;\, \hat\aX, \hat\aM) \quad
    n_0 = (\#T_{1 \!-\! i} \!-\! \#T_i)_{\ge 0} \\[-.4em]
    \text{for any}\ k \in [n_0],\
    \cH(a \!+\! 1 \!+\! \#T_i \!+\! k) = 0 \quad
    \hat\aM_0 = \Braced{{\hat D}^\ac(a \!+\! 1 \!+\! \#T_i \!+\! k) \mid k \in [n_0]}
  \end{gathered}}{\begin{aligned}
    & \readout_{\cH, \hat D}^\ac\bigl(*a\Colon T_0 \!+\! T_1 \bigm| \inj_i \hat v;\, \hat\aX, \hat\aM \!\oplus\! \Braced{{\hat D}^\ac(a)} \!\oplus\!\hat\aM_0\bigr)
  \end{aligned}} \br[.0em]
  \frac{
    \readout_{\cH, \hat D}^\ac\bigl(
      *a\Colon T_0 \bigm| \hat v_0;\, \hat\aX_0, \hat\aM_0
    \bigr) \quad
    \readout_{\cH, \hat D}^\ac\bigl(
      *(a+\#T_0)\Colon T_1 \bigm| \hat v_1;\, \hat\aX_1, \hat\aM_1
    \bigr)
  }{
    \readout_{\cH, \hat D}^\ac\bigl(
      *a\Colon T_0 \!\times\! T_1 \bigm| (\hat v_0, \hat v_1);\,
      \hat\aX_0 \!\oplus\!\hat\aX_1,\, \hat\aM_0 \!\oplus\!\hat\aM_1
    \bigr)
  }
\end{gather*}
\endgroup

Next, `\(\readout_{\cH,\cTheta}(\cF\Colon \cGamma \mid \aF;\, \hat\aX, \hat\aM)\)'
(the stack frame \(\cF\) respecting the variable context \(\cGamma\) can be read out from \(\cH\) as an abstract stack frame \(\aF\), yielding \(\hat\aX\) and \(\hat\aM\), under the lifetime parameter context \(\cTheta\))
is defined as follows.
\begingroup\small
\begin{gather*}
  \frac{
    \dom \cF = \dom \cGamma \quad
    \text{for any}\ x\colonu\ac T \in \cGamma,\
    \readout_{\cH,\hot}^\ac(\cF(x)\Colon T\cTheta \mid \hat v_x;\, \hat\aX_x, \hat\aM_x)
  }{
    \readout_{\cH,\cTheta}\bigl(
      \cF\Colon \cGamma \bigm| \{(x, \hat v_x) \mid x \in \dom \cGamma\};\,
      \bigoplus_{x \in \dom \cGamma} \hat\aX_x,\,
      \bigoplus_{x \in \dom \cGamma} \hat\aM_x
    \bigr)
  }
\end{gather*}
\endgroup

Finally, `\(\readout_\varPi(\cC \mid \aC;\, \hat\aX, \hat\aM)\)'
(the data of the concrete configuration \(\cC\) can be read out as the abstract configuration \(\aC\), yielding \(\hat\aX\) and \(\hat\aM\), under the program \(\varPi\))
is defined as follows.
\begingroup\small
\begin{gather*}
  \frac{
    \text{for any}\ i \in [n \!+\! 1],\ \,
    \readout_{\cH,\cTheta_i}(\cF_i\Colon \cGamma_{\varPi,f_i,L_i} \mid \aF_i;\, \hat\aX_i, \hat\aM_i)
  }{\begin{aligned}
    & \readout_\varPi\bigl(\,
      [f_0,L_0]\, \cF_0;\,
      [f_1,L_1]\, x_1,\cF_1;\, \cdots;\,
      [f_n,L_n]\, x_n,\cF_n \mid \cH \\[-.5em]
    & \textstyle
      \ \ \bigm|\
      [f_0,L_0]_{\cTheta_0}\, \aF_0;\,
      [f_1,L_1]_{\cTheta_1}\, x_1,\aF_1;\, \cdots;\,
      [f_n,L_n]_{\cTheta_n}\, x_n,\aF_n
      \mid_\cA;\
      \bigoplus_{i=0}^n \hat\aX_i, \bigoplus_{i=0}^n \hat\aM_i \,\bigr)
  \end{aligned}}
\end{gather*}
\endgroup

\Subsubsection{Safety}
We define the safety on extended abstract variable summaries and extended memory footprints.

`\(\safe_\cA(\ab{x}, \hat\aX)\)' is defined as follows.
\begingroup\small
\begin{gather*}
  \frac{
    \hat\aX(\ab{x}) = \Braced{\give_\alpha(*a; \ab{x}\Colon T),\, \take^\beta(*a; \ab{x}\Colon T')} \quad
    T \sim_\cA T' \quad
    \alpha \le_\cA \beta
  }{
    \safe_\cA(\ab{x}, \hat\aX)
  } \quad
  \frac{
    \hat\aX(\ab{x}) = \emp
  }{
    \safe_\cA(\ab{x}, \hat\aX)
  } \br[-.1em]
  \text{
    \(\hat\aX(\ab{x})\): the multiset of items of form `\(\give_\gamma(*b; \ab{x} \Colon U)\)'/`\(\take^\gamma(*b; \ab{x} \Colon U)\)' in \(\hat\aX\)
  }
\end{gather*}
\endgroup
`\(\safe_\cA(\hat\aX)\)' means that \(\safe_\cA(\ab{x}, \hat\aX)\) holds for any \(\ab{x}\).

`\(\safe_\cA(a, \hat\aM)\)' is defined as follows.
\begingroup\small
\begin{gather*}
  \frac{
    \hat\aM(a) = \{\hot^\ac(a)\}
  }{
    \safe_\cA(a, \hat\aM)
  } \quad
  \frac{
    \hat\aM(a) = \emp
  }{
    \safe_\cA(a, \hat\aM)
  } \br[.1em]
  \frac{
    \hat\aM(a) = \Braced{\hot^{\dagger\alpha}(a), \cold_{\beta_0}(a), \dots, \cold_{\beta_{n-1}}(a)} \quad
    \text{for any}\ i \in [n],\
    \beta_i \le_\cA \alpha
  }{
    \safe_\cA(a, \hat\aM)
  } \br[-.1em]
  \text{
    \(\hat\aM(a)\): the multiset of items of form \(\hot^\ac(a)\)/\(\cold_\alpha(a)\) in \(\hat\aM\)
  }
\end{gather*}
\endgroup
`\(\safe_\cA(\hat\aM)\)' means that \(\safe_\cA(a, \hat\aM)\) holds for any address \(a\).

\Subsubsection{Safe Readout}

Finally, `\(\safe_\varPi(\cC \mid \aC)\)'
(the data of the concrete configuration \(\cC\) can be \emph{safely} read out as the abstract configuration \(\aC\) under \(\varPi\))
is defined as follows.
\begingroup\small
\begin{gather*}
  \frac{
    \readout_\varPi(\cC \mid \aC;\, \hat\aX, \hat\aM) \quad
    \lifetimeSafe_\varPi(\aC) \quad
    \aC = {\cdots} \mid_\cA \quad
    \safe_\cA(\hat\aX) \quad
    \safe_\cA(\hat\aM)
  }{
    \safe_\varPi(\cC \mid \aC)
  }
\end{gather*}
\endgroup
`\(\safe_\varPi(\cC)\)' means that \(\safe_\varPi(\cC \mid \aC)\) holds for some \(\aC\).

\begin{property}[Safety on a Concrete Configuration Ensures Progression]
  For any \(\varPi\) and \(\cC\)
  such that \(\safe_\varPi(\cC)\) holds and \(\final_\varPi(\cC)\) does not hold,
  there exists some \(\cC'\) satisfying \(\cC \to_\varPi \cC'\).
\end{property}
\begin{proof}
  Clear.
  One important guarantee the safety provides is that the data is stored in the heap in an expected form.
\qed\end{proof}

\begin{lemma}[Safe Readout Ensures Safety on the Abstract Configuration]
\label{lemma:trans-safe}
  For \(\varPi\), \(\cC\) and \(\aC\) such that \(\safe_\varPi(\cC \mid \aC)\) holds, \(\safe_\varPi(\aC)\) holds.
\end{lemma}
\begin{proof}
  By straightforward induction over the judgment deduction.
  Note that safety on a \emph{extended} abstract variable summary is in fact an extension of safety on an abstract variable summary.
\qed\end{proof}

\Subsubsection{Bisimulation Lemma}

The safe readout defined above is actually a \emph{bisimulation} between concrete and abstract operational semantics.

\begin{lemma}[Bisimulation between Concrete and Abstract Operational Semantics]
\label{lemma:bisim-cos-aos}
  Take any \(\varPi\), \(\cC\) and \(\aC\) satisfying \(\safe_\varPi(\cC \mid \aC)\).

  For any \(\cC'\) satisfying \(\cC \to_\varPi \cC'\),
  there exists \(\aC'\) satisfying \(\aC \to_\varPi \aC'\) and \(\safe_\varPi(\cC' \mid \aC')\).
  Likewise, for any \(\aC'\) satisfying \(\aC \to_\varPi \aC'\) holds,
  there exists \(\cC'\) satisfying \(\cC \to_\varPi \cC'\) and \(\safe_\varPi(\cC' \mid \aC')\).
\end{lemma}
\begin{proof}
  How to take \(\aC'\) according to \(\cC'\) and vice versa can be decided in a straightforward way that we do not explicitly describe here.
  The property \(\safe_\varPi(\cC' \mid \aC')\) can be justified by the following observations.

  \Paragraph{No Unexpected Changes on Unrelated Data}
  The safety on the extended memory footprint ensures that operations on hotly accessed data do not affect unrelated data.
  Here, the following property plays a role: when \(\readout_{\cH,\hot}(a\Colon P\, T \mid \hat v;\, \hat\aX, \hat\aM)\) holds and
  \(P\) is of form \(\own\) or \(\mut_\alpha\),
  \(\Braced{\hot(a+k) \mid k \in [\#T]} \subseteq \hat\aM\) holds.

  \Paragraph{Preservation of the Safety on the Extended Abstract Variable Summary}
  It can be shown in a similar way to the proof of \cref{lemma:preserve-aos}.

  \Paragraph{Preservation of Safety on the Extended Memory Footprint}
  It can be shown by straightforward case analysis.

  One important point is that, on lifetime elimination (\(\now \alpha\)), a frozen hot access (\(\hot^{\dagger\alpha}(a)\)) can be safely made active (\(\hot^\Active(a)\)), because there are no cold accesses on \(a\), which is guaranteed by the type system.

  Another point is that swap (\(\swap(*x,*y)\)) does not change the extended memory footprint.
\qed\end{proof}

\begin{property}[Safety on the Concrete Configuration is Preserved]
  For any \(\varPi\) and \(\cC,\cC'\) such that \(\safe_\varPi(\cC)\) and \(\cC \to_\varPi \cC'\) hold, \(\safe_\varPi(\cC')\) is satisfied.
\end{property}
\begin{proof}
  It immediately follows by \cref{lemma:bisim-cos-aos}.
\qed\end{proof}

\subsection{Equivalence of the COS-based and AOS-based Models}

After introducing some easy lemmas,
we prove the equivalence of the COS-based and AOS-based models (\cref{theorem:cos-aos-equivalent}), relying on the bisimulation lemma \cref{lemma:bisim-cos-aos} proved above.
Finally, we achieve the complete proof of \cref{theorem:chc-correct}.

\begin{lemma}
\label{lemma:safe-cos-logic-aos}
  Take any \(\varPi\), simple \(f\) and \(L\).
  For any \(\cF\), \(\cH\) and \(\aF\), the following equivalence holds, if \(L = \entry\) or the statement at \(L\) is of form \(\return x\).
  \begin{gather*}
    \safe_\cH(\cF\Colon \cGamma_{\varPi,f,L} \mid \aF) \iff
    \safe_\varPi\bigl(\, [f,L]\, \cF \mid \cH \,\bigm|\, [f,L]_\emp\, \aF \mid_{(\!\emp,\emp\!)} \,\bigr)
  \end{gather*}
  (The \(\safe_\cH\) judgment is defined in \cref{index:chc-correct}.)
\end{lemma}
\begin{proof}
  By straightforward induction.
\qed\end{proof}

\begin{lemma}
\label{lemma:simple-unique}
  For any \(\varPi\) and \(\cC\) of form \([f,L]\, \cF \mid \cH\), when \(f\) is simple,
  there is at most one \(\aC\) satisfying \(\safe_\varPi(\cC \mid \aC)\).
\end{lemma}
\begin{proof}
  By straightforward induction.
  The simpleness of \(f\) has made the situation easy, because abstract variables do not occur in \(\aC\).
\qed\end{proof}

\begin{lemma}
\label{lemma:concrete-aos}
  For any \(\varPi\) and \(\aC\) of form \([f,L]_\emp\, \aF \,\rvert_{(\emp, \emp)}\), when \(f\) is simple and \(\aC\) is safe,
  there exists \(\cC\) satisfying \(\safe_\varPi(\cC \mid \aC)\).
\end{lemma}
\begin{proof}
  By straightforward construction.
\qed\end{proof}

\begin{theorem}[Equivalence of the COS-based Model and the AOS-based Model]
\label{theorem:cos-aos-equivalent}
  For any \(\varPi\) and simple \(f\),
  \(f^\COS_\varPi\) is equivalent to \(f^\AOS_\varPi\).
\end{theorem}
\begin{proof}
  Let us show that
  \begin{gather*}
    f^\COS_\varPi(v_0,\dots,v_{n-1},w) \iff f^\AOS_\varPi(v_0,\dots,v_{n-1},w)
  \end{gather*}
  holds for any values \(v_0,\dots,v_{n-1},w\) of the sorts \(\Parened{T_0}, \dots, \Parened{T_{n-1}}, \Parened{U}\),
  where \(\varSigma_{\varPi,f} = (x_0\colon T_0, \dots, x_{n-1}\colon T_{n-1}) \to U\).

  \Paragraph{\((\Longrightarrow)\)}
  By assumption, we can take concrete configurations \(\cC_0,\dots,\cC_N\) satisfying the following
  (for some \(L\), \(y\), \(\cF\), \(\cH\), \(\cF'\) and \(\cH'\)).
  \begin{gather*}
    \cC_0 \to_\varPi \cdots \to_\varPi \cC_N \quad
    \final_\varPi(\cC_N) \br[-.1em]
    \cC_0 = [f,\entry]\, \cF \mid \cH \quad
    \cC_N = [f,L]\, \cF' \mid \cH' \br[-.1em]
    \safe_\cH\bigl(\cF\Colon \cGamma_{\varPi,f,\entry} \bigm| \{(x_i,v_i) \!\mid\! i \!\in\! [n]\}\bigr) \quad
    \safe_{\cH'}\bigl(\cF'\Colon \cGamma_{\varPi,f,L} \bigm| \{(y,w)\}\bigr)
  \end{gather*}

  By \cref{lemma:safe-cos-logic-aos}, taking abstract configurations
  \begin{gather*}
    \aC_0 \defeq [f,\entry]_\emp\, \{(x_i,v_i) \!\mid\! i \!\in\! [n]\} \mid_{(\!\emp,\emp\!)} \quad
    \aC'_N \defeq [f,L]_\emp\, \{(y,w)\} \mid_{(\!\emp,\emp\!)},
  \end{gather*}
  we have \(\safe_\varPi(\cC_0 \mid \aC_0)\) and \(\safe_\varPi(\cC_N \mid \aC'_N)\).
  By \cref{lemma:trans-safe}, \(\safe_\varPi(\aC_0)\) also holds.
  By \cref{lemma:bisim-cos-aos}, we can take \(\aC_1,\dots,\aC_N\)
  satisfying \(\aC_0 \to_\varPi \cdots \to_\varPi \aC_N\), \(\final_\varPi(\aC_N)\),
  and \(\safe_\varPi(\cC_{k+1} \mid \aC_{k+1})\) (for any \(k \!\in\! [N]\)).

  Since \(\safe_\varPi(\cC_N \mid \aC_N)\) and \(\safe_\varPi(\cC_N \mid \aC'_N)\) hold,
  by \cref{lemma:simple-unique} we have \(\aC_N = \aC'_N\).
  Therefore, \(f^\AOS_\varPi(v_0,\dots,v_{n-1},w)\) holds.

  \Paragraph{\((\Longleftarrow)\)}
  By assumption, we can take abstract configurations \(\aC_0,\dots,\aC_N\) satisfying the following (for some \(L\) and \(y\)).
  \begin{gather*}
    \aC_0 \to_\varPi \cdots \to_\varPi \aC_N \quad
    \final_\varPi(\aC_N) \br[-.3em]
    \aC_0 = [f,\entry]_\emp\, \{(x_i,v_i) \!\mid\! i \!\in\! [n]\} \mid_{(\!\emp,\emp\!)} \quad
    \aC_N = [f,L]_\emp\, \{(y,w)\} \mid_{(\!\emp,\emp\!)}
  \end{gather*}

  By \cref{lemma:concrete-aos}, there exists \(\cC_0\) such that \(\safe_\varPi(\cC_0 \mid \aC_0)\) holds.
  By \cref{lemma:bisim-cos-aos}, we can take \(\cC_1,\dots,\cC_N\)
  satisfying \(\cC_0 \to_\varPi \cdots \to_\varPi \cC_N\), \(\final_\varPi(\cC_N)\), and \(\safe_\varPi(\cC_{k+1} \mid \aC_{k+1})\) (for any \(k \!\in\! [N]\)).

  \(\cC_0\) and \(\cC_N\) have form
  \begin{gather*}
    \cC_0 = [f,\entry]\, \cF \mid \cH \quad
    \cC_N = [f,L]\, \cF' \mid \cH',
  \end{gather*}
  and by \cref{lemma:safe-cos-logic-aos} the following judgments hold.
  \begin{gather*}
    \safe_{\cH}\bigl(\cF\Colon \cGamma_{\varPi,f,\entry} \bigm| \{(x_i,v_i) \!\mid\! i \!\in\! [n]\}\bigr) \quad
    \safe_{\cH'}\bigl(\cF'\Colon \cGamma_{\varPi,f,L} \bigm| \{(y,w)\}\bigr)
  \end{gather*}
  Therefore, \(f^\COS_\varPi(v_0,\dots,v_{n-1},w)\) holds.
\qed\end{proof}

Combining the equivalences of \cref{theorem:aos-chc-equivalent} and \cref{theorem:cos-aos-equivalent}, we finally achieve the proof of \cref{theorem:chc-correct}.

\fi

\end{document}